\documentclass[]{fairmeta}
\usepackage{enumitem}
\title{TextSeal: A Localized LLM Watermark \\ for Provenance \& Distillation Protection}

\newcommand{\cross}{\mathsection} %
\newcommand{\core}{\dagger}      %

\author[\star,\core]{Tom Sander}
\author[\core]{Hongyan Chang}
\author{Sylvestre-Alvise Rebuffi}
\author{Tom\'{a}\v{s} Sou\v{c}ek}
\author{Tuan Tran}
\author{Valeriu Lacatusu}
\author{Alexandre Mourachko}
\author[\cross]{Surya Parimi}
\author[\cross]{Christophe Ropers}
\author[\cross]{Rashel Moritz}
\author[\cross]{Vanessa Stark}
\author[\core]{Hady Elsahar}
\author[\star,\core]{Pierre Fernandez}
\affiliation{FAIR, Meta Superintelligence Labs}
\contribution[\star]{Equal contributors}
\contribution[\core]{Core team}
\contribution[\cross]{Project Support.}
 
\correspondence{\email{tomsander@meta.com}, \email{pfz@meta.com}}
\metadata[Code]{\url{https://github.com/facebookresearch/textseal}}

\abstract{
We introduce \emph{TextSeal}, a state-of-the-art watermark for large language models.
Building on Gumbel-max sampling, TextSeal introduces dual-key generation to restore output diversity, along with entropy-weighted scoring and multi-region localization for improved detection.
It supports serving optimizations such as speculative decoding and multi-token prediction, and does not add any inference overhead.
TextSeal strictly dominates baselines like SynthID-text in detection strength and is robust to dilution, maintaining confident localized detection even in heavily mixed human/AI documents.
The scheme is theoretically distortion-free, and evaluation across reasoning benchmarks confirms that it preserves downstream performance; while a multilingual human evaluation (6{,}000 A/B comparisons, 5 languages) shows no perceptible quality difference.
Beyond its use for provenance detection, TextSeal is also ``radioactive'': its watermark signal transfers through model distillation, enabling detection of unauthorized use. \vspace{-1em}
}

\usepackage[utf8]{inputenc}
\usepackage{amsmath, amssymb, amsthm}
\usepackage{algorithm}
\usepackage{algpseudocode}
\usepackage{pgfplots}
\usepackage{wrapfig}
\usepackage{float}
\usepackage{graphicx}

\usepackage{capt-of}     
\usepackage{subcaption}  
\usepackage{booktabs}    
\usepackage{pifont}

\usepackage{tikz}
\usetikzlibrary{decorations.pathreplacing, positioning, arrows.meta, shapes.geometric, fit, backgrounds, calc, patterns}

\newcommand{\figTextSealFull}{%
\scalebox{0.9}{%
\begin{tikzpicture}[
    >=Stealth,
    font=\small,
    every node/.style={font=\small},
    token/.style={rectangle, draw=black!70, rounded corners, minimum height=1.5em, minimum width=1.5em, fill=white, align=center, font=\small, inner sep=2pt},
    tokengray/.style={token, fill=gray!20},
    block/.style={rectangle, draw=blue!70, rounded corners, align=center, minimum height=1.6em, fill=blue!3, font=\small, inner sep=3pt},
    llmblock/.style={block, draw=orange!80, fill=orange!10},
    router/.style={diamond, draw=green!70!black, aspect=2, align=center, fill=green!5, inner sep=1pt, font=\small},
    container/.style={draw=gray!50, thick, dashed, inner sep=8pt, rounded corners, fill=gray!2}
]

\begin{scope}[shift={(0,0)}]

\node[token] (x1) at (0, 0) {$x_1$};
\node[token] (x2) at (0.7, 0) {$x_2$};
\node[font=\scriptsize] at (1.25, 0) {$\cdots$};
\node[tokengray] (xtk) at (1.85, 0) {$x_{t\text{-}k}$};
\node[font=\scriptsize] at (2.5, 0) {$\cdots$};
\node[tokengray] (xt1) at (3.1, 0) {$x_{t\text{-}1}$};

\draw[decorate, decoration={brace, amplitude=3pt, raise=2pt}]
    (xtk.north west) -- (xt1.north east)
    node[midway, above=6pt, font=\footnotesize] (wlabel) {$\mathbf{w}$};

\draw[decorate, decoration={brace, amplitude=3pt, mirror, raise=2pt}]
    (x1.south west) -- (xt1.south east)
    node[midway, below=6pt, font=\scriptsize] (llm_ctx) {LLM context};

\node[llmblock] (llm) at (0.5, -3.8) {LLM};
\node[below=0.05cm of llm, font=\scriptsize] (pvec) {${\vec{p}}^{(t)}$};
\draw[->, thick] (0.5, -0.85) -- (llm.north);

\node[router] (router) at (3.7, -1.9) {\scriptsize Route};

\draw[->, thick, rounded corners] (wlabel.east) -| (router.north);

\node[font=\scriptsize] (cand_label) at (5.5, -0.4) {Candidates $v \in \mathcal{V}$};
\node[token, fill=white, minimum width=0.6cm] (v1) at (4.7, -0.9) {$v_1$};
\node[token, fill=white, minimum width=0.6cm] (v2) at (5.4, -0.9) {$v_2$};
\node[font=\scriptsize] at (5.85, -0.9) {$..$};
\node[token, fill=white, minimum width=0.6cm] (vV) at (6.3, -0.9) {$v_V$};
\node[font=\scriptsize, gray] at (5.7, -2.0) {Each $v$};

\node[block, minimum width=1.5cm, fill opacity=0.7, text opacity=1] (k1) at (2.5, -3.1) {PRF$(k^{(1)})$};
\draw[->, thick] (router.south west) -- node[left, font=\scriptsize] {$1{-}\alpha$} (k1.north);
\draw[->, thick, gray, rounded corners] (vV.south) -- (6.3, -2.8) |- ([yshift=-2pt]k1.east);

\node[block, minimum width=1.5cm, fill opacity=0.7, text opacity=1] (k2) at (4.9, -3.1) {PRF$(k^{(2)})$};
\draw[->, thick] (router.south east) -- node[right, font=\scriptsize] {$\alpha$} (k2.north);
\draw[->, thick, gray, rounded corners] (vV.south) -- (6.3, -2.8) |- ([yshift=2pt]k2.east);

\node[block, fill=yellow!10, draw=yellow!60!black, minimum width=1.8cm] (selR) at (3.7, -4.2) {$\vec{R}^{(j)}$};
\draw[->, thick] (k1.south) |- (selR.west);
\draw[->, thick] (k2.south) |- (selR.east);

\node[block, fill=red!5, draw=red!70, minimum width=3.5cm, inner sep=4pt] (argmax) at (3.7, -5.3) {$x_t = \arg\max_v R_v^{(j),\; 1/p_v^{(t)}}$};
\draw[->, thick] (selR.south) -- (argmax.north);
\draw[->, thick, rounded corners] (pvec.south) -- (0.5, -5.3) -- (argmax.west);

\node[token, fill=purple!20, font=\footnotesize\bfseries] (xt) at (3.7, -6.4) {$x_t$};
\draw[->, thick] (argmax.south) -- (xt.north);
\node[font=\scriptsize, gray, right=6pt, align=left] at (xt.east) {$\Prob(x_t=i) = p_i$\\[0pt]{\scriptsize (distortion-free)}};

\end{scope}

\begin{scope}[shift={(10, 0)}]

\node[token] (d1) at (0, 0) {$x_1$};
\node[token] (d2) at (0.8, 0) {$x_2$};
\node[font=\scriptsize] at (1.4, 0) {$\cdots$};
\node[tokengray] (dtk) at (1.95, 0) {$x_{i\text{-}k}$};
\node[tokengray, font=\scriptsize, minimum width=0.6em] (ddots_w) at (2.6, 0) {$\cdots$};
\node[tokengray] (dt1) at (3.2, 0) {$x_{i\text{-}1}$};
\node[token, fill=purple!15] (di) at (3.95, 0) {$x_i$};
\node[font=\scriptsize] at (4.5, 0) {$\cdots$};
\node[token] (dn) at (5.0, 0) {$x_n$};

\draw[decorate, decoration={brace, amplitude=3pt, raise=2pt}]
    (d1.north west) -- (dt1.north east)
    node[midway, above=6pt, font=\scriptsize] (det_ctx) {LLM context};

\draw[decorate, decoration={brace, amplitude=3pt, mirror, raise=2pt}]
    (dtk.south west) -- (dt1.south east)
    node[midway, below=6pt, font=\scriptsize] (det_wlabel) {$\mathbf{w}_i$};

\node[block, minimum width=1.5cm] (dk1) at (0.6, -2.4) {PRF$(k^{(1)})$};
\node[block, minimum width=1.5cm] (dk2) at (2.4, -2.4) {PRF$(k^{(2)})$};
\node[block, draw=orange!60, fill=orange!6, minimum width=1.1cm, minimum height=1.3em, font=\scriptsize, inner sep=2pt] (llm_det) at (4.5, -2.4) {LLM\textsuperscript{\scalebox{0.6}{proxy}}};
\node[font=\tiny, gray, above right=2pt and -14pt of llm_det, align=left] (llm_det_note) {e.g., small,\\quantized};

\draw[thick, rounded corners] (det_wlabel.south) -- (det_wlabel.south |- 0,-1.1);
\draw[thick, rounded corners] (di.south) -- (di.south |- 0,-1.1);
\draw[thick] (0.6,-1.1) -- (di.south |- 0,-1.1);
\draw[->, thick] (0.6,-1.1) -- (dk1.north);
\draw[->, thick] (2.4,-1.1) -- (dk2.north);

\draw[->, thick, dashed, rounded corners] (det_ctx.east) -| (llm_det.north);

\node[block, draw=orange!80, fill=orange!10, minimum width=1.5cm, inner sep=2pt] (ent_box) at (4.5, -3.5) {$H_i \!\to\! w_i^{\text{ent}}$};
\draw[->, thick, dashed] (llm_det.south) -- (ent_box.north);

\node[block, minimum width=4.0cm, fill=blue!5, inner sep=3pt] (fusion) at (1.5, -3.9) {$s_i = (1{-}\alpha)\,s_i^{(1)} + \alpha\,s_i^{(2)}$};
\draw[->, thick] (dk1.south) -- node[left, font=\scriptsize] {$s_i^{(1)} \!=\! {\scriptstyle -\ln(1 - R_i^{(1)})}$} (0.6, -3.55);
\draw[->, thick] (dk2.south) -- node[right, font=\scriptsize] {$s_i^{(2)}$} (2.4, -3.55);

\node[block, minimum width=4.0cm, fill=yellow!10, draw=yellow!60!black, dashed, inner sep=3pt] (score) at (1.5, -4.9) {Reweight: $\tilde{s}_i = w_i^{\text{ent}} \cdot s_i$};
\draw[->, thick] (fusion.south) -- (score.north);
\draw[->, thick, dashed, rounded corners] (ent_box.south) -- (4.5, -4.9) -- (score.east);

\draw[->, thick] (score.south) -- (1.5, -5.5);

\begin{scope}[shift={(-0.6, -5.8)}]
    \def\barW{4.8}
    \def\barH{0.3}

    \fill[gray!8] (0, 0) rectangle (\barW, \barH);
    \draw[gray!40] (0, 0) rectangle (\barW, \barH);
    \foreach \x/\h in {
        0.0/0.15, 0.2/0.20, 0.4/0.25, 0.6/0.22, 0.8/0.18, 1.0/0.24, 1.2/0.20, 1.4/0.15,
        1.6/0.05, 1.8/0.04, 2.0/0.06, 2.2/0.03, 2.4/0.05, 2.6/0.08,
        2.8/0.12, 3.0/0.22, 3.2/0.25, 3.4/0.20,
        3.6/0.04, 3.8/0.05, 4.0/0.03, 4.2/0.06, 4.4/0.04, 4.6/0.02} {
        \fill[purple!40, opacity=0.8] (\x, 0) rectangle (\x+0.16, \h);
    }
    \node[font=\tiny, left] at (-0.05, \barH/2) {$\tilde{s}_i$};

    \def\wy{-0.25}  %
    \def\wym{-0.55} %
    \def\wyl{-0.85} %

    \node[font=\tiny, gray, left, align=right] at (-0.1, -0.55) {Dyadic\\candidate\\windows};

    \foreach \a in {0, 0.4, 0.8, 1.2, 1.6, 2.0, 2.4, 2.8, 3.2, 3.6, 4.0} {
        \draw[gray!50, thin, fill=white] (\a, \wy-0.08) rectangle (\a+0.8, \wy+0.08);
    }
    \foreach \a in {0, 0.8, 1.6, 2.4, 3.2} {
        \draw[gray!50, thin, fill=white] (\a, \wym-0.08) rectangle (\a+1.6, \wym+0.08);
    }
    \foreach \a in {0, 1.6} {
        \draw[gray!50, thin, fill=white] (\a, \wyl-0.08) rectangle (\a+3.2, \wyl+0.08);
    }

    \fill[green!20, draw=green!70!black, thick] (0, \wym-0.08) rectangle (1.6, \wym+0.08);
    \fill[green!20, draw=green!70!black, thick] (2.8, \wy-0.08) rectangle (3.6, \wy+0.08);

    \node[font=\tiny, gray, right] at (\barW+0.05, \wy) {$L_0$};
    \node[font=\tiny, gray, right] at (\barW+0.05, \wym) {$2L_0$};
    \node[font=\tiny, gray, right] at (\barW+0.05, \wyl) {$4L_0$};

    \def\outy{-1.6}
    \fill[gray!15] (0, \outy) rectangle (\barW, \outy+\barH);
    
    \fill[green!40] (0, \outy) rectangle (1.6, \outy+\barH);
    \fill[green!40] (2.8, \outy) rectangle (3.6, \outy+\barH);
    
    \draw[thick] (0, \outy) rectangle (\barW, \outy+\barH);
    \draw[thick] (1.6, \outy) -- (1.6, \outy+\barH);
    \draw[thick] (2.8, \outy) -- (2.8, \outy+\barH);
    \draw[thick] (3.6, \outy) -- (3.6, \outy+\barH);

    \node[font=\tiny] at (-0.2, \outy+0.15) {$x_1$};
    \node[font=\tiny] at (\barW+0.25, \outy+0.15) {$x_n$};

    \draw[->, green!60!black, thick] (0.8, \wym-0.08) -- (0.8, \outy+\barH);
    \draw[->, green!60!black, thick] (3.2, \wy-0.08) -- (3.2, \outy+\barH);

    \node[font=\fontsize{4}{5}\selectfont, green!70!black] at (0.8, \wym) {$p_{\text{raw}}$};
    \node[font=\fontsize{4}{5}\selectfont, green!70!black] at (3.2, \wy) {$p_{\text{raw}}$};

    \draw[decorate, decoration={brace, amplitude=2pt, mirror, raise=1pt}] (0, \outy) -- (1.6, \outy);
    \node[font=\tiny, green!50!black, align=center] at (0.8, \outy-0.35) {Region 1\\$p < 10^{-6}$};
    \draw[decorate, decoration={brace, amplitude=2pt, mirror, raise=1pt}] (2.8, \outy) -- (3.6, \outy);
    \node[font=\tiny, green!50!black, align=center] at (3.2, \outy-0.35) {Region 2\\$p < 10^{-4}$};
\end{scope}

\coordinate (regbar_bot) at (-1.8, -8.1);
\coordinate (regbar_right) at (5.5, -7.8);

\end{scope}

\coordinate (emb_top) at (0, 0.9);
\coordinate (det_top) at (8.5, 0.9);

\begin{scope}[on background layer]
    \node[container, fit={(emb_top) (x1) (xt1) (xt) (argmax) (k2) (vV) (cand_label) (router) (wlabel)}, inner sep=10pt] (emb_box) {};
    \node[container, fit={(det_top) (d1) (dn) (dk1) (llm_det) (llm_det_note) (ent_box) (fusion) (score) (regbar_bot) (regbar_right) (det_wlabel)}, inner sep=10pt] (det_box) {};
\end{scope}

\node[above, font=\footnotesize\bfseries] at (emb_box.north) {Embedding};
\node[above, font=\footnotesize\bfseries] at (det_box.north) {Detection};

\end{tikzpicture}%
}%
}

\newcommand{\figDiversityMechanisms}{%
\resizebox{\textwidth}{!}{%
\begin{tikzpicture}[
    box/.style={draw, rounded corners=2pt, minimum width=1.4cm, minimum height=0.45cm, align=center, font=\tiny},
    arrow/.style={-{Stealth[length=1.5mm]}},
    title/.style={font=\small\bfseries, align=center},
    note/.style={font=\tiny, align=center},
    colsep/.style={draw, dashed, gray!50}
]

\draw[colsep] (3.6, 4.3) -- (3.6, -2.1);
\draw[colsep] (7.4, 4.3) -- (7.4, -2.1);
\draw[colsep] (11.2, 4.3) -- (11.2, -2.1);
\draw[colsep] (15.0, 4.3) -- (15.0, -2.1);
\draw[colsep] (18.8, 4.3) -- (18.8, -2.1);
\draw[colsep] (22.6, 4.3) -- (22.6, -2.1);

\begin{scope}[shift={(0, 0)}]
    \node[title] at (1.7, 3.6) {Standard\\[-2pt]Gumbel};
    \node[box, fill=orange!20] (prf1) at (1.7, 2.6) {PRF$(\mathbf{w}, K)$};
    \node[box, fill=gray!15] (r1) at (1.7, 1.6) {$\vec{R}$};
    \node[box, fill=green!15] (sel1) at (1.7, 0.6) {$\arg\max R_v^{1/p_v}$};
    \node[box, fill=purple!20] (out1) at (1.7, -0.6) {$x_t$};
    
    \draw[arrow] (prf1.south) -- (r1.north);
    \draw[arrow] (r1.south) -- (sel1.north);
    \draw[arrow] (sel1.south) -- (out1.north);
    
    \node[note, red!70!black] at (1.7, -1.4) {Deterministic};
    \node[note, gray] at (1.7, -1.8) {$\checkmark$ Distortion-free};
\end{scope}

\begin{scope}[shift={(3.8, 0)}]
    \node[title] at (1.7, 3.6) {Stochastic\\[-2pt]Mixing};
    \node[box, fill=orange!20] (prf2) at (0.7, 2.6) {PRF};
    \node[box, fill=red!20] (rnd2) at (2.7, 2.6) {$r_0$};
    \node[box, fill=yellow!20] (mix2) at (1.7, 1.6) {Mix $r_1, r_0$};
    \node[box, fill=green!15] (sel2) at (1.7, 0.6) {$\arg\max r_v^{1/p_v}$};
    \node[box, fill=purple!20] (out2) at (1.7, -0.6) {$x_t$};
    
    \draw[rounded corners=2pt] (prf2.south) |- (1.7, 2.1);
    \draw[rounded corners=2pt] (rnd2.south) |- (1.7, 2.1);
    \draw[arrow] (1.7, 2.1) -- (mix2.north);
    
    \draw[arrow] (mix2.south) -- (sel2.north);
    \draw[arrow] (sel2.south) -- (out2.north);
    
    \node[note, green!50!black] at (1.7, -1.4) {Stochastic};
    \node[note, gray] at (1.7, -1.8) {$\checkmark$ Distortion-free};
\end{scope}

\begin{scope}[shift={(7.6, 0)}]
    \node[title] at (1.7, 3.6) {Entropy\\[-2pt]Warmup};
    \node[box, fill=blue!15] (phase) at (1.7, 2.6) {$\sum H_i > \tau_s$?};
    \node[box, fill=red!15] (nowm) at (0.7, 1.6) {Sample $\vec{p}$};
    \node[box, fill=orange!20] (yeswm) at (2.7, 1.6) {Gumbel};
    \node[box, fill=purple!20] (out3) at (1.7, -0.6) {$x_t$};
    
    \draw[arrow, rounded corners=2pt] (phase.south) -- (1.7, 2.1) -| (nowm.north);
    \node[font=\tiny] at (1.1, 2.25) {no};
    \draw[arrow, rounded corners=2pt] (phase.south) -- (1.7, 2.1) -| (yeswm.north);
    \node[font=\tiny] at (2.3, 2.25) {yes};
    
    \draw[rounded corners=2pt] (nowm.south) |- (1.7, 0.0);
    \draw[rounded corners=2pt] (yeswm.south) |- (1.7, 0.0);
    \draw[arrow] (1.7, 0.0) -- (out3.north);
    
    \node[note, green!50!black] at (1.7, -1.4) {Stochastic (prefix)};
    \node[note, gray] at (1.7, -1.8) {$\checkmark$ Distortion-free};
\end{scope}

\begin{scope}[shift={(11.4, 0)}]
    \node[title] at (1.7, 3.6) {Random\\[-2pt]Skip};
    \node[box, fill=red!20] (coin4) at (1.7, 2.6) {Coin flip $\alpha$};
    \node[box, fill=red!15] (skip4) at (0.7, 1.6) {Sample $\vec{p}$};
    \node[box, fill=orange!20] (wm4) at (2.7, 1.6) {Gumbel};
    \node[box, fill=purple!20] (out4) at (1.7, -0.6) {$x_t$};
    
    \draw[arrow, rounded corners=2pt] (coin4.south) -- (1.7, 2.1) -| (skip4.north);
    \node[font=\tiny] at (1.1, 2.25) {skip};
    \draw[arrow, rounded corners=2pt] (coin4.south) -- (1.7, 2.1) -| (wm4.north);
    \node[font=\tiny] at (2.3, 2.25) {keep};
    
    \draw[rounded corners=2pt] (skip4.south) |- (1.7, 0.0);
    \draw[rounded corners=2pt] (wm4.south) |- (1.7, 0.0);
    \draw[arrow] (1.7, 0.0) -- (out4.north);
    
    \node[note, green!50!black] at (1.7, -1.4) {Stochastic};
    \node[note, gray] at (1.7, -1.8) {$\checkmark$ Distortion-free};
\end{scope}

\begin{scope}[shift={(15.2, 0)}]
    \node[title] at (1.7, 3.6) {Adaptive\\[-2pt]Skip};
    \node[box, fill=orange!20] (prf5) at (1.7, 2.6) {Gumbel};
    \node[box, fill=blue!15] (test5) at (1.7, 1.6) {$R_{V^\star}\!<\!\tau$?};
    \node[box, fill=red!15] (skip5) at (0.7, 0.6) {Sample $\vec{p}$};
    \node[box, fill=green!15] (keep5) at (2.7, 0.6) {Keep $V^\star$};
    \node[box, fill=purple!20] (out5) at (1.7, -0.6) {$x_t$};
    
    \draw[arrow] (prf5.south) -- (test5.north);
    
    \draw[arrow, rounded corners=2pt] (test5.south) -- (1.7, 1.1) -| (skip5.north);
    \node[font=\tiny] at (1.1, 1.25) {yes};
    \draw[arrow, rounded corners=2pt] (test5.south) -- (1.7, 1.1) -| (keep5.north);
    \node[font=\tiny] at (2.3, 1.25) {no};
    
    \draw[rounded corners=2pt] (skip5.south) |- (1.7, -0.1);
    \draw[rounded corners=2pt] (keep5.south) |- (1.7, -0.1);
    \draw[arrow] (1.7, -0.1) -- (out5.north);
    
    \node[note, green!50!black] at (1.7, -1.4) {Stochastic};
    \node[note, gray] at (1.7, -1.8) {$\times$ Not distortion-free};
\end{scope}

\begin{scope}[shift={(19.0, 0)}]
    \node[title] at (1.7, 3.6) {Ent-Norm\\[-2pt]Skip};
    \node[box, fill=orange!20] (prf6) at (1.7, 2.6) {Gumbel};
    \node[box, fill=blue!15] (test6) at (1.7, 1.6) {$R_{V^\star}\!<\!\tau^{p_{V^\star}}$?};
    \node[box, fill=red!15] (skip6) at (0.7, 0.6) {Sample $\vec{p}$};
    \node[box, fill=green!15] (keep6) at (2.7, 0.6) {Keep $V^\star$};
    \node[box, fill=purple!20] (out6) at (1.7, -0.6) {$x_t$};
    
    \draw[arrow] (prf6.south) -- (test6.north);
    
    \draw[arrow, rounded corners=2pt] (test6.south) -- (1.7, 1.1) -| (skip6.north);
    \node[font=\tiny] at (1.1, 1.25) {yes};
    \draw[arrow, rounded corners=2pt] (test6.south) -- (1.7, 1.1) -| (keep6.north);
    \node[font=\tiny] at (2.3, 1.25) {no};
    
    \draw[rounded corners=2pt] (skip6.south) |- (1.7, -0.1);
    \draw[rounded corners=2pt] (keep6.south) |- (1.7, -0.1);
    \draw[arrow] (1.7, -0.1) -- (out6.north);
    
    \node[note, green!50!black] at (1.7, -1.4) {Stochastic};
    \node[note, gray] at (1.7, -1.8) {$\checkmark$ Distortion-free};
\end{scope}

\begin{scope}[shift={(22.8, 0)}]
    \node[title] at (1.7, 3.6) {Dual-Key\\[-2pt]Routing};
    \node[box, fill=red!20] (route7) at (1.7, 2.6) {Route $\alpha$};
    \node[box, fill=orange!20] (kd7) at (0.5, 1.6) {PRF$(k^{(1)})$};
    \node[box, fill=orange!20] (kt7) at (2.9, 1.6) {PRF$(k^{(2)})$};
    \node[box, fill=green!15] (gum7) at (1.7, 0.6) {$\arg\max R_v^{1/p_v}$};
    \node[box, fill=purple!20] (out7) at (1.7, -0.6) {$x_t$};

    \draw[arrow] (route7.south) -- node[left, font=\tiny] {$1\!-\!\alpha$} (kd7.north);
    \draw[arrow] (route7.south) -- node[right, font=\tiny] {$\alpha$} (kt7.north);
    \draw[arrow] (kd7.south) -- (0.5, 0.95) -- (1.7, 0.95) -- (gum7.north);
    \draw[arrow] (kt7.south) -- (2.9, 0.95) -- (1.7, 0.95);
    \draw[arrow] (gum7.south) -- (out7.north);

    \node[note, green!50!black] at (1.7, -1.4) {Stochastic};
    \node[note, gray] at (1.7, -1.8) {$\checkmark$ Distortion-free};
\end{scope}

\end{tikzpicture}%
}%
}

\newcommand{\figGumbelMaxTikz}{%
\begin{tikzpicture}[
    node distance=0.8cm,
    box/.style={draw, rounded corners, minimum width=2cm, minimum height=0.7cm, align=center},
    tokenbox/.style={draw, rounded corners, minimum width=0.8cm, minimum height=0.6cm, align=center, font=\small},
    arrow/.style={-{Stealth[length=2mm]}, thick},
    brace/.style={decorate, decoration={brace, amplitude=5pt, raise=2pt}}
]

\node[font=\small] at (-1.8, 4.2) {Generated:};
\node[tokenbox, fill=gray!10] (t1) at (0, 4.2) {$x_1$};
\node[tokenbox, fill=gray!10] (t2) at (1, 4.2) {$x_2$};
\node[font=\small] at (1.8, 4.2) {$\cdots$};
\node[tokenbox, fill=gray!30] (tk1) at (2.8, 4.2) {$x_{t-k}$};
\node[font=\small] at (3.6, 4.2) {$\cdots$};
\node[tokenbox, fill=gray!30] (tk) at (4.4, 4.2) {$x_{t-1}$};

\draw[decorate, decoration={brace, amplitude=4pt, raise=3pt}] (2.4, 4.55) -- (4.8, 4.55);
\node[font=\scriptsize] at (3.6, 5.2) {$\mathbf{w} = (x_{t-k}, \ldots, x_{t-1})$};

\draw[decorate, decoration={brace, amplitude=4pt, mirror, raise=3pt}] (-0.4, 3.85) -- (4.8, 3.85);
\node[font=\scriptsize] at (2.2, 3.35) {LLM context (all tokens)};

\node[box, fill=blue!10] (llm) at (1.9, 2.2) {LLM};
\node[below=0.2cm of llm, font=\small] (prob) {$\vec{p}^{(t)} = (p_1, \ldots, p_V)$};

\draw[arrow] (1.9, 3.2) -- (llm);

\node[box, fill=orange!20, minimum width=3cm] (prf) at (8, 2.2) {PRF$(\mathbf{w}, K)$};
\node[below=0.2cm of prf, font=\small] (rand) {$\vec{R} = (R_1, \ldots, R_V)$};
\node[below=0.1cm of rand, font=\scriptsize, gray] {$R_v = \text{PRF}(v, \mathbf{w}, K)$};

\draw[arrow] (4.8, 4.8) -- (8, 4.8) -- (8, 2.55);
\node[font=\scriptsize] at (5.4, 5.1) {$\mathbf{w}$};

\node[font=\scriptsize] at (11, 4.2) {Candidates $v \in \mathcal{V}$:};
\node[tokenbox, fill=white, minimum width=0.6cm] at (10.2, 3.5) {$v_1$};
\node[tokenbox, fill=white, minimum width=0.6cm] at (10.9, 3.5) {$v_2$};
\node[font=\scriptsize] at (11.5, 3.5) {$\cdots$};
\node[tokenbox, fill=white, minimum width=0.6cm] at (12.1, 3.5) {$v_V$};

\draw[arrow] (11, 3.1) -- (11, 2.2) -- (prf.east);
\node[font=\scriptsize] at (11.7, 2.5) {each $v$};

\node[box, fill=green!15, minimum width=2.5cm] (select) at (5, 0.2) {$x_t = \arg\max_v R_v^{1/p_v}$};

\draw[arrow] (prob.south) -- (1.9, 0.2) -- (select.west);
\draw[arrow] (rand.south) -- (8, 0.2) -- (select.east);

\node[tokenbox, fill=purple!20, minimum width=1cm] (output) at (5, -1.2) {$x_t$};
\draw[arrow] (select) -- (output);

\begin{scope}[shift={(0, -3.2)}]
    \node[font=\small] at (1.5, 0.5) {$R_v^{1/p_v}$:};
    
    \fill[blue!30] (3, 0) rectangle (3.4, 0.6);
    \fill[green!50] (3.6, 0) rectangle (4.0, 1.2);
    \fill[blue!30] (4.2, 0) rectangle (4.6, 0.4);
    \fill[blue!30] (4.8, 0) rectangle (5.2, 0.8);
    \fill[blue!30] (5.4, 0) rectangle (5.8, 0.5);
    
    \node[font=\tiny] at (3.2, -0.2) {$v_1$};
    \node[font=\tiny] at (3.8, -0.2) {$v_2$};
    \node[font=\tiny] at (4.4, -0.2) {$v_3$};
    \node[font=\tiny] at (5.0, -0.2) {$v_4$};
    \node[font=\tiny] at (5.6, -0.2) {$v_5$};
    
    \draw[-{Stealth}, red, thick] (3.8, 1.5) -- (3.8, 1.3);
    \node[font=\tiny, red] at (3.8, 1.7) {max};
    
    \node[font=\scriptsize, align=left, gray] at (8.5, 0.5) {Selected token $v_2$\\has highest $R_v^{1/p_v}$};
\end{scope}

\end{tikzpicture}%
}

\newcommand{\Autoref}[1]{%
  \begingroup%
  \def\chapterautorefname{Chapter}%
  \def\sectionautorefname{Section}%
  \def\subsectionautorefname{Subsection}%
  \autoref{#1}%
  \endgroup%
}

\newtheorem{proposition}{Proposition}
\newtheorem{corollary}{Corollary}

\theoremstyle{definition}
\newtheorem{definition}{Definition}
\newtheorem{remark}{Remark}

\renewcommand{\vec}[1]{\boldsymbol{#1}}

\newcommand{\E}{\mathbb{E}}
\newcommand{\Prob}{\mathbb{P}}
\newcommand{\V}{\mathcal{V}}
\newcommand{\Hnull}{\mathcal{H}_0}
\newcommand{\Halt}{\mathcal{H}_1}
\newcommand{\sk}{K} %

\newcommand{\cmark}{\ding{51}}
\newcommand{\xmark}{\ding{55}}

\newcommand{\eg}{e.g.,\@ }
\newcommand{\ie}{i.e.,\@ }

\begin{document}

\maketitle

\begin{figure}[b!]
  \centering
  \begin{minipage}{\textwidth}
    \centering
    \begin{subfigure}[b]{0.35\textwidth}
      \vspace{0pt} 
      \centering
      \includegraphics[width=\textwidth]{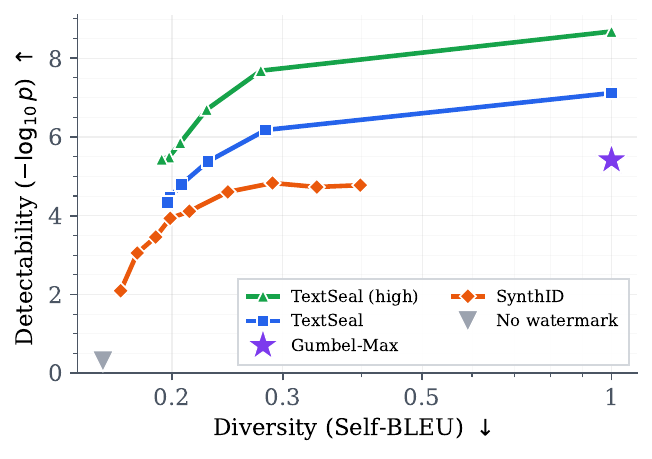}
      \caption{Diversity--detectability trade-off.}
      \label{fig:diversity-detectability}
    \end{subfigure}
    \hfill
    \begin{subfigure}[b]{0.35\textwidth}
      \vspace{0pt} 
      \centering
      \includegraphics[width=\textwidth]{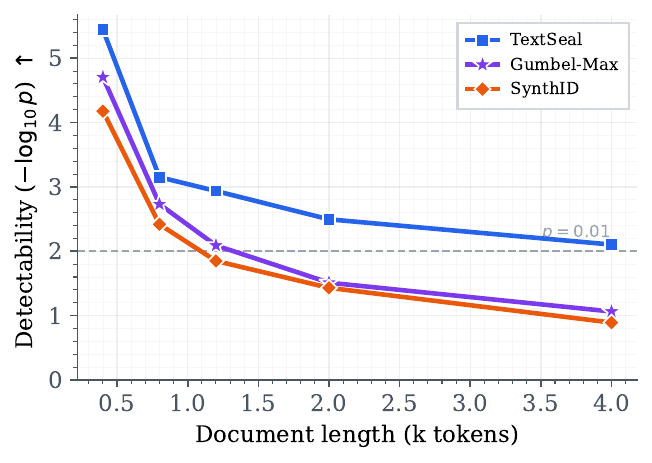}
      \caption{Detectability under dilution.} 
      \label{fig:dilution-catch}
    \end{subfigure}
    \hfill
\begin{subfigure}[b]{0.23\textwidth}
      \vspace{0pt} 
      \centering
      \scriptsize
      \setlength{\tabcolsep}{4pt}
      \resizebox{\textwidth}{!}{
        \begin{tabular}[t]{l cc} 
        \toprule
        \textbf{Task} & \textbf{None} & \textbf{TextSeal} \\
        \midrule
        AIME & 40.1 & 41.1 \\
        MATH & 79.8 & 79.8 \\
        GSM8K & 95.4 & 96.0 \\
        HumanEval & 97.0 & 93.3 \\
        MBPP & 50.2 & 49.2 \\
        ARC-C & 88.3 & 88.5 \\
        ARC-E & 93.4 & 93.7 \\
        GPQA & 50.5 & 50.0 \\
        HellaSwag & 94.7 & 94.8 \\
        MMLU & 49.2 & 51.5 \\
        SQA & 15.8 & 16.0 \\
        WinoGrande & 93.2 & 93.5 \\
        \midrule
        \textbf{Average} & \textbf{70.6} & \textbf{70.6} \\
        \bottomrule
        \end{tabular}
      }
      \caption{Performance.}
      \label{fig:performance-tab}
    \end{subfigure}
  \end{minipage}
  
  {
  \caption{\textbf{TextSeal achieves state-of-the-art detectability while preserving generation diversity and downstream performance} (Qwen3.5-27B).
  \textbf{(a)}~TextSeal strictly dominates SynthID across the diversity-detectability frontier (ELI5, 400 tokens, $T{=}0.8$, top-$p{=}0.9$).
  \textbf{(b)}~Localized detection remains confident even at $10\times$ dilution, where global baselines fail.
  \textbf{(c)}~Accuracy across 12 benchmarks is preserved ($T{=}0.6$).}
  \label{fig:main-results}
  \par
  }
\end{figure}

\section{Introduction}

The rapid adoption of LLMs in production systems has created a need for reliable provenance mechanisms.
Watermarking, by embedding an imperceptible, algorithmically-detectable signal during the generation, addresses several needs at once: detecting AI-generated content, complying with regulations that mandate machine-detectable marking of AI outputs~\citep{EuropeanAIAct, EUCodeOfPractice2026}, and enabling applications such as monitoring model output usage, preventing self-training on generated data, and detecting unauthorized distillation~\citep{sander2024watermarking, sablayrolles2020radioactive}.

For production deployment, it is highly desirable to use \emph{distortion-free} watermarking, which ensures that next-token selection follows exactly the same distribution as that produced by the LLM.
It preserves the exact decoding configuration (temperature, top-$p$) the model was tuned for, embedding the watermark at zero cost to any individual generation's quality.
A recent large-scale comparison~\citep{fernandez2025good} shows that the Gumbel-max watermark~\citep{aaronson2023watermarking} achieves the best detectability-quality Pareto frontier by far among other methods, e.g., green-red list~\citep{kirchenbauer2023watermark}, SynthID~\citep{dathathri2024scalable}, DiPMark~\citep{wu2023dipmark}.
However, Gumbel-max has one important drawback: it is fully \emph{deterministic} (a fixed prompt and secret key always produce the same output, eliminating diversity), which can in turn trigger degenerate loops when repeated n-grams used for hashing cause the pseudo-random function to lock onto the same token (\autoref{rem:deduplication}).
For instance, SynthID (deployed in Google's Gemini) resolves the determinism while remaining distortion-free, with a tournament-sampling design.

\emph{TextSeal} is a distortion-free, non-deterministic watermark for LLMs.
It builds upon the Gumbel-max framework and introduces three core improvements:
\begin{enumerate}
    \item \textbf{Dual-Key Generation:} We overcome determinism by randomly alternating between two secret keys during generation, restoring diversity at low cost to detection power.
    This natively supports speculative decoding and Multi-Token Prediction (MTP) without additional latency (\autoref{subsec:dual_key}).
    \item \textbf{Entropy-Weighted Detection:} We introduce tests tailored to the dual-key generation, that may leverage the entropy of a proxy model, and moment-matched Gamma approximations to have calibrated $p$-values (\autoref{subsec:entropy_detection}).
    \item \textbf{Localized Detection:} We identify individual watermarked segments within a document via a multi-region geometric cover search, dramatically boosting detection under dilution (\autoref{subsec:localization}).
\end{enumerate}

As summarized in \autoref{fig:main-results}, TextSeal achieves state-of-the-art detectability while offering a superior diversity-detectability trade-off compared to existing methods.
Our localized detection is robust to dilution within long documents, and TextSeal preserves the downstream capabilities of the model across 12 complex benchmarks.
TextSeal adds only ${\le}0.3\%$ sampling overhead ($3\times$ faster than SynthID; \autoref{sec:exp3}).
Beyond provenance, TextSeal is \emph{radioactive}~\citep{sander2024watermarking, sablayrolles2020radioactive}: the watermark signal transfers through model distillation, meaning that a student model trained on watermarked outputs inherits a detectable trace.
This provides a practical safeguard against unauthorized distillation and enables monitoring of how model outputs are used downstream (in training pipelines, RAG systems, or by competitors).
We demonstrate this experimentally in \autoref{sec:learnability}.

The paper is organized as follows.
\Autoref{sec:background} presents the technical background .
\Autoref{sec:method} describes the TextSeal method.
\Autoref{sec:experiments} presents the main experimental results.
\Autoref{sec:ablations} provides ablation studies and additional analyses.
\Autoref{sec:learnability} demonstrates watermark transfer through distillation.

\section{Background and Related Work}
\label{sec:background}

\subsection{LLM Watermarking}
Early text watermarking relied on edit-based methods~\citep{topkara2005natural, topkara2006hiding} with low robustness.
For LLMs, two concurrent approaches appeared after ChatGPT: green-red list~\citep{kirchenbauer2023watermark} and Gumbel-max sampling~\citep{aaronson2023watermarking}, both using pseudorandom seeds from a secret key and preceding tokens, enabling lightweight detection without access to the model.
Some subsequent work explores multi-bit watermarking~\citep{fernandez2023three, yoo2024advancing, qu2024provably}, undetectable constructions~\citep{christ2023undetectable, kuditipudi2023robust}, low-entropy optimizations~\citep{lee2023wrote, huang2023optimal}, adaptive green-red variants~\citep{wang2025morphmark}, distillation for open-weights model~\citep{gu2023learnability}, etc.
See \autoref{app:method_overview} for detailed scheme descriptions.
Semantic watermarks~\citep{liu2023semantic, liu2024adaptive, hou2023semstamp} offer increased robustness to adversaries, but require auxiliary semantic encoders.
This makes them harder to deploy, which is the reason why we do not consider them in the remaining of the work.

Beyond detection, watermark radioactivity~\citep{sander2024watermarking} has been leveraged for data protection (RAG~\citep{jovanovic2025ward}, contamination~\citep{sander2025detecting}, copyright~\citep{zhang2025leave}), which we extend in \autoref{sec:learnability} to reasoning-trace distillation.

\subsection{Distortion-Freeness and Choice of Baselines}
In the literature on LLM watermarking, schemes are typically divided into two families: \textit{distortionary} (biased) and \textit{distortion-free} (unbiased/distribution-preserving). The key distinction is whether the watermark alters text quality.
A watermarking sampler is \emph{distortion-free} (or \emph{non-distortionary}) if, in expectation over the random seed, it outputs each token with exactly its original LLM probability: no token is favored or suppressed by the watermark.
This corresponds to the \emph{single-token non-distortion} property of \citet{dathathri2024scalable}; formal definitions and the stronger \emph{single-sequence} variant are given in \autoref{app:distortion_definitions}.

For instance, green-red list~\citep{kirchenbauer2023watermark} and low-entropy filtering methods, like SWEET \citep{lee2023wrote} which skips watermarking on low-entropy tokens, are \emph{distortionary}: they shift the output distribution, degrading generation.
MorphMark~\citep{wang2025morphmark} adaptively scales the green-red bias based on the natural green-list probability mass, reducing distortion in low-entropy contexts, but remains distortionary since it still applies a logit bias.
Gumbel-max~\citep{aaronson2023watermarking}, Permute-and-Flip~\citep{zhao2024permute}, DiPMark~\citep{wu2023dipmark} (non-distortionary green-red via pseudorandom permutations), SynthID-Text~\citep{dathathri2024scalable} (deployed in Google Gemini and detailed in \autoref{app:synthid_details}), and WaterMax~\citep{giboulot2024watermax} (multiple generations per query, impractical for production) are non-distortionary methods.

Aligned with recent large-scale evaluations~\citep{fernandez2025good}, we found that Gumbel-max and SynthID achieved the best detectability-quality Pareto frontier. 
Therefore, we build TextSeal upon Gumbel-Max (which we detail next), and compare against these two baselines.
Because all three are single-token non-distortionary, we can fix the LLM, temperature, and top-$p$, and vary only their watermark-specific diversity parameter, isolating the watermark's effect.

\subsection{Gumbel-Max Watermarking}
\label{subsec:gumbel_mechanism}

We consider a language model generating a sequence of tokens. At each time step $t$, the model predicts a probability distribution $\vec{p}^{(t)} = (p_1, \dots, p_{|\V|})$ over the vocabulary $\V$. Let $\sk$ be a secret key used for watermarking and $h_t$ the context (history of tokens) at step $t$. The goal of watermarking is to select a token $x_t$ such that its selection is statistically correlated with a pseudo-random value derived from $h_t$ and $\sk$, while preserving the original distribution $\vec{p}$ (single-token non-distortion, \autoref{def:single_token_distortion_free}). This concept was introduced for LLMs by \citet{aaronson2023watermarking} with the Gumbel-max scheme.

\subsubsection{Gumbel-max mechanism}

The standard Gumbel watermarking ensures detectability by making the sampling process deterministic given the secret key and watermark context (see App.~Fig~\ref{fig:gumbel-max} for an overview).

\paragraph{Embedding.} At each generation step $t$, the watermark operates on a \emph{watermark context window} $\mathbf{w} = (x_{t-k}, \ldots, x_{t-1})$, consisting of the $k$ last generated tokens. This window, together with the secret key $\sk$, seeds a Pseudo-Random Function (PRF) that assigns a pseudo-random value $R_v \in [0,1]$ to every candidate token $v$ in the vocabulary:
$$R_v = \text{PRF}(v, \mathbf{w}, \sk).$$
The PRF is deterministic: for a given context window, secret key, and candidate token, it always returns the same value. However, its output is indistinguishable from uniform randomness to anyone who does not know $\sk$ (see \autoref{app:hash} for implementation details of the PRF).

The watermark then selects the next token by combining these pseudo-random values with the LLM's probability distribution. Concretely, it picks:
$$x_t = \arg\max_{v \in \V} \; R_v^{1/p_v^{(t)}},$$
where $p_v^{(t)}$ is the probability assigned to token $v$ by the LLM at step $t$. 
This balances two factors: tokens with high model probability $p_v$ are naturally favored, but among tokens of similar probability, the one with the highest PRF value $R_v$ wins. 
This creates a statistical correlation between the chosen tokens and the secret key, which can later be detected.

This selection rule is equivalent to two well-known sampling schemes:
\begin{itemize}
    \item \textbf{Inverse Transform Method:} Sort tokens by descending probability, compute the CDF, and select the token corresponding to the quantile $u = \text{PRF}(\mathbf{w}, \sk)$.
    \item \textbf{Gumbel-Max Trick:} Sample Gumbel noise $G_v = -\log(-\log(R_v))$ for each token and select $x_t = \arg\max_v (G_v + \log p_v^{(t)})$.
\end{itemize}
Put differently, the watermarking scheme samples from the original distribution $\vec{p}$, but uses a deterministic source of randomness derived from the secret key and context, instead of true randomness.
This is what gives it the single-token non-distortion property (\autoref{def:single_token_distortion_free}), as formalized in \autoref{prop:sampling} below.

\paragraph{Detection.} Given a text, the detector re-computes the PRF values using the secret key and the preceding tokens, then checks whether the score is higher than expected by chance.

We denote by $x^{(1)}, \ldots, x^{(T)}$ the sequence of tokens in the text, and by $\vec{R}^{(t)} \in [0,1]^{|\mathcal{V}|}$ the key random vector re-computed from the $k$ preceding tokens and the secret key.
We define $R_t := R^{(t)}_{x^{(t)}}$, the PRF value of the token selected at time-step $t$.
The detection score is calculated as:
$$ S_T=-\sum_{t=1}^{T} \ln (1-R_t). $$
Intuitively, watermarked tokens tend to have high $R_t$ values (since the selection rule favors them), making $-\ln(1-R_t)$ large. For unwatermarked text, $R_t$ values are essentially random, yielding a lower score. 
A statistical test then determines whether the observed score is significantly higher than expected under the null hypothesis $\Hnull$ (no watermark).
In practice, we choose a threshold $\tau$ (depending on the desired false positive rate) and flag a text as watermarked if $S_T > \tau$.

\subsubsection{Theoretical Properties}

The following results formalize the two key guarantees of Gumbel-max watermarking: single-token non-distortion and detectability. The proofs are not original contributions; they were presented by \citet{aaronson2023watermarking} and formalized by \citet{fernandez2023three}. We provide them in \autoref{app:gumbel_proofs}.

\subsubsection{Single-Token Non-Distortion}

\begin{proposition}[Sampling probability]
\label{prop:sampling}
Consider a discrete distribution $\vec{p}=(p_1,\ldots,p_V)$
and $V=|\V |$ random variables $\vec{R} = (R_1,\ldots,R_V)$ s.t. $R_v\overset{iid}{\sim}\mathcal{U}_{[0,1]}$.
Let $V^\star = \arg \max_v R_v^{1/p_v}$.
Then: $$\Prob(V^\star=v) = p_v.$$
\end{proposition}

\begin{corollary}
\label{cor:beta}
Conditionally on $V^\star = v$, $R_{V^\star} \sim \text{Beta}(1/p_v, 1)$.
\end{corollary}

\autoref{prop:sampling} establishes that Gumbel-max is single-token non-distortionary (\autoref{def:single_token_distortion_free}): in expectation over the random key, the selected token follows exactly the LLM's original distribution. 
The corollary characterizes the distribution of the PRF value for the selected token, which is useful for the detection analysis below.

\begin{remark}[Repeated $n$-grams and single-sequence non-distortion]
\label{rem:deduplication}
Single-token non-distortion (\autoref{prop:sampling}) does not guarantee that a full \emph{sequence} is distributed as the original LLM.
\citet{dathathri2024scalable} define a stronger property: \emph{single-sequence non-distortion} (\autoref{def:single_sequence_distortion_free}), requiring that the joint probability of a complete response is preserved: $\E_k[P_{\mathrm{wm}}(\vec{y} \mid \vec{x}, k)] = p_{\mathrm{LM}}(\vec{y} \mid \vec{x})$.
This is violated whenever the same $k$-gram context repeats within a generation: the PRF produces identical values, so the same token is always re-selected, creating unwanted correlations in the sequence distribution.

To restore single-sequence non-distortion, we can apply \emph{repeated context masking}~\citep{dathathri2024scalable}: a set $\mathcal{S}$ of seen context windows is maintained per generation; on the first occurrence the watermark is applied, on subsequent occurrences the sampler falls back to standard unwatermarked sampling.
Our main evaluations do not enforce this protocol (except in \autoref{app:variance_analysis}), as repeated $k$-grams are rare in practice with $k \geq 3$.
\end{remark}

\subsubsection{Detectability}

\begin{proposition}[$p$-value under $\Hnull$]
\label{prop:pvalue}
Under $\Hnull$ (text not watermarked), the score $S_T$ follows a $\Gamma(T,1)$ distribution. The $p$-value associated to a score $s$ is:
\begin{equation}
\text{$p$-value}(s) = \Prob(S_T>s \mid \Hnull) = \frac{\Gamma(T,s)}{\Gamma(T)},
\end{equation}
where $\Gamma(T,s)$ is the upper incomplete gamma function.
\end{proposition}

This provides an \emph{exact} false positive rate: given any desired significance level $\alpha$, we can compute a detection threshold $\tau$ such that the probability of wrongly flagging unwatermarked text as watermarked is exactly $\alpha$. 

\begin{proposition}[Expected score under $\Halt$]
\label{prop:expected_score}
Under $\Halt$ (text is watermarked),
$\displaystyle \mathbb{E}(S_T) \geq T +  \left( \frac{\pi^2}{6} -1 \right) H_T$,
where $H_T = - \sum_{t=1}^T p_t\ln(p_t)$ is the entropy of the completion.
\end{proposition}

\emph{This bound reveals that detectability scales with the \emph{entropy} of the generated text.} When the LLM is uncertain (high entropy), many tokens have non-negligible probability, giving the watermark more room to influence the selection and producing a stronger signal. Conversely, when the model is very confident (low entropy), the top token dominates regardless of the PRF values, and the watermark signal is weak. This entropy dependence motivates the entropy-weighted detection of \autoref{subsec:entropy_detection}.

\section{Method: TextSeal}
\label{sec:method}

\begin{figure*}[b!]
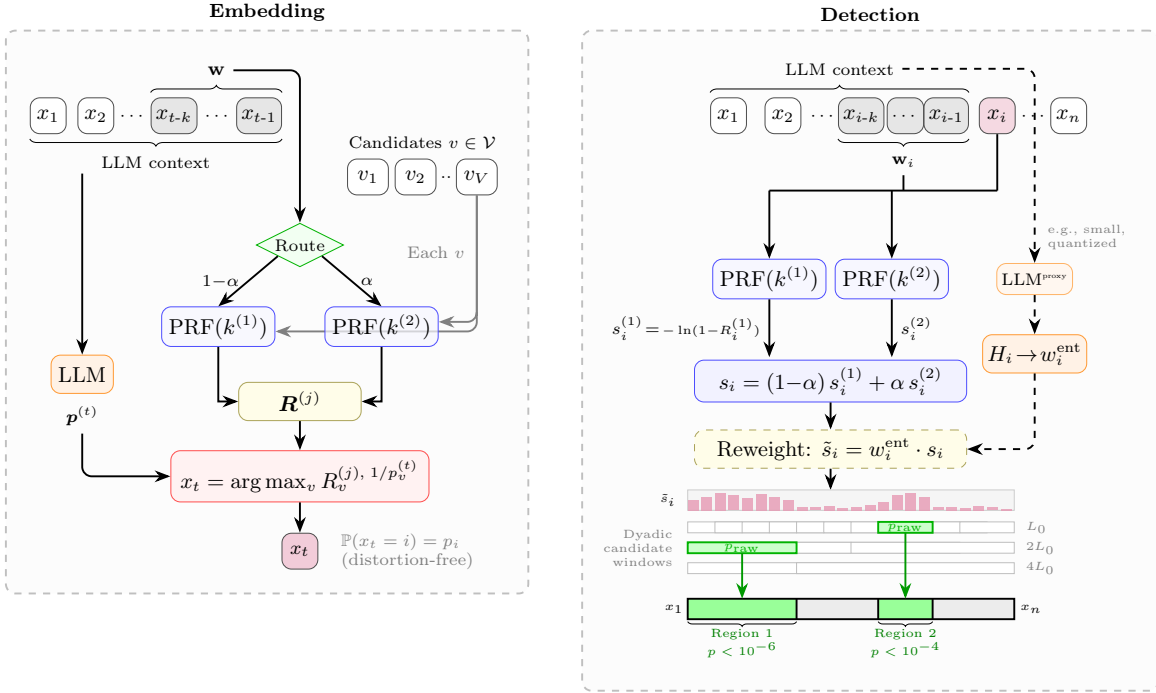

\centering
\figTextSealFull
\caption{TextSeal overview. \textbf{Left (Embedding):} At each step, one of two keys is randomly selected (probability $\alpha$ for $k^{(2)}$, $1{-}\alpha$ for $k^{(1)}$), and the token is chosen via Gumbel-Max using the selected key's PRF (\autoref{subsec:dual_key}). \textbf{Right (Detection):} Scores are computed under both keys and fused per-token, weighted by entropy (\autoref{subsec:entropy_detection}), then a geometric cover search localizes watermarked regions (\autoref{subsec:localization}).}
\label{fig:textseal-full}
\end{figure*}

TextSeal addresses three key limitations of the standard Gumbel-max watermark: its deterministic outputs, its suboptimal detection in mixed-entropy text, and the lack of localized detection capability.
We describe each improvement below, and present an overview in \autoref{fig:textseal-full}.

\subsection{Dual-Key Routing for Diversity and Speculative Decoding}
\label{subsec:dual_key}

Gumbel-Max is deterministic: for a given context and secret key, the output token is fixed, so regenerating the same prompt always produces identical text, limiting user experience and triggering repetitive loops~\citep{holtzman2019curious}.
TextSeal addresses this by maintaining two secret keys $k^{(1)}$ and $k^{(2)}$ that restore diversity while preserving both detectability and single-token non-distortion (\autoref{def:single_token_distortion_free}).

\paragraph{Embedding.}
At each generation step $t$, one key is selected at random: $k^{(1)}$ with probability $1-\alpha$, or $k^{(2)}$ with probability $\alpha$. The token is produced via Gumbel-Max using the selected key's PRF:
\begin{equation}
\label{eq:dual_key_embedding}
    x_t = \arg\max_v \; R_v^{(k), \; 1/p_v^{(t)}}, \quad k \in \{k^{(1)}, k^{(2)}\}
\end{equation}
The routing probability $\alpha \in [0, 0.5]$ controls the diversity-detectability trade-off: $\alpha=0$ is the original single-key scheme with maximum detectability but no diversity, while $\alpha=0.5 $ routes evenly between the two keys.
Dual-key routing also doubles the tolerance to repeated $n$-grams before single-sequence non-distortion is compromised (\autoref{rem:deduplication}): when a context window appears for the first time it is watermarked with one key; if it recurs, the other key is used; only on a third occurrence must the sampler fall back to unwatermarked sampling.

\paragraph{Detection.}
The detector does not know which key generated each token. To capture signal from both potential paths, we compute scores under both keys and aggregate them as a weighted sum:
\begin{equation}
\label{eq:early_fusion}
    s_i = (1-\alpha) \cdot s_i^{(1)} + \alpha \cdot s_i^{(2)}, \quad \text{where } s_i^{(j)} = -\ln(1 - R_i^{(j)})
\end{equation}
We call this strategy ``early fusion'', in contrast with methods that would compute two p-values and aggregate them later.
Under $\Hnull$, $s_i$ is a weighted combination of independent exponentials with mean $1$ and variance $\theta_R = \alpha^2 + (1-\alpha)^2$. 
The final p-value is computed using the unified framework in \autoref{subsec:entropy_detection}.
We show that this early-fusion approach is better than Fisher or Bonferroni aggregations in~\autoref{subsec:early_vs_late_proof}, and  support it empirically in~\autoref{subsec:diversity}.

\paragraph{Compatibility with speculative decoding.}
In speculative decoding~\citep{leviathan2023fast}, a draft model $P_D$ proposes tokens accepted by a target model $P_T$. With dual-key watermarking, the draft uses $k^{(1)}$ and rejected correction tokens use $k^{(2)}$. The draft acceptance rate naturally determines the routing ratio $\alpha$. Since this rate varies by domain and model pair, $\alpha$ can be calibrated at detection time or set to $0.5$ as a robust default that makes detection method invariant to the true routing ratio. 
This extends naturally to Multi-Token Prediction (MTP)~\citep{gloeckle2024better}, where all $K$ auxiliary heads share $k^{(1)}$ and fall back to $k^{(2)}$, preventing the fracturing of statistical power across many keys.

\subsection{Entropy-Weighted Detection}
\label{subsec:entropy_detection}

\paragraph{Entropy Weighting.}
When the next-token distribution has low entropy, the top token already has probability close to $1$, so the choice is weakly influenced by the PRF value $R_v$ and carries little watermark signal. 
TextSeal therefore weights each token's detection score by its local entropy $H_i$, so that high-entropy positions contribute more to the final statistic. 
We estimate $H_i$ with a single forward pass of an auxiliary model, \eg a smaller or quantized model from the same family as the generator.

Formally, we assign each token-level score $s_i$ an entropy weight $w_i^{\text{ent}}$ and compute
\begin{equation}
\label{eq:s_combined}
    S_{\text{combined}} = \sum_{i=1}^n w_i^{\text{ent}} \cdot s_i, \qquad \text{where }
    w_i^{\text{ent}} = 0.1 + 0.9 \cdot \frac{H_i - H_{\min}}{H_{\max} - H_{\min}}.
\end{equation}
where the entropy is normalized within the sequence so the weights span a broad dynamic range regardless of the absolute entropy scale. 
This attenuates low-entropy tokens instead of letting them dilute the score, while preserving the strongest evidence from uncertain positions. Unlike prior entropy-filtering approaches~\citep{lee2023wrote} that threshold the entropies, our scheme is continuous: every token still contributes, but with strength matched to its expected usefulness. 
Since the null statistic is a weighted sum of independent exponentials, the moment-matched Gamma approximation below provides calibrated $p$-values that explicitly account for these entropy weights.

\paragraph{Moment-Matched Gamma Approximation.}
$S_{\text{combined}}$ is a weighted sum of independent, non-identical exponentials, which follows a hypoexponential whose CDF, while closed-form, is numerically unstable when rates are similar and costly to evaluate for large $n$.\footnote{The hypoexponential CDF involves a sum of $n$ exponentials with coefficients $\prod_{j \neq i} \lambda_j / (\lambda_j - \lambda_i)$, which become unstable when rates are close. 
In contrast, the Gamma CDF reduces to the well-optimized incomplete gamma function.}
A Gaussian approximation fails to capture the heavy-tailed scores, so we use moment matching instead.
Under $\Hnull$, each term has mean $w_i^{\text{ent}}$ and variance $(w_i^{\text{ent}})^2 \theta_R$.
We fit $S_{\text{combined}} \sim \mathrm{Gamma}(k_{\text{new}}, \theta_{\text{new}})$ by matching the first two moments:
\begin{equation}
\label{eq:gamma_moments}
    \theta_{\text{new}} = \theta_R \frac{\sum (w_i^{\text{ent}})^2}{\sum w_i^{\text{ent}}}, \quad \quad k_{\text{new}} = \frac{(\sum w_i^{\text{ent}})^2}{\theta_R \sum (w_i^{\text{ent}})^2}
\end{equation}
The resulting $p$-value is:
\begin{equation}
\label{eq:p_value_gamma_combined}
p\text{-value}(S_{\text{combined}}) = 1 - F_{\Gamma}\!\left(S_{\text{combined}}; k_{\text{new}}, \theta_{\text{new}} \right)
\end{equation}
We show in \autoref{sec:fpr-check} that this approximation is well-calibrated under $\Hnull$, and in \autoref{subsec:diversity} that it significantly outperforms unweighted detection.
This framework handles dual-key routing ($\alpha$) and entropy gating ($w_i^{\text{ent}}$) in a single frequentist test.

\subsection{Multi-Region Localization and Adaptive Ensemble}
\label{subsec:localization}

When a document contains multiple scattered watermarked regions (e.g., distinct AI-generated paragraphs pasted into a human-written essay), evaluating a global score suffers from two critical flaws. 
First, the unwatermarked background tokens severely dilute the statistical signal. 
Second, it fails to identify the specific provenance of individual segments, which is critical for practical attribution.

\paragraph{\textbf{Geometric Cover Search \& Greedy Extraction.}}
To solve this, our goal is to extract a set of disjoint watermarked intervals $\{[a_1, b_1], \dots, [a_y, b_y]\}$. A naive search over all $\mathcal{O}(n^2)$ possible start and end pairs is computationally prohibitive and incurs a massive multiple-testing penalty. Instead, we employ a geometric cover search, reducing the space to dyadic window lengths $L \in \{L_0, 2L_0, 4L_0, \ldots, 2^{\lfloor \log_2 n \rfloor}\}$, sliding each window across the text at half-length offsets. This yields a strictly bounded number of candidate windows, $M \approx 4n / L_{\min}$.

The extraction proceeds in two stages. First, we rank all $M$ windows by their raw score sum (computed in $\mathcal{O}(1)$ per window via prefix sums). Then, for the top candidates only, we compute the rigorous entropy-weighted Gamma $p$-value. The greedy extraction selects the window with the lowest $p$-value, flags it as watermarked, masks its tokens, and repeats on the residual text, aggregating intervals as long as their combined significance overcomes the multiple-testing tax.
This localized extraction is governed by the minimum zone length $L_{\min}$ (default 50) and the maximum number of zones $Y_{\max}$ (default 5). Full mathematical details are provided in Appendix~\ref{app:localization_math}.

\paragraph{\textbf{Adaptive Ensemble Detection.}}
Discovering $y$ regions among $M$ candidates incurs a combinatorial multiple-testing tax. To adapt to any editing behavior, our ensemble selects the most significant among three strategies, applying a flat Bonferroni correction ($k=3$): (1) \textbf{Global} full-text test (no search penalty), (2) \textbf{Single-Best} window (penalized by $\log_{10} M$), and (3) \textbf{Multi-Region} aggregation over $y$ zones (penalized by $\log_{10} \binom{M}{y} + \log_{10} Y_{\max}$). The final significance score is:
\begin{equation}
\log_{10} p_{\text{final}} = \min\!\big(\log_{10} p_{\text{global}}, \log_{10} p_{\text{single}}, \log_{10} p_{\text{multi}}\big) + \log_{10} 3.
\end{equation}

\paragraph{\textbf{The Dilution Rescue Effect.}}
For largely unedited text, the ensemble gracefully defaults to the global test, paying a negligible worst-case penalty of $\log_{10}(3)$. Consider $w$ watermarked tokens with expected per-token score $\mu > 1$ (as given by \autoref{prop:expected_score}), split into $y$ chunks within a document of length $n$. Under extreme dilution ($n \gg w$), the global test's significance drops as its signal-to-noise ratio scales by $\mathcal{O}(w^2 (\mu-1)^2 / n)$. The multi-region strategy isolates the pure signal ($\mathcal{O}(w(\mu-1)^2)$) but pays a combinatorial tax scaling as $\mathcal{O}(y \log_{10} n)$. Localization rescues detection when the isolated signal outpaces this logarithmic tax: $w(\mu-1)^2 \gtrsim y \log_{10} n$. For instance, $w=800$ tokens ($\mu=1.2$) in $y=5$ chunks easily overcome the $5 \log_{10} n$ tax, allowing confident detection even within $n=100,000$ tokens---a scenario where the global signal is destroyed.

\paragraph{\textbf{High-Resolution Boundary Annotation (mIoU).}}
While the greedy ensemble rigorously bounds the False Positive Rate, the harsh combinatorial tax forces it to prematurely discard small fragments, making it suboptimal for exact boundary estimation (mean Intersection over Union, or mIoU). To achieve high-resolution localization, we decouple \emph{detection} from \emph{annotation}. If the ensemble definitively rejects the null hypothesis $\mathcal{H}_0$, we drop the search taxes and apply a localized density smoother. Tokens satisfying a normalized weighted moving average $\bar{S}_i > \tau$ are locally annotated as watermarked, allowing the recovery of fine-grained, sentence-level provenance (see Appendix~\ref{app:localization_math} for exact formulation).

\section{Main Experiments}
\label{sec:experiments}

\subsection{Experimental Setup}
\label{subsec:exp_setup}

\paragraph{Models \& Datasets}
Unless stated otherwise, we use Qwen~3.5-27B~\citep{qwen3.5} for generation, with $T=0.8$, top-$p=0.9$, and reasoning disabled, and we use entropy-weighted detection (\autoref{subsec:entropy_detection}) with the more lightweight Qwen~3.5-0.8B model.
We evaluate on 1k prompts from the ELI5 dataset~\citep{eli5} (with 5 different seeds), truncating answers to 400 tokens. 

We compare \emph{TextSeal} (default mixing parameter $\alpha = 0.1$ from~\autoref{subsec:dual_key}) to \emph{Gumbel-Max}~\citep{aaronson2023watermarking} and \emph{SynthID-Text}~\citep{dathathri2024scalable} with depth 10.
\emph{SynthID-Text} embeds a watermark via multi-layered tournament sampling with binary random functions, and proposes two detection methods: (i)~a frequentist Z-test over the mean tournament score, and (ii)~a Bayesian detector that estimates the posterior $\Prob(\text{watermarked} \mid \text{scores})$ via a logistic regression or MLP trained on a representative dataset.
We use the frequentist Z-test, because the Bayesian detector provides no controlled false-positive rate, does not generalize across domains (its posteriors depend on the training distribution), and is incompatible with localized multi-window testing (full discussion in Appendix~\ref{app:synthid_details}).
We fix the watermark context window size to $k=3$ for all methods, meaning the pseudo-random function depends on the three preceding tokens.
At detection time, we deduplicate (context window, token) tuples, because the PRF is deterministic and repeated tuples would yield identical scores, violating the independence assumption of the statistical test~\citep{fernandez2023three}.

\subsection{Detectability-Diversity Trade-off}
\label{sec:exp1}

A practical watermark must embed a robust signal without changing the output distribution. 
For~\autoref{fig:diversity-detectability}, we vary for \emph{TextSeal} the mixing parameter $\alpha$ from~\autoref{subsec:dual_key} from 0 (deterministic) to 0.5 (blue and green curves). 
For \emph{SynthID}, we vary the depth from 2 to 20.
\emph{TextSeal} consistently dominates the detectability--diversity trade-off.
Furthermore, using the 27B model for entropy detection (``TextSeal high'') boosts detectability by 1--2 orders of magnitude at higher detection cost.

\subsection{Performance on Benchmarks}
\label{subsec:benchmarks}

\begin{table}[t!]
\caption{Accuracy across multiple benchmarks with and without \emph{TextSeal} (SQA\textsuperscript{*}~=~SimpleQA).
No significant performance drop is observed across benchmarks, confirming that TextSeal preserves the capabilities of the underlying model.
}\label{tab:benchmark_results}
\centering
\small
\resizebox{\textwidth}{!}{
\setlength{\tabcolsep}{3pt}
    \begin{tabular}{llccccccccccccccccc}
    \toprule
    \multicolumn{2}{c}{} & \multicolumn{4}{c}{Math} & \multicolumn{3}{c}{Code} & \multicolumn{4}{c}{Knowledge} & \multicolumn{5}{c}{Common Sense} &  \\
    \cmidrule(lr){3-6} \cmidrule(lr){7-9} \cmidrule(lr){10-13} \cmidrule(lr){14-18}
    Reasoning temp. & WM & AIME & MATH & GSM8K & Avg & HE & MBPP & Avg & MMLU & GPQA & SQA\textsuperscript{*} & Avg & HS & WG & ARC-E & ARC-C & Avg & Avg \\
    \midrule
    \addlinespace[2pt]
    0.6 & \cmark & 41.1 & 79.8 & 96.0 & 72.3 & 93.3 & 49.2 & 71.2 & 51.5 & 50.0 & 16.0 & 39.2 & 94.8 & 93.5 & 93.7 & 88.5 & 92.6 & \textbf{70.6} \\
    & \xmark & 40.1 & 79.8 & 95.4 & 71.7 & 97.0 & 50.2 & 73.6 & 49.2 & 50.5 & 15.8 & 38.5 & 94.7 & 93.2 & 93.4 & 88.3 & 92.4 & \textbf{70.6} \\
    \midrule
    \addlinespace[2pt]
    1.0 & \cmark & 37.1 & 77.9 & 96.1 & 70.4 & 94.5 & 48.5 & 71.5 & 48.9 & 45.5 & 13.7 & 36.0 & 94.6 & 93.8 & 92.8 & 86.0 & 91.8 & \textbf{69.1} \\
    & \xmark & 35.8 & 78.4 & 96.1 & 70.1 & 98.2 & 49.3 & 73.7 & 46.4 & 42.9 & 15.5 & 34.9 & 94.6 & 93.6 & 93.8 & 86.6 & 92.2 & \textbf{69.3} \\
    \bottomrule
    \end{tabular}
}
\end{table}

We evaluate how TextSeal's watermarking impacts performance across a suite of 12 benchmarks spanning math, code, general knowledge, and common sense domains. 
We use Qwen~3.5-27B with $T=0.6$ (mild watermarking) and $T=1.0$ (stronger watermarking\footnote{The temperature controls the strength of the watermark since an increased temperature leads to higher entropy.}) and compare against vanilla generation without watermarking.
Each benchmark is evaluated with generation at top-$p=0.95$, reasoning temperatures $0.6$ or $1.0$ and a maximum reasoning budget of 3,000 tokens.

On average, TextSeal preserves the performance of the underlying model across benchmarks and temperature settings, with no significant differences.
However, we observe a slight performance drop on code benchmarks (Human-eval: HE and MBPP) of 1-2 points. 
Analyzing the outputs suggests that this drop comes from minor formatting omissions rather than incorrect reasoning or algorithmic failures. 
In particular, all watermarked generations from HE that fail with watermarking and not without fail because they give only the function definition while still using annotations such as \texttt{List[...]}, without adding the required \texttt{from typing import List} import.
We note that benchmark evaluation inherently involves noise from the stochastic generation.
To quantify this variance, we re-ran a subset of benchmarks with multiple random seeds and secret keys.
The observed differences between watermarked and non-watermarked conditions fall within the variance introduced by seed/key changes, confirming that watermarking does not systematically degrade or improve performance (see \autoref{app:variance_analysis}).

\subsection{Human Evaluation of Imperceptibility}
\label{subsec:human_eval}

\begin{table}[t!]
\centering
\small
\caption{
    Human preference evaluation (majority vote aggregation over 3 annotators per sample).
    Net Win Rate: $(n_{\text{WM}} - n_{\text{Base}}) / N$.
    $p$-value: two-sided binomial test on decisive samples against 50\%.
    No test reaches significance after Bonferroni correction ($\alpha/6 = 0.008$).
}
\label{tab:human_eval}
\begin{tabular}{lrrrccr}
\toprule
Language & WM Wins & Base Wins & Ties & WM Preference Rate & $p$-value & Net Win Rate \\
\midrule
English  & 124 & 146 & 1{,}730 & 45.9\% & 0.20 & $-1.1$\% \\
Arabic   & 201 & 181 & 618 & 52.6\% & 0.33 & $+2.0$\% \\
Chinese  & 84 & 74 & 842 & 53.2\% & 0.47 & $+1.0$\% \\
Hindi    & 91 & 89 & 820 & 50.6\% & 0.94 & $+0.2$\% \\
Japanese & 143 & 130 & 727 & 52.4\% & 0.47 & $+1.3$\% \\
\midrule
Overall  & 643 & 620 & 4{,}737 & 50.9\% & 0.54 & $+0.4$\% \\
\bottomrule
\end{tabular}
\end{table}

We assess whether the watermark introduces perceptible quality degradation through a human A/B preference study.
Following the methodology of \citet{dathathri2024scalable}, we generate paired responses to questions from ELI5~\citep{eli5} (2{,}000 English samples) and CaLMQA~\citep{arora2025calmqa} (1{,}000 each for Arabic, Chinese, Hindi, Japanese), totaling 6{,}000 question-answer pairs.

Each pair is evaluated by three annotators (via Appen) with qualifications requiring post-graduate education, native-level language fluency, and at least two years of experience.
Annotators select among four options: \emph{A is preferred}, \emph{B is preferred}, \emph{both equally good}, or \emph{both equally bad}, without knowing which output is watermarked.
We aggregate via majority vote, merging the two tie categories and defaulting split votes (one vote per category) to tie.

\paragraph{Results.}
\autoref{tab:human_eval} reports preference rates after majority vote aggregation.
We test whether the watermark win rate among decisive (non-tie) samples differs from 50\% using a two-sided binomial test.
No individual language reaches significance (all $p > 0.05$), and no test is significant after Bonferroni correction for the six comparisons ($\alpha/6 = 0.008$).
The majority of samples (79\%) result in ties, and inter-annotator agreement is high (88\% of samples have at least 2/3 consensus on the 4-class scale).
We also report the \emph{net win rate}, defined as $(n_{\text{WM}} - n_{\text{Base}})/N$, with $n_{\text{WM}}$ and $n_{\text{Base}}$ the counts of samples where the watermarked or baseline response is preferred, and $N$ the total number of samples.
The overall net win rate is $+0.38\%$, indicating a negligible difference.

To rigorously establish imperceptibility rather than failing to detect a difference, we apply the Two One-Sided Tests (TOST) procedure~\citep{schuirmann1987comparison} for equivalence testing.
We test $|P(\text{WM preferred}) - P(\text{Base preferred})| < \Delta$ over all $N$ samples (including ties), which provides greater statistical power than restricting to decisive samples alone.
With a smallest effect size of interest $\Delta = 5\%$, equivalence is established overall and for four of five individual languages ($p < 0.05$); Arabic marginally fails ($p = 0.06$) due to its wider confidence interval.
The overall 90\% CI is $[-0.6\%,\, +1.4\%] \subset [-5\%,\, +5\%]$.
Full breakdowns and the equivalence testing methodology are provided in Appendix~\ref{app:human_eval_details}.

\subsection{Localization in Mixed Documents}\label{sec:exp_localization}                                

In practice, watermarked text often forms only a fraction of a document (e.g., AI-generated paragraphs within a human-written report). A global detector scoring the entire text faces two primary challenges: \emph{dilution}, where unwatermarked tokens degrade the signal-to-noise ratio, and \emph{fragmentation}, where watermarked content is scattered across non-contiguous regions.
We evaluate TextSeal's adaptive ensemble (Section~\ref{subsec:localization}) against global detection by embedding 400-token watermarked answers ($T=1.0$, top-$p=0.95$, chosen to increase the watermark signal) into unwatermarked Wikipedia texts. Under \textbf{dilution ($K{=}1$)}, we place a single contiguous 400-token block inside documents of increasing length, up to $12{,}000$ tokens (watermarked fraction: $3.3\%$). Under \textbf{fragmentation ($K{>}1$)}, we split the 400 tokens into $K \in \{1, 2, 3, 5\}$ equal fragments interleaved within a fixed $8{,}000$-token document.

\begin{figure}[b!]
  \centering
  \includegraphics[width=\linewidth]{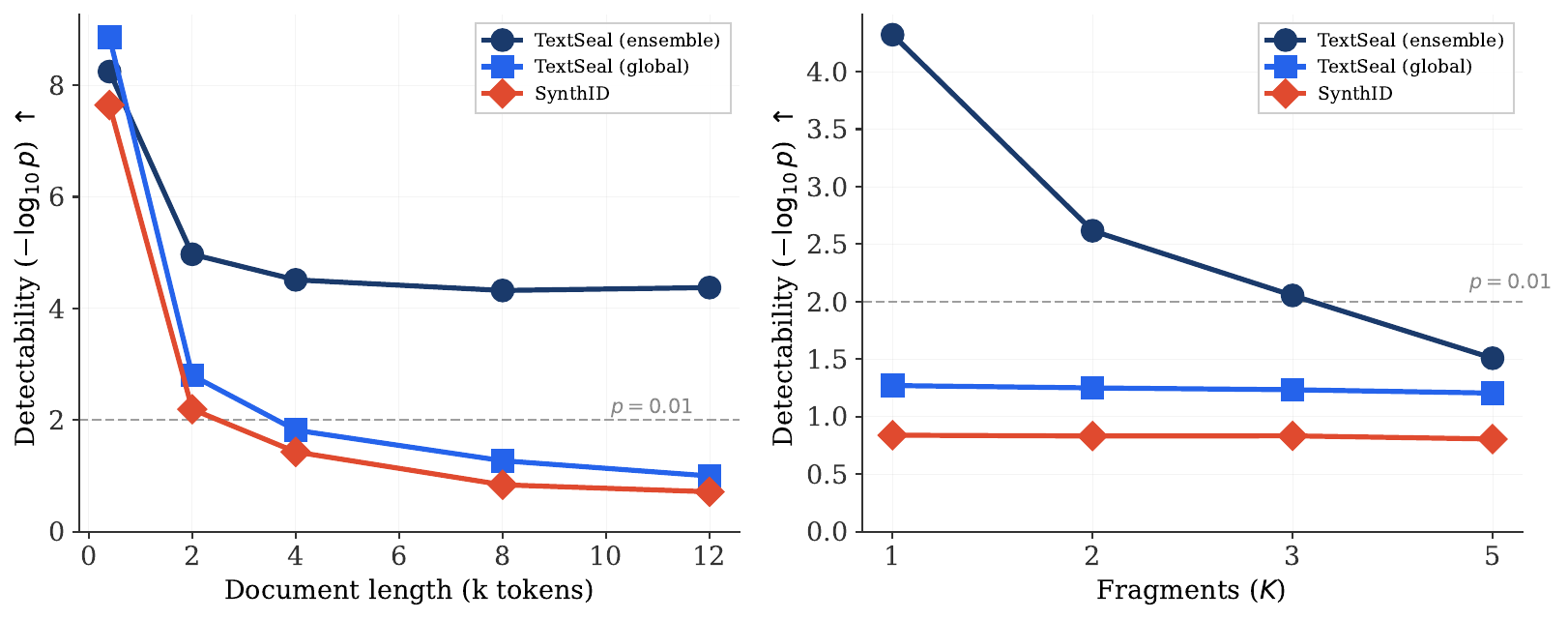}
  \caption{
      Localized detection in mixed documents containing 400 watermarked tokens. 
      \textbf{(Left)}~\emph{Dilution}: A single watermarked block ($K{=}1$) embedded in documents of increasing length. Global detection (light blue and red curves) degrades rapidly as the watermarked fraction shrinks, dropping below the $p{=}0.01$ significance threshold around $4$k tokens. The adaptive ensemble (dark blue curve) maintains strong detectability ($-\log_{10}p > 4$) even at $12$k tokens ($3.3\%$ watermarked).
      \textbf{(Right)}~\emph{Fragmentation}: Watermarked text split into $K$ fragments within an $8000$-token document. Global detectors are insensitive to fragmentation (flat curves at $-\log_{10}p \approx 1$), while the ensemble leverages localized search to extract the signal at $K{\leq}3$ fragments.
  }
  \label{fig:localization}
\end{figure}

\paragraph{Results.}
As shown in \autoref{fig:localization} (left), global detection suffers heavily from dilution, degrading at roughly $O(1/T)$ and failing to reach significance ($p{=}0.01$) beyond $T{=}4000$. TextSeal's adaptive ensemble, however, efficiently isolates the signal, maintaining strong detectability ($-\log_{10} p > 4$) even at $T{=}12{,}000$ (a $30{\times}$ dilution).
For fragmentation (\autoref{fig:localization}, right), global detectors exhibit flat performance, as they are blind to spatial arrangement. Conversely, the ensemble successfully detects the watermark for up to $K{=}3$ fragments. Performance only degrades at $K{=}5$, where individual fragments (${\sim}80$~tokens) become too small to overcome the statistical penalty of multiple-hypothesis testing. Overall, TextSeal's localized approach dramatically outperforms global baselines whenever watermarked content is reasonably concentrated within the document.

\section{Ablations and Analyses}
\label{sec:ablations}

\subsection{Diversity Strategies Comparison}
\label{subsec:diversity}

We compare four strategies for restoring diversity in Gumbel-max watermarking (full descriptions and proofs in Appendix~\ref{app:diversity_bounds}):
\textbf{(1)~Stochastic Mixing} mixes the PRF value with a Bernoulli coin (control: mixing rate $a$);
\textbf{(2)~Periodic Skip} disables watermarking at fixed intervals (control: skip rate $\alpha$);
\textbf{(3)~Entropy-Normalized Skip} skips watermarking with probability $\tau$ uniformly across entropy regimes, preserving distortion-freeness (control: target skip rate $\tau$);
\textbf{(4)~Dual-Key Routing} (\autoref{subsec:dual_key}) alternates between two secret keys (control: routing probability $\alpha \in [0, 0.5]$).

\begin{figure}[b!]
    \centering
    \includegraphics[width=0.95\textwidth]{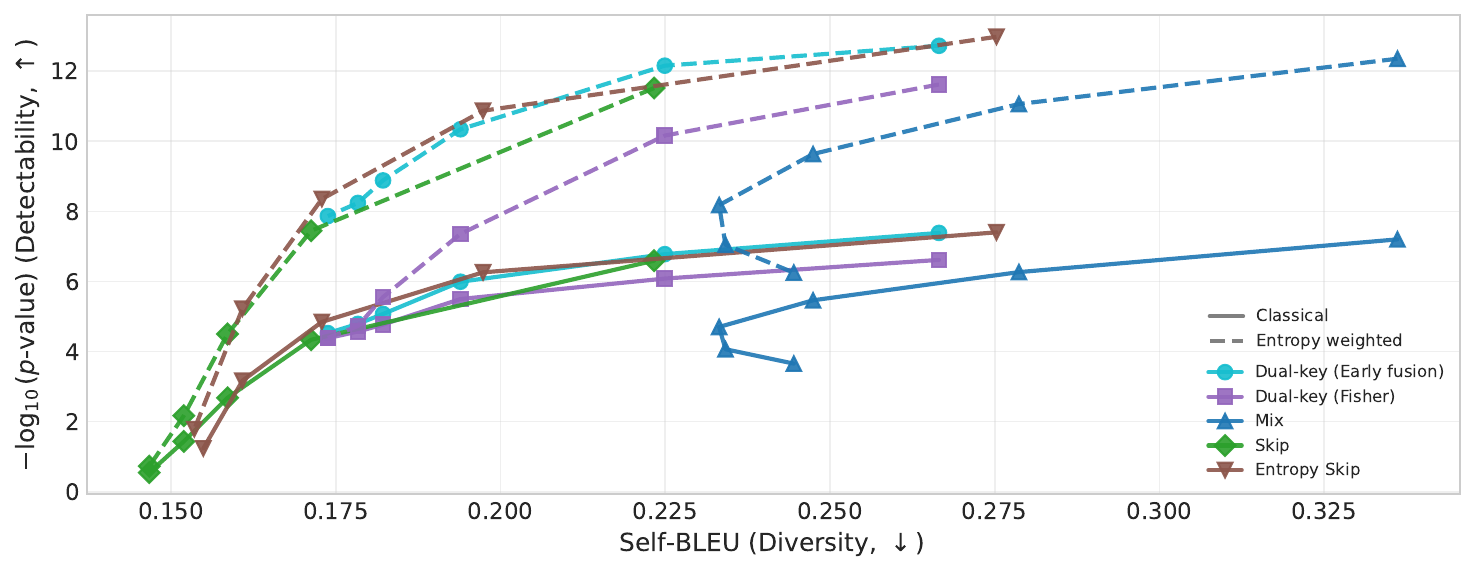}
    \caption{
    \textbf{Pareto frontier of diversity strategies.} 
    Self-BLEU is lower-is-better and median $-\log_{10}(p)$ is higher-is-better.
    Solid lines use the classical detector; dashed lines use entropy weighting.
    Early-fusion dual-key routing outperforms Fisher at matched diversity, and together with entropy skip defines the strongest frontier.
    }
    \label{fig:pareto_diversity}
\end{figure}

\paragraph{Results: Diversity vs.\ Detectability Trade-off.}
We evaluate Qwen~3.5-27B on 1k ELI5 prompts, with reasoning disabled, temperature $1.0$, top-$p=0.95$, maximum generation length 2,048, watermark context size $k=3$, and two generations per prompt to compute Self-BLEU. 
For detection, we report both the classical test and the entropy-weighted Gamma test.
For each method, we vary the control hyperparameter to trace the Pareto frontier between \emph{diversity} (measured by Self-BLEU, where lower indicates more diverse) and \emph{detectability} ($p$-value under $\Hnull$; lower means stronger watermark).

\autoref{fig:pareto_diversity} illustrates the Pareto frontiers for all five methods. 
Ideally, a method should push towards the top-left corner (low Self-BLEU, low $p$-value).
Several trends stand out. First, entropy-weighted detection consistently improves every method, often by several orders of magnitude in median $p$-value, without changing the generation diversity. Second, early-fusion dual-key routing clearly outperforms Fisher-style dual-key aggregation at comparable Self-BLEU, confirming that early fusion is the right detector for routed generation. Third, stochastic mixing is consistently dominated: it reaches similar or worse detectability only at much higher Self-BLEU, making it a poor trade-off in practice.

Among the strongest methods, entropy skip and early-fusion dual-key routing define the best Pareto frontier. Entropy skip is slightly stronger at the highest-detectability end, while early-fusion dual-key routing remains very close across the full sweep and has the practical advantage of mapping directly to speculative decoding and MTP-style deployments. We therefore select dual-key routing as the default diversity mechanism for TextSeal in all experiments.

\subsection{False Positive Rate Check}
\label{sec:fpr-check}

A reliable detector must strictly control its empirical False Positive Rate (FPR) at any nominal threshold $\tau$. 
We validate this on 1 million unwatermarked Wikipedia passages (256 tokens each), rather than ELI5 answers to have more texts and cover a wider distribution.
We plot in~\autoref{fig:fpr_calibration} the empirical FPR against $\tau$; perfect calibration aligns with the diagonal, while curves \emph{above} it indicate safe, conservative behavior. 
Under standard unweighted dual-key detection presented in~\autoref{subsec:dual_key} (\autoref{fig:fpr_calibration}, left), all methods tightly track the diagonal down to $\tau \approx 10^{-4}$. 
Under lightweight (0.8B) entropy-weighted detection as described in~\autoref{subsec:entropy_detection} (\autoref{fig:fpr_calibration}, right), 
TextSeal remains strictly well-calibrated. 

\begin{figure}[b!]
    \centering
    \includegraphics[width=1\linewidth]{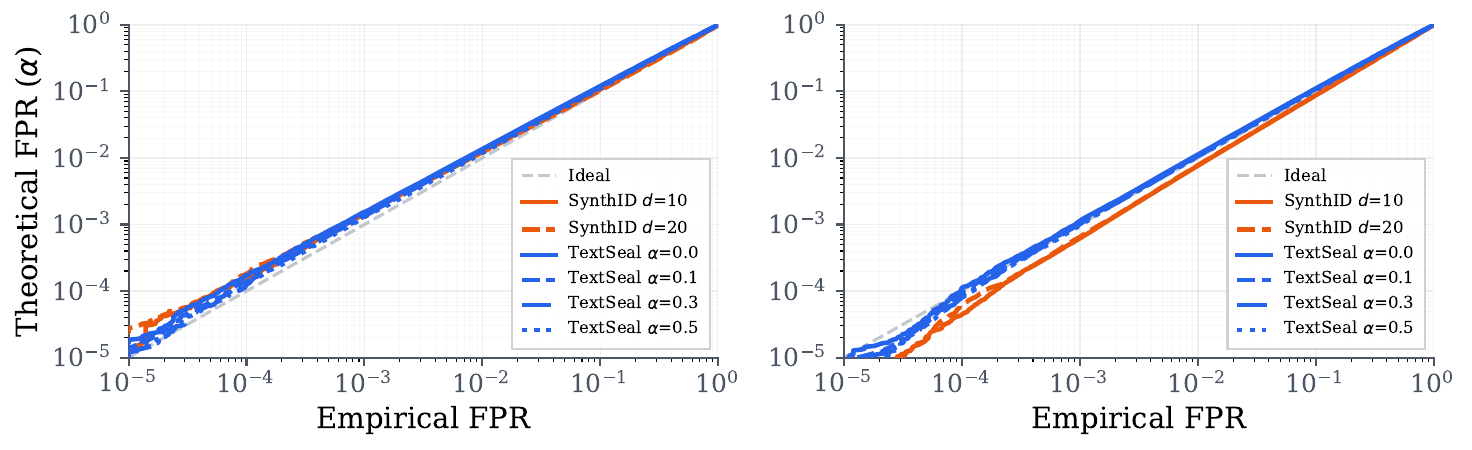}
    \caption{
        Theoretical FPR ($\tau$) vs.\ empirical FPR under standard detection (left) and entropy-weighted linear detection (right), on 1M unwatermarked Wikipedia texts (256~tokens each)
        The dashed diagonal indicates perfect calibration.
        All curves lie above the diagonal (conservative): the empirical FPR never exceeds the nominal level.
        Line styles distinguish parameter settings within each method.
    }
    \label{fig:fpr_calibration}
\end{figure}

\subsection{Generalization: Multilingual Question Answering}
\label{subsec:multilingual_qa}

The experiments above use a single model (Qwen 3.5-27B) on English text.
To test whether TextSeal generalizes across models, languages, and scripts, we evaluate on a multilingual question-answering task using a different model: GPT-OSS-20B, OpenAI's open-weights 20B-parameter reasoning model, with reasoning enabled, on two datasets:
ELI5 (English, 2,000 questions) and CalmQA (Arabic, Chinese, Hindi, Japanese; 1,000 questions each), totaling 6,000 paired samples. 
Each question is answered with and without watermarking (top-$p=0.95$, temperature$=0.7$). 
Full experimental details are in \autoref{app:multilingual_setup}.

\paragraph{Quality Comparison.}
\autoref{tab:multilingual_quality} compares generation quality between watermarked and non-watermarked outputs.
Reasoning lengths show a small increase with watermarking ($\sim$16\% more reasoning tokens, especially in other languages than English), likely due to sampling variance rather than a systematic effect.
Refusal rates are low under both conditions ($<1\%$). 
Script consistency is $>98\%$ for all languages except Japanese (90\%), with a small increase of 1\% for WM.
We use McNemar's tests~\citep{mcnemar1947note} to confirm that there is no statistically significant difference between conditions for either refusal rates ($p=0.41$) or script consistency ($p=0.21$); see \autoref{app:multilingual_setup} for details.
TextSeal achieves 63.3\% TPR at 0.1\% FPR overall; per-language detection results are in App.~\ref{app:multilingual_setup}.

\begin{table}[t!]
\centering
\caption{
    Quality comparison between watermarked (WM) and non-watermarked (Non-WM) answers on 6,000 multilingual QA pairs. 
    Reasoning/Answer tokens: average per response. 
    Refusal/Script: percentage of responses with refusal or wrong language script. 
    Results show no meaningful quality difference between conditions.
}
\label{tab:multilingual_quality}
\small
\setlength{\tabcolsep}{4pt}
\begin{tabular}{l cc cc cc cc}
\toprule
 & \multicolumn{2}{c}{Reasoning tokens} & \multicolumn{2}{c}{Answer tokens} & \multicolumn{2}{c}{Refusal \%} & \multicolumn{2}{c}{Wrong Script \%} \\
\cmidrule(lr){2-3} \cmidrule(lr){4-5} \cmidrule(lr){6-7} \cmidrule(lr){8-9}
Language & WM & Non-WM & WM & Non-WM & WM & Non-WM & WM & Non-WM \\
\midrule
English   & 168 & 149 & 124 & 121 & 0.8 & 0.7 & 0.0 & 0.0 \\
Arabic    & 296 & 232 & 161 & 155 & 0.6 & 0.6 & 0.6 & 0.6 \\
Chinese   & 224 & 203 & 144 & 146 & 0.6 & 0.3 & 0.7 & 0.6 \\
Hindi     & 293 & 229 & 166 & 161 & 1.1 & 1.1 & 1.1 & 1.1 \\
Japanese  & 294 & 280 & 181 & 181 & 0.6 & 0.2 & 9.0 & 7.8 \\
\midrule
Overall & 240 & 207 & 150 & 147 & 0.7 & 0.6 & 1.9 & 1.7 \\
\bottomrule
\end{tabular}
\end{table}

\subsection{Real-World Considerations: Embedding and Detection Efficiency}\label{sec:exp3}

We evaluate TextSeal's computational efficiency during both generation and detection to ensure it is lightweight enough for large-scale deployment.

\begin{table}[t!]
\centering
\caption{Per-token sampling overhead of TextSeal and SynthID watermarking on a single H200 GPU.
Each method is measured on the \emph{same} logits, isolating the sampling cost.
TextSeal uses dual-key Gumbel-Max ($\alpha{=}0.1$, $n$-gram${=}3$); SynthID uses tournament depth $d{=}10$.
Median over 30 prompts of ELI-5 $\times$ 400 tokens.}
\label{tab:generation_overhead}
\smallskip
\small
\begin{tabular}{lrrrrrrr}
\toprule
& & \multicolumn{2}{c}{\textbf{No Watermark}} & \multicolumn{2}{c}{\textbf{TextSeal}} & \multicolumn{2}{c}{\textbf{SynthID}} \\
\cmidrule(lr){3-4} \cmidrule(lr){5-6} \cmidrule(lr){7-8}
\textbf{Model} & \textbf{Fwd} & \textbf{Sample} & \textbf{tok/s} & \textbf{Sample} & \textbf{Overhead} & \textbf{Sample} & \textbf{Overhead} \\
& (ms) & (ms) & & (ms) & & (ms) & \\
\midrule
Qwen~3.5-0.8B  & 21.4 & 0.37 & 45.9 & 0.43 & $0.3\%$    & 0.61 & $1.1\%$ \\
Qwen~3.5-2B    & 21.5 & 0.36 & 45.8 & 0.43 & $0.3\%$    & 0.60 & $1.1\%$ \\
Qwen~3.5-4B    & 30.3 & 0.38 & 32.6 & 0.44 & $0.2\%$    & 0.62 & $0.8\%$ \\
Qwen~3.5-9B    & 31.3 & 0.38 & 31.5 & 0.45 & $0.2\%$    & 0.62 & $0.8\%$ \\
Qwen~3.5-27B   & 61.9 & 0.39 & 16.1 & 0.46 & $0.1\%$    & 0.63 & $0.4\%$ \\
\bottomrule
\end{tabular}
\end{table}

\paragraph{Generation Overhead.}
Table~\ref{tab:generation_overhead} details the sampling overhead during autoregressive decoding. 
TextSeal evaluates a fused dual-key pseudorandom function (PRF) restricted strictly to the top-$p$ survivor tokens (${\sim}200$ tokens), avoiding full-vocabulary hashes. This adds only ${\sim}0.07$\,ms per token (${\le}0.3\%$ overhead).
In contrast, SynthID's iterative tournament sampling requires $d$ sequential rounds of top-$p$ reweightings and a multinomial sampling step, costing ${\sim}0.6$\,ms per token ($0.4$--$1.1\%$ overhead). 
Crucially, both methods operate entirely on the logits, requiring no model parameter changes or KV-cache modifications, ensuring immediate compatibility with standard serving infrastructure.

\begin{figure}[b!]
  \centering
\includegraphics[width=0.8\linewidth]{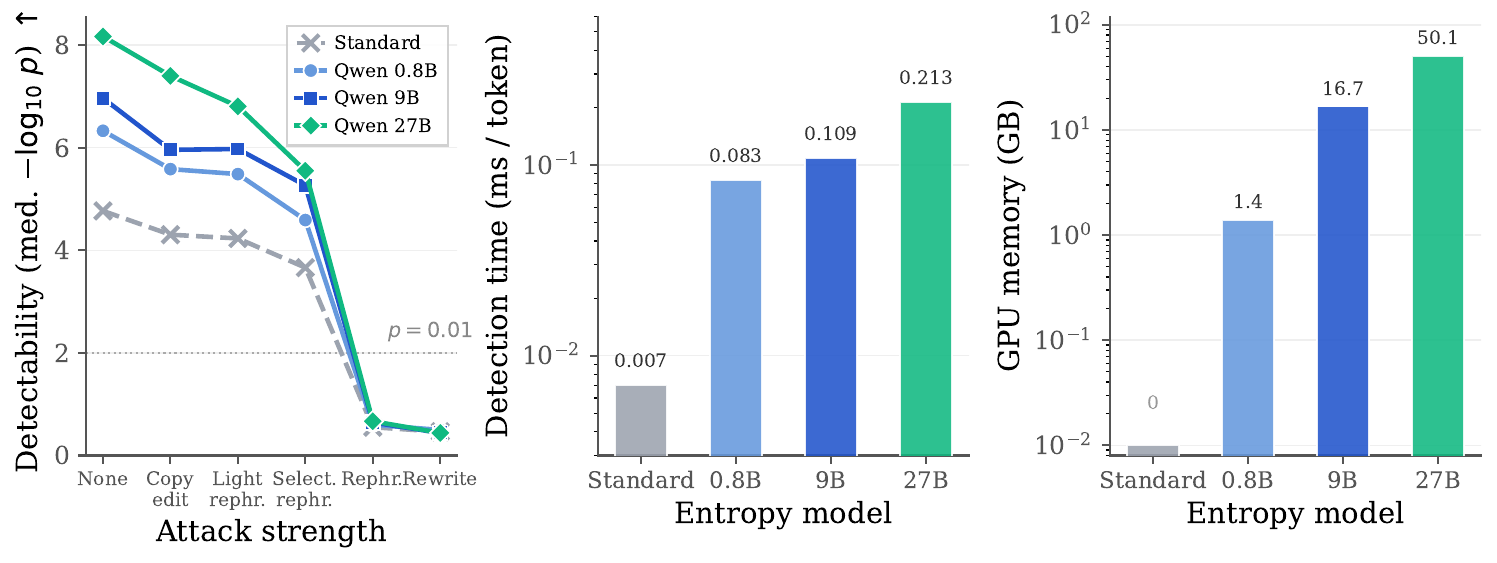}
  \caption{
      Entropy-aware detection performance and computational costs. \textbf{(Left)} Detectability under varying attack strengths. The highly efficient 4-bit 0.8B model boosts base detectability by ${\sim}1.3$ orders of magnitude, capturing much of the theoretical maximum boost ($+3.4$) provided by the 27B generation model. \textbf{(Middle)} Detection time per token (log scale). \textbf{(Right)} Peak GPU memory allocation (log scale). The 4-bit 0.8B model offers an excellent trade-off, recovering most of the watermark signal while requiring $35\times$ less memory.
  }
  \label{fig:entropy_detection}
\end{figure}

\paragraph{Detection Efficiency and Proxy Scaling.}
For detection, we evaluate the optimal proxy model size for entropy weighting (Section~\ref{subsec:entropy_detection}) across varying attack strengths (\autoref{fig:entropy_detection}). 
Standard unweighted detection is highly efficient ($0.007$\,ms/token, $0$\,GB VRAM overhead) but yields a baseline median $-\log_{10}p$ of $4.8$. 
Entropy weighting with the full 27B model significantly boosts this score to $8.2$, but incurs massive overhead ($50.1$\,GB VRAM, $0.213$\,ms/token). 
However, the 4-bit quantized 0.8B model emerges as the optimal practical choice: it achieves near-parity detectability ($6.2$) and scales identically against robust attacks, while requiring only $1.4$\,GB VRAM and $0.115$\,ms/token.

\paragraph{MTP Speculative Decoding.}\label{sec:exp5}Multi-token prediction (MTP) speculative decoding accelerates inference via lightweight draft heads that propose multiple tokens in parallel~\citep{gloeckle2024better}.
TextSeal natively supports this by assigning key~$A$ to draft-accepted tokens and key~$B$ to target-resampled tokens. Consequently, the mixing parameter $\alpha$ dynamically matches the empirical acceptance rate.
We evaluate Qwen~3.5 (2B, 9B, 27B) generating 400-token ELI5 answers under three conditions: standard MTP, MTP with TextSeal, and autoregressive TextSeal. 
As shown in Figure~\ref{fig:mtp_speculative}, MTP draft acceptance rates remain identical ($29$--$46\%$) with and without TextSeal, confirming the dual-key approach is perfectly distortion-free and preserves all speculative efficiency gains. 
While MTP TextSeal's detection signal is slightly lower than standard TextSeal due to key mixing dilution ($\alpha < 1$), entropy weighting easily recovers strong significance well above the $p{=}0.01$ threshold. 
Perplexity remains identical across all conditions and model sizes, confirming that the dual-key watermark introduces no quality degradation.

\begin{figure}[b!]
    \centering
    \includegraphics[width=0.8\linewidth]{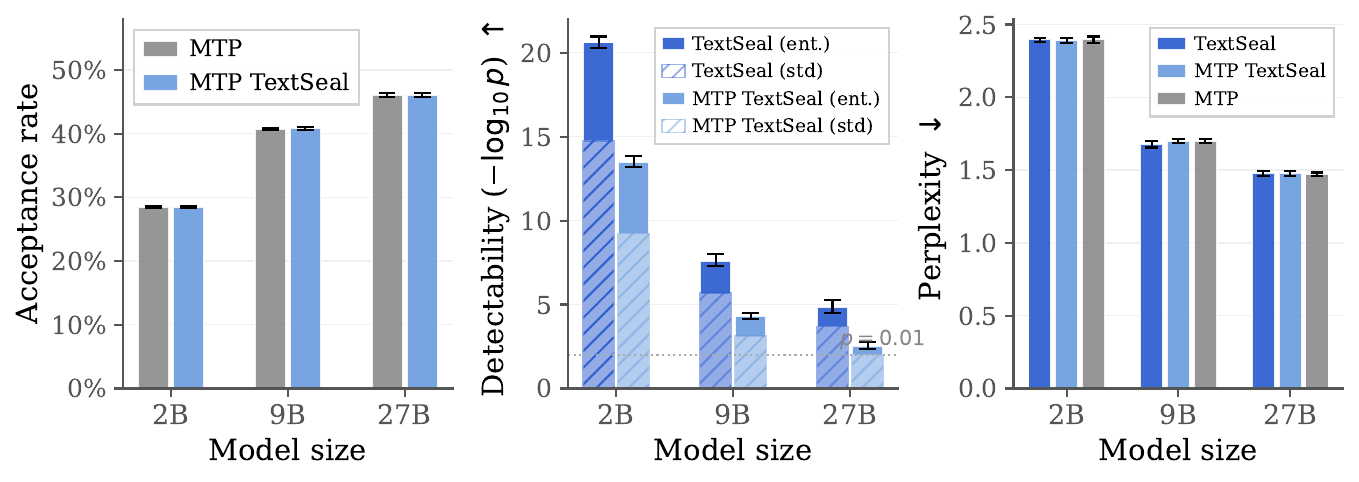}
    \caption{%
        MTP speculative decoding with TextSeal watermarking across Qwen~3.5 model sizes (2B, 9B, 27B) at temperature~$0.8$, top-$p$~$0.9$.
        \emph{Left:}~Draft acceptance rate is unchanged by dual-key watermarking, confirming zero overhead.
        \emph{Center:}~Both TextSeal and MTP TextSeal are well above the $p{=}0.01$ detection threshold; solid bars show entropy-weighted detection, hatched bars show standard detection. The modest gap between TextSeal and MTP TextSeal is explained by the key-mixing parameter $\alpha < 1$.
        \emph{Right:}~Perplexity is identical across all conditions, confirming that watermarking is distortion-free.
    }
    \label{fig:mtp_speculative}
\end{figure}

\section{Watermark Radioactivity: Detecting Distillation via Learnability}
\label{sec:learnability}

A watermark is \emph{radioactive}~\citep{sander2024watermarking} if, when a model is trained on watermarked data, it inherits a detectable token bias.
This enables a powerful application beyond text provenance: detecting whether a competitor has distilled your model's outputs into their own.

\paragraph{Setup.}
We distill DeepSeek-R1-Distill-Qwen-14B~\citep{guo2025deepseek} (teacher) into Qwen2.5-3B~\citep{qwen25} (student) on 5{,}000 curated problems from OpenR1-Math-220k~\citep{openr1}.
The teacher generates reasoning traces under four watermark schemes: Gumbel-Max~\citep{aaronson2023watermarking}, TextSeal ($\alpha{=}0.1$), SynthID depth 10~\citep{dathathri2024scalable}, and an unwatermarked control (all with watermark windows of size 3).
Following \citet{muennighoff2025s1simpletesttimescaling}, we retain traces only if they close their \texttt{</think>} block, contain a \verb|\boxed{}| answer (when required), match the reference solution under \texttt{math\_verify}, and have no 100-character span recurring $\geq 3$ times.
We also remove problems that the base student already solves correctly, so the distillation set only includes traces that teach the student something new.
We then apply LoRA fine-tuning on the remaining traces.

\paragraph{Detection methodology.}
To test whether the watermark transferred, we use the open-model radioactivity test~\citep{sander2024watermarking}.
We feed each training trace into the student using \emph{teacher forcing} (providing the ground-truth prefix at each position) and record the student's top-1 prediction.
If the student internalized the watermark bias during training, its predictions should be skewed toward high-PRF tokens.
We score each prediction with the watermark PRF and aggregate into a $p$-value.
To get a statistically valid test, we deduplicate at two levels.
\emph{Within each trace,} each watermark context window $\mathbf{w}_t$ is scored at most once: if the same $k$-gram appears more than once in a trace, we score the student's prediction only at the first occurrence. This is needed to avoid spurious signal: a high-PRF token already inside the input context can be copied by the student through attention rather than retrieved from internalized watermark bias.
\emph{Across traces,} we further deduplicate \emph{(context window, predicted token)} tuples globally so that repeated tuples are counted only once. The PRF is deterministic in $(v, \mathbf{w}, \sk)$, so duplicated tuples produce identical scores and would violate independence in the statistical test.
After deduplication, this yields ${\sim}1.4$--$2.2$M unique scored tokens per method (full setup in \autoref{app:learnability_details}).

\begin{figure}[b!]
    \centering
    \includegraphics[width=\linewidth]{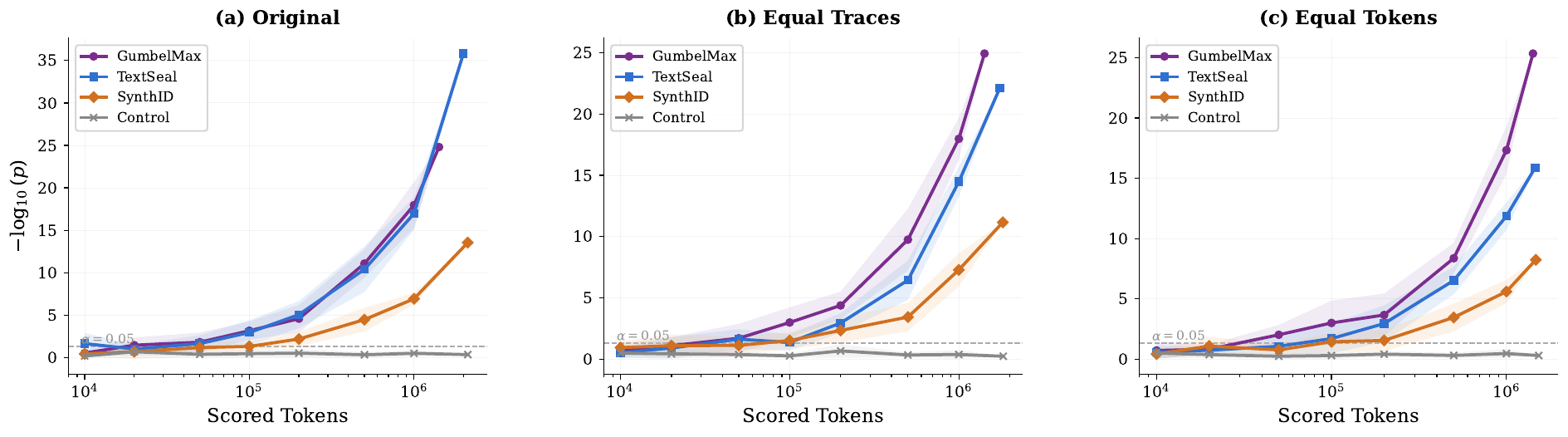}
    \caption{
        \textbf{Watermark radioactivity through distillation.}
        Detection power ($-\log_{10}(p)$) vs.\ number of unique scored tokens under three conditions: original traces, equal-trace control ($1{,}991$ each), and equal-token control (${\sim}15.1$M chars each).
        TextSeal achieves the strongest signal under the original setup thanks to retaining more traces, while Gumbel-Max dominates under controlled conditions, confirming its stronger per-token signal.
    }
    \label{fig:learnability}
\end{figure}

\begin{table}[t!]
\centering
\small
\caption{
    \textbf{Teacher trace quality and detectability.}
    The teacher generates 5{,}000 traces per method; pass rate is the fraction retained by the four-stage quality filter.
    Teacher $-\log_{10}(p)$ reports the mean watermark detection power across individual teacher traces.
    Accuracy is measured on GSM8K (1{,}319 problems, greedy decoding); the baseline is the pre-training Qwen2.5-3B.
    \textsuperscript{\dag}TextSeal uses entropy-weighted scoring.
}
\label{tab:learnability_impact}
\begin{tabular}{lccccc}
\toprule
\textbf{Method} & \textbf{Retained} & \textbf{Pass} & \textbf{Teacher} & \textbf{GSM8K} & \textbf{$\Delta$ vs} \\
 & \textbf{Traces} & \textbf{Rate} & $-\log_{10}(p)$ & \textbf{Acc} & \textbf{Base} \\
\midrule
Base Model (Qwen2.5-3B) & --- & --- & --- & 64.5\% & --- \\
Gumbel-Max   & 1{,}991 & 39.8\% & 14.89 & 78.8\% & +14.3 \\
TextSeal     & 2{,}352 & 47.0\% & 33.15\textsuperscript{\dag} & 79.9\% & +15.4 \\
SynthID      & 2{,}408 & 48.2\% & 14.39 & 75.2\% & +10.7 \\
Control      & 2{,}400 & 48.0\% & 0.39 & 75.5\% & +11.0 \\
\bottomrule
\end{tabular}
\end{table}

\paragraph{Results.}
\autoref{fig:learnability} shows that all three watermarks reliably transfer through distillation, with detection power far exceeding the significance threshold.
Under the original setup (each method uses all its retained traces), TextSeal achieves the strongest signal thanks to higher data volume.
Once data volume is equalized (controlled conditions in \autoref{fig:learnability}b,c), Gumbel-Max dominates, confirming a stronger per-token signal via deterministic argmax; TextSeal achieves comparable overall detectability by retaining more training data.
All distilled students substantially improve over the base model (+10--15\% on GSM8K), and distilling on watermarked traces does not lead to significant changes compared to the unwatermarked control.

\paragraph{Controlled comparisons.}
To rule out training data volume as a confound, we repeat the experiment under two controlled conditions:
(i)~\emph{equal traces}, where each method uses exactly $1{,}991$ traces (the Gumbel-Max minimum, randomly subsampled for the other methods); and
(ii)~\emph{equal tokens}, where each method is allocated ${\sim}15.1$M characters.
Under equal traces, TextSeal achieves the highest student accuracy ($81.0\%$), followed by SynthID and Control ($78.8\%$ each) and Gumbel-Max ($77.7\%$).
Under equal tokens, the spread narrows ($79.7\%$/$78.6\%$/$79.6\%$/$77.6\%$ for TextSeal/Gumbel-Max/SynthID/Control).
Detection remains strong under both controls, validating that the conclusions of \autoref{fig:learnability} are not artifacts of unequal training data volume.

\paragraph{Entropy weighting ablation.}
For TextSeal we use $\sqrt{\hat{H}}$ entropy-aware scoring by default (\autoref{subsec:entropy_detection}).
\autoref{fig:learnability_entropy} compares eight weighting functions in the same teacher-forcing setup, spanning normalized-entropy transforms (Sqrt, Log, Linear, Tanh of $\hat{H}_i$) and raw entropy power functions ($H_i^\beta$ for $\beta \in \{0.5, 1.0, 1.5\}$).
The concave $\sqrt{\hat{H}}$ weighting achieves the strongest detection ($p = 3.7 \times 10^{-110}$), improving over the uniform baseline ($p = 2.1 \times 10^{-84}$) by more than $25$ orders of magnitude.
Concave functions outperform linear and superlinear alternatives because they moderately upweight high-entropy positions---where the watermark has more room to influence token selection (\autoref{prop:expected_score})---without over-amplifying noisy extreme-entropy tokens.

\begin{figure}[t!]
    \centering
    \includegraphics[width=0.6\linewidth]{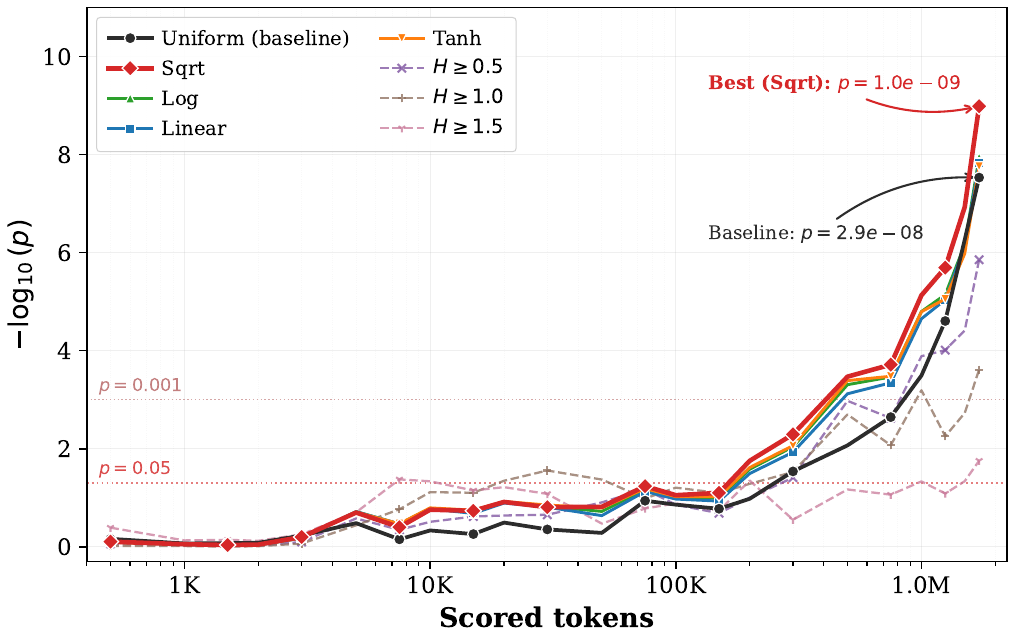}
    \caption{
        \textbf{Entropy-aware scoring for watermark learnability detection.}
        Detection power ($-\log_{10}(p)$) vs.\ number of unique scored tokens in the teacher-forcing radioactivity test for TextSeal ($\alpha{=}0.1$), comparing different entropy weighting functions $w_i^{\text{ent}} = f(H_i)$ against a uniform (unweighted) baseline.
        The concave $\sqrt{\hat{H}}$ weighting achieves the strongest signal ($p = 3.7 \times 10^{-110}$), improving over the uniform baseline ($p = 2.1 \times 10^{-84}$) by more than $25$ orders of magnitude.
    }
    \label{fig:learnability_entropy}
\end{figure}

The full benchmark accuracy numbers and further details are given in~\autoref{app:learnability_details}.

\section{Conclusion}

We introduced TextSeal, a distortion-free watermark for LLMs that achieves state-of-the-art detectability through dual-key generation, entropy-weighted detection, and localized multi-region search.
TextSeal strictly dominates SynthID on the diversity-detectability frontier, preserves model performance across 12 benchmarks, supports speculative decoding and MTP, and transfers through distillation for radioactive tracing.

\paragraph{Limitations.}
Like all distortion-free sampling watermarks, TextSeal trades diversity, not quality, for detectability.
While this trade-off is invisible to users who observe a single generation, it may affect workflows that rely on diverse outputs (best-of-$N$ reranking, creative brainstorming).
In practice, modern reasoning models trained with RL already exhibit collapsed entropy, limiting the marginal diversity loss; quantifying this across model families remains open.

\clearpage
\bibliographystyle{ieeenat_fullname}
\bibliography{references}

\clearpage
\beginappendix
\section{More Technical Details on the Methods}
\label{app:method_details}

\subsection{Hash Function Implementation}
\label{app:hash}

The PRF takes as input the candidate token $x$, a context window $\mathbf{w} = (w_1, \ldots, w_k)$ of $k$ token IDs, and the secret key $\sk$ (all of them are integers), and outputs a random integer in $[0, M)$.

We compute the hash as follows:
\begin{align}
h'(x, \mathbf{w}, \sk) & = \left( p_2 \cdot x + \sum_{i=1}^{k} w_i \cdot q_i + p_3 \cdot \sk \right) \cdot p_4, \\
h (x, \mathbf{w}, \sk) & = \text{XORShift}(h'(x, \mathbf{w}, \sk)) \mod M,
\end{align}
where $q_1, \ldots, q_k$ are distinct large primes (to ensure that different orderings of the same tokens produce different values), and $p_2, p_3, p_4$ are additional primes. The first result $h'$ undergoes XOR-shift for better bit dispersion: $h = (h' \cdot p_{\text{mix}}) \oplus ( (h' \cdot p_{\text{mix}}) \gg s)$, where $p_{\text{mix}}$ is a mixing prime and $s$ is a shift constant.

Finally, we normalize to obtain the uniform pseudo-random value:
$$ u = \frac{h (x, \mathbf{w}, \sk)}{M} \in [0, 1] $$

\subsection{Details on SynthID-Text Evaluation}
\label{app:synthid_details}

In our experiments, we evaluate \emph{SynthID-Text} \citep{dathathri2024scalable} as the state-of-the-art generation-time, distortion free and non deterministic watermark.
While traditional methods (such as Gumbel-Max or Soft Red List) apply a single, global shift to the logit distribution, SynthID-Text embeds its signal through a multi-layered \emph{Tournament sampling} mechanism. 

\paragraph{Tournament Generation.} 
At each generation step $t$, the method seeds a pseudo-random function using the preceding $k$ tokens (the context window).
Using this seed, the vocabulary $\mathcal{V}$ is pseudo-randomly partitioned into a tournament structure with $m$ layers. 
At each layer $l \in \{1, \dots, m\}$, a pseudo-random $g$-value $g_{t,l}$ is computed. Instead of a single binary split, SynthID-Text iteratively reshapes the target LLM's probability distribution across these $m$ layers.
Tokens that consistently win their tournament matches (i.e., those assigned high $g$-values across multiple layers) see their sampling probabilities exponentially increased.
By distributing the watermark across multiple layers, SynthID-Text preserves text quality while embedding a robust signal.
In our implementation, we follow the authors' specification for a binary random function (Bernoulli $g$-value distribution) to construct this tournament.

\paragraph{Why we avoid SynthID's Bayesian detector.}
The original SynthID-Text framework proposes a Bayesian neural network (logistic regression or MLP) trained on a representative dataset to estimate posterior probabilities $P(w|g)$.
We avoid this approach for several reasons.

\emph{(i)~No false-positive-rate guarantee.}
A Bayesian posterior score has no frequentist calibration: there is no principled way to set a decision threshold that guarantees, for example, at most one false accusation in $10^4$ documents.
This is essential for any legal or regulatory use of watermark detection, where a false positive can constitute a wrongful accusation of AI generation.

\emph{(ii)~Distribution dependence.}
The trained classifier learns the joint distribution of token scores and watermark presence from its training corpus.
Deploying on a different model, domain, language, or decoding strategy invalidates these learned posteriors; in practice, we observed that the Bayesian detector degrades sharply on out-of-domain text, requiring retraining for every new deployment setting.

\emph{(iii)~Incompatibility with multiple-testing correction.}
Localized detection (\autoref{subsec:localization}) requires evaluating thousands of candidate windows and applying Bonferroni correction to control the family-wise error rate.
This demands a well-calibrated null distribution for each window, which a learned classifier cannot provide.
The Bayesian scores are not $p$-values and cannot be combined or corrected in a statistically valid manner.

\emph{(iv)~Opacity and reproducibility.}
A learned classifier is a black box whose decision boundary cannot be formally audited.
For provenance claims that may carry legal weight, a closed-form statistical test with an analytically derived null distribution is far more defensible.
Moreover, the Bayesian detector is not open-sourced, and despite following the specification in the original supplementary material, we were unable to reproduce comparable results, making fair comparison infeasible.

\paragraph{Our frequentist alternative.}
To ensure a fair, threshold-independent comparison, we implement a mathematically rigorous frequentist detection pipeline.
At detection time, for a given token $x_t$ and its context, we reconstruct the PRF-seeded tournament and extract the sequence of $m$ layer-wise $g$-values $g_{t,1}, \dots, g_{t,m}$. 
Because earlier layers in the tournament contribute more watermarking evidence than later layers, we compute a \emph{Weighted Mean Score} for the token:
\begin{equation}
s_t = \sum_{l=1}^{m} \alpha_l g_{t,l}
\end{equation}
where $\alpha_1 \ge \dots \ge \alpha_m \ge 0$ are linearly decaying weights. Over a sequence of $N$ valid tokens, we sum the scores to obtain a test statistic $S = \sum_{t=1}^N s_t$. 
Under the null hypothesis $\mathcal{H}_0$ (unwatermarked text), the $g$-values follow the unwatermarked uniform or Bernoulli distribution.
We analytically compute the mean $\mu_0$ and variance $\sigma_0^2$ of the weighted sum under $\mathcal{H}_0$.
We then compute a final Z-score for the sequence:
\begin{equation}
Z_{\text{SynthID}} = \frac{S - N\mu_0}{\sigma_0 \sqrt{N}}
\end{equation}
The significance is given by the standard normal survival function $p = 1 - \Phi(Z_{\text{SynthID}})$.

\subsection{Other Watermark Schemes}\label{app:method_overview}

We describe below the other watermarking schemes referenced in this work.

\paragraph{Green-list/Red-list.}
\citet{kirchenbauer2023watermark} modify the logit vector based on the watermark context window and secret key $\sk$.
A token $v$ is classified as ``green'' if $\mathrm{PRF}(v, \mathbf{w}, \sk) < \gamma$ (typically $\gamma{=}0.5$), and its logit is incremented by $\delta$: $\tilde{\ell}_v = \ell_v + \delta$ for green tokens, $\tilde{\ell}_v = \ell_v$ otherwise.
Detection counts green tokens and performs a binomial test. This method is \emph{not} distortion-free: the additive bias alters every generation.

\paragraph{MorphMark.}
\citet{wang2025morphmark} adaptively adjust watermark strength based on context.
Let $P_G = \sum_{v \in \text{GreenList}} p_v$ be the total probability mass on green tokens.
If $P_G \leq p_0$ (a threshold), no watermark is applied; otherwise, probabilities are rescaled with an adaptive boost factor $r = \min(\kappa P_G, 1)$.
This reduces distortion compared to vanilla green-red but is still \emph{not} distortion-free.

\paragraph{DiPMark.}
\citet{wu2023dipmark} introduce a distortion-free variant of green-red watermarks using a pseudorandom permutation $\pi$ (seeded by context and $\sk$) to reorder tokens before applying a bias.
The bias preserves the original distribution in expectation over the randomness of $\pi$.

\paragraph{WaterMax.}
\citet{giboulot2024watermax} generate several candidate chunks from the original LLM distribution and select outputs with the highest watermark score.
This is distortion-free by construction but requires multiple generations per query, making it impractical for production.

\subsection{Radioactivity Test Protocol}\label{app:radioactivity_protocol}

We detail the radioactivity test methodology from~\citet{sander2024watermarking, sander2025detecting}.

\paragraph{Teacher-forcing setup.}
We feed the watermarked training traces into the suspect (student) model using teacher forcing: at each position $t$, the model receives the ground-truth prefix $x_{<t}$ from the watermarked trace and produces a prediction.
Let $\hat{x}_t = \arg\max_{v \in \V} P_\theta(v \mid x_{<t})$ denote the student's top-1 prediction at step $t$.
The key insight is that teacher forcing isolates the model's learned token preferences from confounding factors like sampling noise, requiring only a single forward pass over existing traces rather than expensive autoregressive generation.

\paragraph{Test statistic.}
We score each prediction using the watermark's PRF: $R_t = \text{PRF}(\hat{x}_t, \mathbf{w}_t, \sk)$, where $\mathbf{w}_t$ is the context window of teacher tokens preceding position $t$.
If the student internalized the watermark bias during training, its top-1 predictions will be systematically skewed toward high-PRF tokens, producing a significant test statistic.

\paragraph{Deduplication.}
We deduplicate at two levels, each for a different reason:
(i)~\emph{within each trace}, each context window $\mathbf{w}_t$ is scored only once. This is necessary because the teacher's watermarked tokens appear in the student's input context during teacher forcing: if an n-gram that was biased toward high-PRF values appears multiple times, the student might simply copy it from context rather than predicting it from internalized preferences, creating a false signal~\citep{sander2024watermarking};
(ii)~\emph{across traces}, all (context window, predicted token) pairs are pooled and deduplicated, because the PRF is deterministic and shared (context, token) tuples across different training examples would yield identical scores, violating independence~\citep{fernandez2023three}.
After deduplication, under $\Hnull$ (the student is unaware of $\sk$), the scores are independent and follow their null distribution, enabling exact $p$-value computation.

\subsection{Formal Definitions of Non-Distortion}
\label{app:distortion_definitions}

Following \citet{dathathri2024scalable}, we provide the formal definitions of the two levels of non-distortion used throughout the paper.
Let $\Delta_\V$ denote the probability simplex over the vocabulary $\V$, let $\mathcal{R}$ be the space of random seeds, and let $\mathcal{K}$ be the space of secret keys.

\begin{definition}[Single-token non-distortion~{\citep[Def.~16]{dathathri2024scalable}}]
\label{def:single_token_distortion_free}
A sampling algorithm $S : \Delta_\V \times \mathcal{R} \to \V$ is \emph{single-token non-distortionary} if for any probability distribution $\vec{p} \in \Delta_\V$ and token $x \in \V$:
$$\E_{r \sim \mathrm{Unif}(\mathcal{R})} \left[ \Prob(S(\vec{p}, r) = x) \right] = p_x.$$
\end{definition}

Marginalizing over the random seed, the sampling algorithm produces each token with exactly its original LLM probability.
This is a property of the sampling algorithm alone (e.g., Gumbel-max or two-sample Tournament sampling), independent of how the seed is generated across time steps.

\begin{definition}[Single-sequence non-distortion~{\citep[Def.~20 with $K{=}1$]{dathathri2024scalable}}]
\label{def:single_sequence_distortion_free}
A watermarking scheme $P_{\mathrm{wm}}$ is \emph{single-sequence non-distortionary} if, for any prompt $\vec{x}$ and response $\vec{y} \in \V^*$:
$$\E_{k \sim \mathrm{Unif}(\mathcal{K})} \left[ P_{\mathrm{wm}}(\vec{y} \mid \vec{x}, k) \right] = p_{\mathrm{LM}}(\vec{y} \mid \vec{x}).$$
\end{definition}

This is strictly stronger than single-token non-distortion: it requires the \emph{joint} distribution over a full response to match the original LLM, not just the per-token marginals.
A single-token non-distortionary sampler paired with a sliding-window seed generator violates this property whenever the same context window repeats (producing deterministic rather than stochastic outputs).
Repeated context masking~\citep{dathathri2024scalable, kuditipudi2023robust} restores single-sequence non-distortion by falling back to unwatermarked sampling on repeated contexts (\autoref{rem:deduplication}).

\section{Gumbel-max proofs}\label{app:gumbel_proofs}

The following results were presented by \citet{aaronson2023watermarking} and formalized by \citet{fernandez2023three}. 
Some elements of these proofs are used later, so we restate them here.
An overview of the Gumbel-max generation scheme is presented in \autoref{fig:gumbel-max}.

\begin{figure*}[b!]
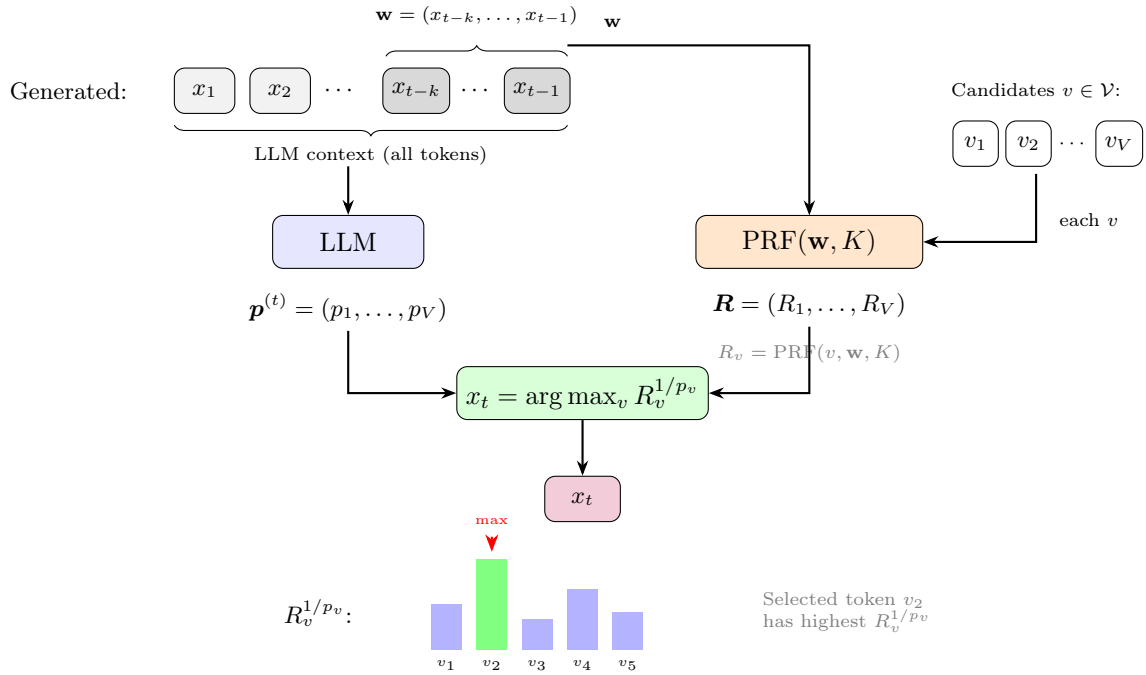

\centering
\figGumbelMaxTikz
\caption{Standard Gumbel-Max watermarking (see \autoref{sec:background}). The LLM uses all previous tokens to predict probabilities, while the PRF uses only the last $k$ tokens (watermark context $\mathbf{w}$) to generate pseudo-random values $R_v$ for each candidate token $v$. The token maximizing $R_v^{1/p_v}$ is selected.}
\label{fig:gumbel-max}
\end{figure*}

\begin{proposition}[Sampling probability, restated from \autoref{prop:sampling}]
Consider a discrete distribution $\vec{p}=(p_1,\ldots,p_V)$
and $V=|\V |$ random variables $\vec{R} = (R_1,\ldots,R_V)$ s.t. $R_v\overset{iid}{\sim}\mathcal{U}_{[0,1]}$.
Let $V^\star = \arg \max_v R_v^{1/p_v}$.
Then: $\Prob(V^\star=v) = p_v$.
\end{proposition}

\begin{proof}[Proof of \autoref{prop:sampling}]
For any $v \in \V$, $R_v\overset{iid}{\sim}\mathcal{U}_{[0,1]}$ so, $- \ln(R_v)$ follows an \href{https://en.wikipedia.org/wiki/Exponential_distribution}{exponential distribution} $\mathcal{E}(1)$.
Let $Z_v := -\frac{1}{p_v} \ln(R_v)$. By construction, $Z_v\sim\mathcal{E}(p_v)$, with density $f_{Z_v}(z) = p_v e^{-p_v.z}$.
We now have:
\begin{equation}
V^\star = \arg \max_v R_v^{\frac{1}{p_v}} = \arg \min_v Z_v.
\end{equation}
A well known result about exponential laws is that:
\begin{eqnarray}
\underline{Z}  &=& \min_v Z_v \sim \mathcal{E}\left(\sum_v p_v\right)=\mathcal{E}\left(1\right),\\ 
\Prob(V^\star=v) &=& \frac{p_v}{\sum_j p_j}  = p_v.
\end{eqnarray}
This shows that for a given secret vector $\vec{r}$, the watermarking chooses a word which may be unlikely (low probability $p_{V^\star}$). 
Yet, on expectation over the secret keys, i.e., over r.v. $\vec{R} = (R_1, \ldots, R_V)$, the distribution of the chosen token follows the distribution given by the LLM.
\end{proof}

\begin{corollary}[Restated from \autoref{cor:beta}]
Conditionally on $V^\star = v$, $R_{V^\star} \sim \text{Beta}(1/p_v, 1)$.
\end{corollary}

\begin{proof}[Proof of \autoref{cor:beta}]
From the proof above, $\underline{Z} = \min_v Z_v \sim \mathcal{E}(1)$ and $V^\star = \arg\min_v Z_v$. A standard property of competing exponentials is that the identity of the winner is independent of the winning time: $V^\star \perp \underline{Z}$. Conditioning on $V^\star = v$, we therefore still have $\underline{Z} \sim \mathcal{E}(1)$, and:
\begin{equation}
\underline{Z}  = Z_v = -\frac{1}{p_v} \ln(R_v) \sim \mathcal{E}(1),
\end{equation}
which gives $R_v = e^{-p_v E}$ with $E \sim \mathcal{E}(1)$, with p.d.f. $f_{R_v}(r) = \frac{r^{1/p_v - 1}}{p_v}$.
Therefore, $R_v \mid V^\star = v \sim \text{Beta}(1/p_v, 1)$.
\end{proof}

\begin{proposition}[Expected score under $\Halt$, restated from \autoref{prop:expected_score}]
Under $\Halt$ (text is watermarked),
$\displaystyle \mathbb{E}(S_T) \geq T +  \left( \frac{\pi^2}{6} -1 \right) H_T$,
where $H_T = - \sum_{t=1}^T p_t\ln(p_t)$ is the entropy of the completion.
\end{proposition}

\begin{proof}[Proof of \autoref{prop:expected_score}]
From the corollary above, $R_t=\exp(-p_t E)$ with $E\sim \mathcal{E}(1)$, so:
\begin{align*}
    \mathbb{E}(S) &= - \mathbb{E} \left[ \sum_{t=1}^T \ln (1-\exp(-p_t E)) \right] \\
    &= - \sum_{t=1}^T \int_0^\infty \ln (1-e^{-p_t x}) e^{-x} dx \\
    &= - \sum_{t=1}^T \int_0^1 \frac{1}{p_t} r^{1/p_t-1} (-\ln ( 1 - r)) dr  \\ 
    & \text{ \qquad (by change of variable $x = -1/p_t \ln (r) $ )} 
\end{align*}
Then, using integration by parts with $u = 1 - r^{1/p_t}$ and $v = \ln(1-r)$, the integral becomes:
\begin{align*}
    -\int_0^1 \frac{1}{p_t} r^{1/p_t-1} \ln ( 1 - r) dr &= \int_0^1 \frac{1-r^{1/p_t}}{1-r} dr = \mathcal{H}_{1/p_t}
\end{align*}
where $\mathcal{H}_{z}$ is the $z$-th \href{https://en.wikipedia.org/wiki/Harmonic_number}{harmonic number}
also defined as $\mathcal{H}_{z} = \sum_{n=1}^\infty \frac{1}{n} - \frac{1}{n+z}$.
Therefore, we have:
\begin{align*}
    -\int_0^1 \frac{1}{p_t} r^{1/p_t-1} \ln ( 1 - r) dr &= 
        \sum_{n=1}^\infty \frac{1}{n} - \frac{1}{n+1/p_t} \\
    &= 1 + \sum_{n=1}^\infty \frac{1}{n+1} - \frac{1}{n+1/p_t}.
\end{align*}
Now, $\forall n\in \mathbb{N^\star}$, we have:
\begin{align*}
   (n+1)^2 \left(\frac{1}{n+1} - \frac{1}{n+1/p_t}\right) &= \frac{(n+1)(n+1/p_t) - (n+1)^2}{n + 1/p_t} \\
    &=  \frac{1+n}{1/p_t + n} \left( 1/p_t -1\right) \\
    &\geq -  \frac{1+n}{1/p_t + n} \ln(p_t) \\
    &\geq -  \, p_t \ln(p_t).
\end{align*}
Therefore, by summing over all $t\in [1,T]$,
\begin{align*}
    \mathbb{E}(S) &\geq T +  \left(\sum_{n=1}^\infty \frac{1}{(n+1)^2}\right)\left(\sum_{t=1}^T- p_t\ln(p_t) \right) \\
    &=T +  \left( \frac{\pi^2}{6} -1 \right) H_T.
\end{align*}
\end{proof}

\section{Proofs on Diversity Schemes for Gumbel Max}
\label{app:diversity_bounds}

\begin{figure*}[b!]
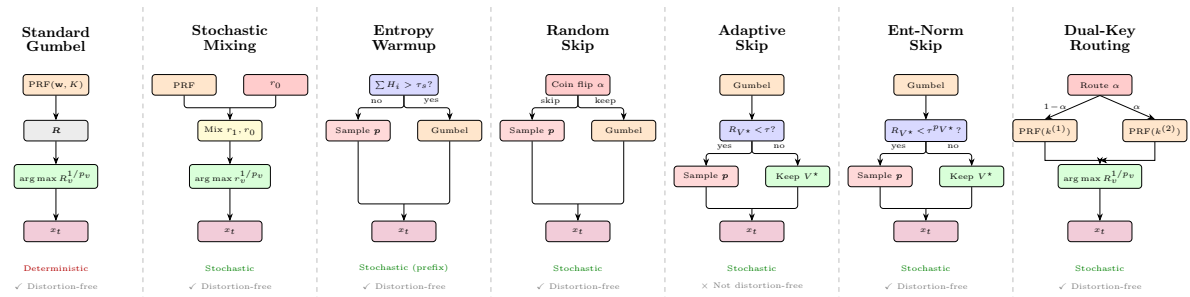

\centering
\figDiversityMechanisms
\caption{Overview of diversity mechanisms for Gumbel-Max watermarking. Each column shows how the token $x_t$ is generated. All methods except Adaptive Skip preserve the distortion-free property. The key distinction lies in \emph{where} randomness is injected: in the PRF value (Mixing), in the decision to watermark (Skip variants, Warmup), or in the key selection (Dual-Key Routing).}
\label{fig:diversity-mechanisms}
\end{figure*}

We derive bounds on the expected detection score $\E[S_T]$ under $\Halt$ for each diversity strategy described in \autoref{subsec:diversity} and illustrated in \autoref{fig:diversity-mechanisms}.
All bounds decompose as the standard Gumbel bound (\autoref{prop:expected_score}) plus a correction term capturing the cost of the diversity mechanism.

Recall that under the standard Gumbel scheme, $R_t \sim \text{Beta}(1/p_t, 1)$ and the expected per-token score is $\E[s_t] = \mathcal{H}_{1/p_t}$, leading to $\E[S_T] \geq T + (\frac{\pi^2}{6} - 1) H_T$. In each case below, some tokens are either unwatermarked or have a modified distribution of $R_t$. For unwatermarked tokens, $R_t \sim \mathcal{U}[0,1]$ and $\E[s_t] = 1$.

\subsection{Dual-Key Routing}

Dual-key routing (\autoref{subsec:dual_key}) maintains two secret keys $k^{(1)}$ and $k^{(2)}$. At each generation step, key $k^{(1)}$ is selected with probability $1-\alpha$ and $k^{(2)}$ with probability $\alpha$. The token is produced via Gumbel-Max using the selected key. Detection aggregates scores from both keys: $s_i = (1-\alpha) \cdot s_i^{(1)} + \alpha \cdot s_i^{(2)}$.

\begin{proposition}[Bound on expected score under dual-key routing, single-key detection]
\label{prop:twokey_score}
Under dual-key routing with parameter $\alpha \in [0,1]$ (key $k^{(1)}$ selected with probability $1-\alpha$, key $k^{(2)}$ with probability $\alpha$), detection using a single key $k^{(1)}$ yields:
\begin{equation}
    \E[S_T^{(1)}] \geq T + (1-\alpha)\left( \frac{\pi^2}{6} -1 \right) H_T
\end{equation}
\end{proposition}

\begin{proof}
At each step $t$, key $k^{(1)}$ is selected with probability $1-\alpha$ and key $k^{(2)}$ with probability $\alpha$. For detection using key $k^{(1)}$:
\begin{itemize}
    \item With probability $1-\alpha$: the PRF value $R_t^{(1)}$ is the one used for generation, so $R_t^{(1)} \sim \text{Beta}(1/p_t, 1)$ and $\E[s_t^{(1)}] = \mathcal{H}_{1/p_t}$.
    \item With probability $\alpha$: the token was generated using key $k^{(2)}$, so $R_t^{(1)}$ is independent of the generation process. It is effectively uniform and $\E[s_t^{(1)}] = 1$.
\end{itemize}
Summing over $T$ tokens:
$$ \E[S_T^{(1)}] = (1-\alpha) \sum_{t=1}^T \mathcal{H}_{1/p_t} + \alpha T $$
Applying the standard bound (\autoref{prop:expected_score}) to $\sum_t \mathcal{H}_{1/p_t} \geq T + (\frac{\pi^2}{6}-1) H_T$:
$$ \E[S_T^{(1)}] \geq (1-\alpha)\left[T + \left(\frac{\pi^2}{6}-1\right) H_T\right] + \alpha T = T + (1-\alpha)\left(\frac{\pi^2}{6}-1\right) H_T $$
\end{proof}

This bound matches the random skip bound (\autoref{prop:periodic_score}) with $\alpha$ playing the role of the skip rate: from the perspective of a single-key detector, tokens generated with the other key look exactly like skipped tokens. The advantage of dual-key routing is that the aggregated score (\autoref{eq:early_fusion}) lets every token contribute signal from at least one key, as formalized below.

\begin{proposition}[Expected score under dual-key Early Fusion detection]
\label{prop:twokey_earlyfusion_score}
Under dual-key routing with parameter $\alpha \in [0,1]$, detecting with the aggregated score $s_i = (1-\alpha) \cdot s_i^{(1)} + \alpha \cdot s_i^{(2)}$ yields:
\begin{equation}
    \E[S_T] \geq T + \left(\alpha^2 + (1-\alpha)^2\right)\left( \frac{\pi^2}{6} -1 \right) H_T
\end{equation}
\end{proposition}

\begin{proof}
At each step $t$, key $k^{(1)}$ is active with probability $1-\alpha$ and key $k^{(2)}$ with probability $\alpha$. The aggregated per-token score is $T_t = (1-\alpha) s_t^{(1)} + \alpha s_t^{(2)}$.
\begin{itemize}
    \item If $k^{(1)}$ was used (prob.\ $1-\alpha$): $s_t^{(1)}$ has the watermarked distribution ($\E[s_t^{(1)}] = \mathcal{H}_{1/p_t}$) and $s_t^{(2)}$ is uniform ($\E[s_t^{(2)}] = 1$), giving $\E[T_t] = (1-\alpha)\mathcal{H}_{1/p_t} + \alpha$.
    \item If $k^{(2)}$ was used (prob.\ $\alpha$): $s_t^{(1)}$ is uniform ($\E[s_t^{(1)}] = 1$) and $s_t^{(2)}$ has the watermarked distribution ($\E[s_t^{(2)}] = \mathcal{H}_{1/p_t}$), giving $\E[T_t] = (1-\alpha) + \alpha\mathcal{H}_{1/p_t}$.
\end{itemize}
Taking expectation over the key choice:
\begin{align*}
    \E[T_t] &= (1-\alpha)\bigl[(1-\alpha)\mathcal{H}_{1/p_t} + \alpha\bigr] + \alpha\bigl[(1-\alpha) + \alpha\mathcal{H}_{1/p_t}\bigr] \\
    &= \bigl[(1-\alpha)^2 + \alpha^2\bigr]\mathcal{H}_{1/p_t} + 2\alpha(1-\alpha) \\
    &= \theta_R\, \mathcal{H}_{1/p_t} + (1-\theta_R)
\end{align*}
where $\theta_R = \alpha^2 + (1-\alpha)^2$. Summing over $T$ tokens:
$$ \E[S_T] = \theta_R \sum_{t=1}^T \mathcal{H}_{1/p_t} + (1-\theta_R)\, T $$
Applying the standard bound $\sum_t \mathcal{H}_{1/p_t} \geq T + (\frac{\pi^2}{6}-1) H_T$:
$$ \E[S_T] \geq \theta_R\left[T + \left(\frac{\pi^2}{6}-1\right) H_T\right] + (1-\theta_R)\, T = T + \theta_R\left(\frac{\pi^2}{6}-1\right) H_T $$
\end{proof}

Note that $\theta_R = \alpha^2 + (1-\alpha)^2 \leq 1-\alpha$ for $\alpha \leq 0.5$, so the expected score under Early Fusion is actually lower than under single-key detection (\autoref{prop:twokey_score}). The power advantage of Early Fusion comes not from a higher expected score, but from the reduced null variance ($\theta_R$ per token instead of $1$), which yields a better Z-score as shown below.

\subsubsection{Power Analysis: Early vs.\ Late Fusion}
\label{subsec:early_vs_late_proof}

We analyze the statistical power of the Early Fusion test compared to a classical single-key baseline and alternative Late Fusion strategies using the Z-score (Signal-to-Noise Ratio) separation:
$$Z = \frac{\mathbb{E}[S|H_1] - \mathbb{E}[S|H_0]}{\sqrt{\mathrm{Var}(S|H_0)}}$$

Assume a standard Gumbel-Max test where an unwatermarked token yields an expected score of $1$ with a variance of $1$, and a successfully watermarked token yields an expected score $\mu_w > 1$.

\paragraph{Single-Key Baseline.}
For a traditional single-key test with $n$ tokens, the expected score sum under $H_1$ is $n\mu_w$, and under $H_0$ is $n$. The null variance is $n$.
$$Z_{\text{base}} = \frac{n\mu_w - n}{\sqrt{n}} = \sqrt{n}(\mu_w - 1)$$

\paragraph{Early Fusion: Unweighted ($w=0.5$).}
For the unweighted test, the expected score per token is $\mathbb{E}[\bar{s}_i] = \frac{\mu_w + 1}{2}$ regardless of which key generated it. The null variance is $\mathrm{Var}(\bar{s}_i) = \frac{1^2 + 1^2}{2^2} = 0.5$.
$$Z_{\text{early}} = \frac{n(\frac{\mu_w + 1}{2}) - n}{\sqrt{0.5n}} = \frac{n(\mu_w - 1)}{2\sqrt{0.5n}} = \frac{\sqrt{n}(\mu_w - 1)}{\sqrt{2}}$$
Thus, $Z_{\text{early}} = \frac{1}{\sqrt{2}} Z_{\text{base}} \approx 0.707 Z_{\text{base}}$. This proves that unweighted Early Fusion is perfectly \textbf{invariant to $\alpha$}, but requires exactly twice as many tokens ($2n$) as the single-key baseline to reach the same statistical confidence.

\paragraph{Early Fusion: Optimal Weighted ($w=\alpha$).}
If the routing probability $\alpha$ is known (e.g., via speculative decoding acceptance rates) and we use optimal weights $w=\alpha$, the expected token score under $H_1$ becomes $\mathbb{E}[s_i] = \alpha(\alpha \mu_w + 1 - \alpha) + (1-\alpha)((1-\alpha)\mu_w + \alpha)$. Simplifying this and calculating the Z-score yields:
$$Z_{\alpha} = \sqrt{n}(\mu_w - 1) \sqrt{\alpha^2 + (1-\alpha)^2}$$
When $\alpha = 0.5$ (maximum diversity), $Z_{\alpha} = Z_{\text{early}} \approx 0.707 Z_{\text{base}}$. When $\alpha = 0.1$ (typical for draft model acceptance in speculative decoding), $Z_{\alpha} = \sqrt{0.1^2 + 0.9^2} Z_{\text{base}} \approx 0.905 Z_{\text{base}}$. This demonstrates that the weighted test recovers nearly 30\% of the statistical power lost to diversity when the generation rate is skewed.

\paragraph{Superiority over Late Fusion.}
We can now formally demonstrate why token-level aggregation outperforms independent per-key testing (late fusion). Late fusion evaluates each key's scores independently across the entire sequence ($S^{(1)} = \sum s_i^{(1)}$ and $S^{(2)} = \sum s_i^{(2)}$) and then combines their resulting p-values (e.g., via Fisher's method or by taking the minimum p-value).

Assuming without loss of generality that $\alpha \ge 0.5$, the expected signal for the dominant key over the null is $n\alpha(\mu_w - 1)$. The variance remains $n$. The statistical power of the combined Late Fusion test is ultimately bounded by the strongest independent signal it receives, which achieves at best:
$$Z_{\text{late}} \approx \frac{n\alpha(\mu_w - 1)}{\sqrt{n}} = \sqrt{n}\alpha(\mu_w - 1) = \alpha Z_{\text{base}}$$

To prove optimal Early Fusion natively dominates Late Fusion, we compare their Z-scores. We must show that $Z_{\alpha} > Z_{\text{late}}$, which simplifies to proving $\sqrt{\alpha^2 + (1-\alpha)^2} > \alpha$ for any $\alpha \in (0, 1)$:
$$\alpha^2 + (1-\alpha)^2 = \alpha^2 + (1 - 2\alpha + \alpha^2) = 2\alpha^2 - 2\alpha + 1$$
We test the inequality $2\alpha^2 - 2\alpha + 1 > \alpha^2$:
$$\alpha^2 - 2\alpha + 1 > 0 \implies (\alpha - 1)^2 > 0$$
Since $(\alpha - 1)^2$ is strictly positive for all $\alpha \in (0, 1)$, it follows that $Z_{\alpha} > Z_{\text{late}}$. Therefore, token-level aggregation strictly dominates independent per-key testing by preserving the complementary signal distributed across both keys ($k^{(1)}$ and $k^{(2)}$) at the token level, rather than systematically treating the minority key's tokens as noise during independent sequence-level evaluations.

\subsection{Stochastic Mixing}

Stochastic mixing introduces true randomness by mixing the deterministic PRF value $r_1$ with a Bernoulli coin. Given a parameter $a \in (0,1)$, the mixed value is $r = a \cdot r_1$ with probability $a$, or $r = a + (1-a) \cdot r_1$ with probability $1-a$. The mixed $r$ remains uniform (distortion-free), but detection uses only $r_1$.

\begin{proposition}[Bound on expected score under mixing]
\label{prop:mixing_score}
Under stochastic mixing with parameter $a \in (0, 1)$, detection is performed using $r_1$ (the deterministic PRF value). The expected score satisfies:
\begin{equation}
    \E[S_T] > T + \left( \frac{\pi^2}{6} -1 \right) H_T + \sum_{t=1}^T (1 - a^{1/p_t})\ln(1-a)
\end{equation}
\end{proposition}

\begin{proof}
Let $R \sim \text{Beta}(1/p, 1)$ be the random variable selected during sampling. The score for a single token is $s = -\ln(1-r_1)$, where $r_1$ is recovered from $R$ as: $r_1 = R/a$ if $R \in [0, a]$, and $r_1 = (R-a)/(1-a)$ if $R \in [a, 1]$.

We decompose $\E[s]$ by interval:
$$ \E[s] = \underbrace{\int_0^{a} -\ln(1 - r/a) f_R(r) dr}_{I_1} + \underbrace{\int_{a}^1 -\ln\!\left(1 - \frac{r-a}{1-a}\right) f_R(r) dr}_{I_2} $$

For $I_2$: using $1 - \frac{r-a}{1-a} = \frac{1-r}{1-a}$:
$$ I_2 = \int_{a}^1 [-\ln(1-r) + \ln(1-a)] f_R(r) dr = \int_{a}^1 -\ln(1-r) f_R(r) dr + (1 - a^{1/p})\ln(1-a) $$
since $\Prob(R > a) = 1 - a^{1/p}$.

For $I_1$: since $r/a \geq r$ for $r \in [0, a]$, we have $-\ln(1-r/a) \geq -\ln(1-r)$, giving $I_1 \geq \int_0^{a} -\ln(1-r) f_R(r) dr$.

Summing yields $\E[s] \geq \E[s_{\text{std}}] + (1 - a^{1/p})\ln(1-a)$ where $\E[s_{\text{std}}] = \mathcal{H}_{1/p}$. Applying the standard bound and summing over $T$ tokens gives the result.
\end{proof}

\begin{proposition}[Distortion-freeness of mixing]
\label{prop:mixing_distortion_free}
The mixed variable $r$ follows $\mathcal{U}[0, 1]$, so the sampled token follows the model distribution $\vec{p}$.
\end{proposition}

\begin{proof}
Let $F_R(x) = \Prob(r \leq x)$. For $x \leq a$: $r \leq x$ requires $r_0=0$, giving $\Prob(r \leq x) = a \cdot \Prob(r_1 \leq x/a) = a \cdot x/a = x$. For $x > a$: $\Prob(r \leq x) = a + (1-a) \cdot \frac{x-a}{1-a} = x$. Since $F_R(x) = x$, we have $r \sim \mathcal{U}[0,1]$.
\end{proof}

\paragraph{Behavior of the penalty.}
The penalty $(1 - a^{1/p})\ln(1-a)$ is always non-positive (since $\ln(1-a) < 0$) and vanishes at both extremes: as $a \to 0$, $\ln(1-a) \to 0$; as $a \to 1$, $(1-(1-\epsilon)^{1/p})\ln(\epsilon) \sim \frac{\epsilon}{p}\ln(\epsilon) \to 0$.
This is expected since in these extremes, all tokens take the same route, which makes it similar to vanilla Gumbel-max.

\subsection{Random Skip}

Random skip disables the watermark independently at each token with probability $\alpha$, reverting to standard sampling from $\vec{p}$. This blindly injects randomness to break deterministic loops, uniformly attenuating the detection signal.

\begin{proposition}[Bound on expected score under periodic skip]
\label{prop:periodic_score}
Under periodic skip with rate $\alpha \in [0, 1]$ (each token is independently skipped with probability $\alpha$), the expected score satisfies:
\begin{equation}
    \E[S_T] \geq T + (1-\alpha)\left( \frac{\pi^2}{6} -1 \right) H_T
\end{equation}
\end{proposition}

\begin{proof}
At each step $t$, with probability $1-\alpha$ the watermark is active and $\E[s_t] = \mathcal{H}_{1/p_t}$; with probability $\alpha$ the watermark is skipped and $\E[s_t] = 1$. Summing:
$$ \E[S_T] = (1-\alpha) \sum_{t=1}^T \mathcal{H}_{1/p_t} + \alpha T $$
Applying the standard bound (\autoref{prop:expected_score}) to $\sum_t \mathcal{H}_{1/p_t} \geq T + (\frac{\pi^2}{6}-1) H_T$:
$$ \E[S_T] \geq (1-\alpha)\left[T + \left(\frac{\pi^2}{6}-1\right) H_T\right] + \alpha T = T + (1-\alpha)\left(\frac{\pi^2}{6}-1\right) H_T $$
\end{proof}

The entropy-dependent signal is uniformly attenuated by a factor $(1-\alpha)$, regardless of the token entropy. This is wasteful compared to adaptive strategies that selectively skip only low-signal tokens.

\subsection{Adaptive Skip}

Adaptive skip disables the watermark selectively when the model is highly confident. At each step, the token is produced via Gumbel-Max, but if the winning PRF value $R_{V^\star}$ falls below a threshold $\tau$, the watermark is discarded and the token is resampled from $\vec{p}$. Low $R_{V^\star}$ indicates the token won due to high probability mass rather than a favorable PRF draw, so skipping it sacrifices little detection signal.

\begin{proposition}[Adaptive skip is not distortion-free]
\label{prop:adaptive_distortion}
Under adaptive skip with threshold $\tau \in (0,1)$, the output distribution is:
\begin{equation}
    \Prob(\text{output} = v) = p_v \left(1 - \tau^{1/p_v} + \sum_{w \in \V} p_w \, \tau^{1/p_w}\right)
\end{equation}
which differs from $p_v$ unless $\vec{p}$ is uniform.
\end{proposition}

\begin{proof}
Let $V^\star$ be the initial token selected by the Gumbel-max trick, where $\Prob(V^\star = v) = p_v$. 
By \autoref{cor:beta}, the conditional distribution of the pseudo-random value $R_v$ is $\text{Beta}(1/p_v, 1)$. 
The watermark is skipped if $R_{V^\star} < \tau$.

The marginal probability of outputting a specific token $v$ decomposes into two disjoint events: keeping the initially selected $v$, or skipping and resampling $X_t = v$ from the original distribution $\vec{p}$:
\begin{align*}
\Prob(\text{output} = v) &= \Prob(V^\star = v, R_v \geq \tau) + \Prob(\text{skip}) \Prob(X_t = v) \\
&= \Prob(V^\star = v) \Prob(R_v \geq \tau \mid V^\star = v) + \left( \sum_{w \in \V} \Prob(V^\star = w) \Prob(R_w < \tau \mid V^\star = w) \right) \Prob(X_t = v) \\
&= p_v \cdot \Prob(R_v \geq \tau \mid V^\star = v) + \left( \sum_{w \in \V} p_w \Prob(R_w < \tau \mid V^\star = w) \right) p_v \\
&= p_v \left( 1 - \tau^{1/p_v} \right) + p_v \sum_{w \in \V} p_w \tau^{1/p_w} \\
&= p_v \left( 1 - \tau^{1/p_v} + \sum_{w \in \V} p_w \tau^{1/p_w} \right)
\end{align*}
For the mechanism to be distortion-free, we require $\Prob(\text{output} = v) = p_v$ for all $v \in \V$. This implies:
$$ \tau^{1/p_v} = \sum_{w \in \V} p_w \tau^{1/p_w} $$
The right-hand side is a constant across all tokens, whereas the left-hand side strictly depends on $p_v$. This equality holds if and only if all tokens have the exact same probability $p_v = 1/|\V|$. 
\end{proof}

\begin{remark}
The distortion shifts mass from high-confidence tokens (large $p_v$, frequently skipped since $\tau^{1/p_v}$ is large) toward low-confidence tokens (small $p_v$, rarely skipped). 
For example, with $p_1 = 0.9$, $p_2 = 0.1$, and $\tau = 0.5$: the output probabilities become $(0.858, 0.142)$ instead of $(0.9, 0.1)$. 
In practice, $\tau$ is small (e.g., $\tau = 0.1$), so the distortion is mild.
\end{remark}

\begin{proposition}[Bound on expected score under adaptive skip]
\label{prop:adaptive_score}
Under adaptive skip with threshold $\tau \in [0, 1]$ (the watermark is disabled when $R_{V^\star}^{(t)} < \tau$), the expected score satisfies:
\begin{equation}
    \E[S_T] \geq T + \left( \frac{\pi^2}{6} -1 \right) H_T + \ln(1-\tau) \sum_{t=1}^T \tau^{1/p_t}
\end{equation}
The correction term is always non-positive, vanishing as $\tau \to 0$.
\end{proposition}

\begin{proof}
We condition on the identity of the selected token. By \autoref{prop:sampling}, $\Prob(V^\star = v) = p_v$. By \autoref{cor:beta}, \emph{conditioned on} $V^\star = v$, the PRF value $R_v \sim \text{Beta}(1/p_v, 1)$ with density $f(r) = \frac{1}{p_v} r^{1/p_v - 1}$ and CDF $F(r) = r^{1/p_v}$. The skip condition $R_v < \tau$ therefore has conditional probability $\Prob(R_v < \tau \mid V^\star = v) = \tau^{1/p_v}$. This decreases with entropy: for confident tokens ($p_v \to 1$), $\tau^{1/p_v} \to \tau$ (frequent skipping); for unlikely tokens ($p_v \to 0$), $\tau^{1/p_v} \to 0$ (rare skipping).

We now bound $\E[s_t \mid V^\star = v]$. Decomposing over the skip decision:
\begin{align*}
\E[s_t \mid V^\star = v] &= \underbrace{\E[-\ln(1-R_v) \cdot \mathbf{1}_{R_v \geq \tau} \mid V^\star = v]}_{\text{not skipped: use watermarked token}} + \underbrace{\E[-\ln(1-R_{X_t}) \cdot \mathbf{1}_{R_v < \tau} \mid V^\star = v]}_{\text{skipped: resample } X_t \sim \vec{p}}
\end{align*}
The first term integrates the score over the non-skip region using the conditional density of $R_v$:
$$ \E[-\ln(1-R_v) \cdot \mathbf{1}_{R_v \geq \tau} \mid V^\star = v] = \int_\tau^1 \frac{-\ln(1-r)}{p_v} r^{1/p_v - 1} dr $$
For the second term, the replacement token $X_t \sim \vec{p}$ is drawn with independent randomness, but its PRF value $R_{X_t}$ comes from the \emph{same} realization $\vec{R}$, so we cannot claim its expected score is $1$ (see \autoref{rem:skip_subtlety}). Since $-\ln(1-R_{X_t}) \geq 0$, the second term is non-negative, so:
$$ \E[s_t \mid V^\star = v] \geq \int_\tau^1 \frac{-\ln(1-r)}{p_v} r^{1/p_v - 1} dr = \mathcal{H}_{1/p_v} - \int_0^\tau \frac{-\ln(1-r)}{p_v} r^{1/p_v - 1} dr $$
where we used $\int_\tau^1 = \int_0^1 - \int_0^\tau$ and $\int_0^1 \frac{-\ln(1-r)}{p_v} r^{1/p_v - 1} dr = \mathcal{H}_{1/p_v}$. 
Since $-\ln(1-r) \leq -\ln(1-\tau)$ for $r \in [0, \tau]$:
\begin{align*}
\int_0^\tau \frac{-\ln(1-r)}{p_v} r^{1/p_v - 1} dr &\leq \int_0^\tau \frac{-\ln(1-\tau)}{p_v} r^{1/p_v - 1} dr \\
&= \frac{-\ln(1-\tau)}{p_v} \int_0^\tau r^{1/p_v - 1} dr \\
&= \frac{-\ln(1-\tau)}{p_v} \left[ p_v \cdot r^{1/p_v} \right]_0^\tau \\
&= \frac{-\ln(1-\tau)}{p_v} \left( p_v \cdot \tau^{1/p_v} - 0 \right) \\
&= -\ln(1-\tau) \cdot \tau^{1/p_v}
\end{align*}
and therefore $\E[s_t \mid V^\star = v] \geq \mathcal{H}_{1/p_v} + \tau^{1/p_v} \ln(1-\tau)$. 
Since this holds for every $v$, it holds for the realized token probability $p_t = p_{V^\star}$. Summing over $T$ steps and applying the standard bound (\autoref{prop:expected_score}) to $\sum_t \mathcal{H}_{1/p_t} \geq T + (\frac{\pi^2}{6}-1) H_T$ gives the result.
\end{proof}

\begin{remark}
The penalty $\ln(1-\tau) \sum_t \tau^{1/p_t}$ is always non-positive (since $\ln(1-\tau) < 0$), confirming that skipping can only reduce the detection signal. For small $\tau$ (e.g., $\tau = 0.1$), the penalty is negligible: $\tau^{1/p_t}$ is small for all but deterministic tokens ($p_t \approx 1$), and those tokens carry no watermark signal anyway ($\mathcal{H}_1 = 1$, equal to the null baseline). The bound is conservative because we dropped the skip contribution entirely; in practice, skipped tokens still contribute positively to the score.
\end{remark}

\begin{remark}[Skip contribution]
\label{rem:skip_subtlety}
A tempting (but incorrect) approach is to claim that skipped tokens contribute expected score $1$, arguing that the replacement token $X_t \sim \vec{p}$ is drawn independently and therefore its PRF value $R_{X_t}$ is uniform. This would yield the decomposition:
$$ \E[s_t] = \int_\tau^1 \frac{-\ln(1-r)}{p_t} r^{1/p_t-1} dr + \tau^{1/p_t} \cdot 1 $$
leading to a correction $\tau^{1/p_t}(1 + \ln(1-\tau))$ that is \emph{positive} for $\tau < 1 - 1/e$---implying that skipping \emph{improves} detection, which is impossible.

The error is that while $X_t$ is drawn independently of $\vec{R}$, the PRF value $R_{X_t} = \vec{R}[X_t]$ shares the \emph{same} realization $\vec{R}$. Since the skip event $\{R_{V^\star} < \tau\}$ constrains $\vec{R}$ (the winning PRF value is low), the conditional expectation $\E[-\ln(1-R_{X_t}) \mid R_{V^\star} < \tau] \neq 1$.
A simple counterexample: for a deterministic token ($p_t = 1$), there is only one possible token, so skipping changes nothing and $\E[s_t] = \mathcal{H}_1 = 1$. Yet the incorrect formula gives $1 + \tau(1+\ln(1-\tau)) > 1$.
\end{remark}

\subsection{Entropy-Normalized Adaptive Skip}
\label{subsec:entropy_normalized_skip}

This variant of adaptive skip replaces the fixed threshold $\tau$ with an entropy-dependent threshold $\tau^{p_{V^\star}}$, which ensures every token is skipped with exactly the same probability $\tau$ regardless of its confidence level. This restores the distortion-free property lost by standard adaptive skip.

For a target skip rate $\tau \in (0, 1)$, the watermark is now disabled and the token is resampled from $\vec{p}$ if:
$$ R_{V^\star} < \tau^{p_{V^\star}} $$

\begin{proposition}[Distortion-freeness of entropy-normalized skip]
\label{prop:normalized_distortion_free}
The entropy-normalized adaptive skip mechanism is distortion-free, \ie: $\Prob(\text{output} = v) = p_v$ for all $v \in \V$.
\end{proposition}

\begin{proof}
We first evaluate the conditional probability of a skip occurring given that token $v$ was initially selected. 
By \autoref{cor:beta}, $R_v \mid V^\star = v \sim \text{Beta}(1/p_v, 1)$, which has the cumulative distribution function $F(r) = r^{1/p_v}$. Therefore, the conditional skip probability is:
$$ \Prob(\text{skip} \mid V^\star = v) = \Prob(R_v < \tau^{p_v} \mid V^\star = v) = \left( \tau^{p_v} \right)^{1/p_v} = \tau $$

Because this conditional probability is exactly $\tau$ for \emph{every} token in the vocabulary, the unconditional probability of a skip is also exactly $\tau$. 
Indeed, by the law of total probability:
$$ \Prob(\text{skip}) = \sum_{w \in \V} \Prob(\text{skip} \mid V^\star = w) \Prob(V^\star = w) = \sum_{w \in \V} \tau \cdot p_w = \tau \sum_{w \in \V} p_w = \tau $$

The total marginal probability of outputting token $v$ can then be found by partitioning over the two mutually exclusive generation paths (whether a skip occurs or not):
\begin{align*}
\Prob(\text{output} = v) &= \Prob(\text{output} = v \cap \text{not skipped}) + \Prob(\text{output} = v \cap \text{skip}) \\
&= \Prob(V^\star = v \cap R_{V^\star} \geq \tau^{p_{V^\star}}) + \Prob(\text{skip}) \cdot \Prob(X_t = v) \\
&= \Prob(V^\star = v) \cdot \Prob(R_v \geq \tau^{p_v} \mid V^\star = v) + \Prob(\text{skip}) \cdot \Prob(X_t = v) 
\end{align*}

We know the unconditional probability of a skip is $\tau$, so the conditional probability of not skipping is $1 - \tau$. 
Furthermore, the replacement token $X_t$ is sampled from the original distribution independently of the skip event, so $\Prob(X_t = v) = p_v$. 
Substituting these values yields:
\begin{align*}
\Prob(\text{output} = v) &= p_v \cdot (1 - \tau) + \tau \cdot p_v \\
&= p_v - p_v \tau + p_v \tau \\
&= p_v
\end{align*}
Thus, the marginal distribution is perfectly preserved, making the entropy-normalized mechanism distortion-free.
\end{proof}

\begin{proposition}[Bound on expected score under entropy-normalized skip]
\label{prop:normalized_score}
Under the entropy-normalized adaptive skip with target skip rate $\tau \in (0, 1)$, the expected score satisfies:
\begin{equation}
    \E[S_T] \geq T + \left( \frac{\pi^2}{6} -1 \right) H_T + \tau \sum_{t=1}^T \ln(1-\tau^{p_t})
\end{equation}
\end{proposition}

\begin{proof}
We condition on the identity of the selected token $V^\star = v$. We decompose the expected score into the non-skipped and skipped cases:
\begin{align*}
\E[s_t \mid V^\star = v] &= \E[-\ln(1-R_v) \cdot \mathbf{1}_{R_v \geq \tau^{p_v}} \mid V^\star = v] + \E[-\ln(1-R_{X_t}) \cdot \mathbf{1}_{R_v < \tau^{p_v}} \mid V^\star = v]
\end{align*}
As discussed in \autoref{rem:skip_subtlety}, the replacement token $X_t$ relies on the same PRF realization $\vec{R}$, so its contribution is difficult to isolate but strictly non-negative. Dropping the second term provides a conservative lower bound:
$$ \E[s_t \mid V^\star = v] \geq \int_{\tau^{p_v}}^1 \frac{-\ln(1-r)}{p_v} r^{1/p_v - 1} dr = \mathcal{H}_{1/p_v} - \int_0^{\tau^{p_v}} \frac{-\ln(1-r)}{p_v} r^{1/p_v - 1} dr $$
Since the function $-\ln(1-r)$ is monotonically increasing, for $r \in [0, \tau^{p_v}]$, we have $-\ln(1-r) \leq -\ln(1-\tau^{p_v})$. We can bound the subtracted integral:
\begin{align*}
\int_0^{\tau^{p_v}} \frac{-\ln(1-r)}{p_v} r^{1/p_v - 1} dr &\leq -\ln(1-\tau^{p_v}) \int_0^{\tau^{p_v}} \frac{1}{p_v} r^{1/p_v - 1} dr \\
&= -\ln(1-\tau^{p_v}) \left[ r^{1/p_v} \right]_0^{\tau^{p_v}} \\
&= -\ln(1-\tau^{p_v}) \left( \tau^{p_v} \right)^{1/p_v} \\
&= -\tau \ln(1-\tau^{p_v})
\end{align*}
Substituting this back yields:
$$ \E[s_t \mid V^\star = v] \geq \mathcal{H}_{1/p_v} + \tau \ln(1-\tau^{p_v}) $$
Since this inequality holds for any chosen $v$, it holds for the realized token probability $p_t$. Summing over the sequence of $T$ tokens and applying the standard Gumbel bound (\autoref{prop:expected_score}) gives:
$$ \E[S_T] \geq \sum_{t=1}^T \mathcal{H}_{1/p_t} + \tau \sum_{t=1}^T \ln(1-\tau^{p_t}) \geq T + \left( \frac{\pi^2}{6} -1 \right) H_T + \tau \sum_{t=1}^T \ln(1-\tau^{p_t}) $$
The correction term is strictly non-positive because $\tau^{p_t} \in (0, 1)$, meaning $\ln(1-\tau^{p_t}) < 0$. This accurately reflects the expected loss in signal when skipping exactly $\tau$ fraction of the tokens.
\end{proof}

\begin{remark}[Skip behavior and score penalty]
Unlike standard adaptive skip where the skip probability $\tau^{1/p_v}$ depends on token confidence (skipping high-confidence tokens more often), the entropy-normalized threshold $\tau^{p_v}$ ensures a \emph{uniform} skip rate of exactly $\tau$ for all tokens regardless of their probability. 
However, the per-token score penalty $\tau \ln(1-\tau^{p_t})$ still varies with entropy:
\begin{itemize}
    \item For high-confidence tokens ($p_t \to 1$): the penalty approaches $\tau \ln(1-\tau)$, which is mild. These tokens contribute little watermark signal anyway ($\mathcal{H}_{1} = 1$, equal to the null baseline), so skipping them has minimal impact.
    \item For low-confidence tokens ($p_t \to 0$): the threshold $\tau^{p_t} \to 1$, making the penalty bound $\tau \ln(1-\tau^{p_t}) \to -\infty$, making the bound effectively useless.
    Such tokens occur rarely ($\Prob(V^\star = v) = p_v$), so their contribution to the total penalty is attenuated by their low occurrence frequency.
\end{itemize}
The mechanism thus achieves distortion-freeness while concentrating the detection penalty on tokens that either carry little signal (high-confidence) or appear rarely (low-confidence), preserving most of the watermark power from medium-entropy tokens.
\end{remark}

\section{Fast Localization and Statistical Penalties}
\label{app:localization_math}

In practical settings, the exact start and end indices of a watermarked insertion are unknown. Our objective is to determine a set of disjoint watermarked intervals $\{[a_1, b_1], \dots, [a_y, b_y]\}$. A naive exhaustive search over all possible intervals in a sequence of length $n$ requires evaluating $\binom{n}{2} \approx n^2/2$ windows. This $\mathcal{O}(n^2)$ search space not only introduces severe computational bottlenecks but also imposes an insurmountable statistical penalty via multiple-testing correction, as false positives become significantly more likely as the number of tested hypotheses grows. To optimize both computational and statistical efficiency, we utilize a geometric cover search space \citep{kirchenbauer2023reliability} combined with a fast two-stage extraction pipeline and rigorous Bonferroni correction.

\subsection{The Geometric Cover Search}
To avoid testing $\mathcal{O}(n^2)$ intervals, we constrain our search to a dyadic grid of windows. We define a set of window lengths $L \in \{L_0, 2L_0, 4L_0, \dots, 2^{\lfloor \log_2 n \rfloor}\}$, where $L_0 = 2^{\lceil \log_2 L_{\min} \rceil}$ is the smallest power of two at least as large as the minimum zone length $L_{\min}$. For each length $L$, we slide the window across the text with a stride of $L/2$. 

This geometric grid guarantees that any arbitrary watermarked region of length $L^* \ge L_{\min}$ will be at least $50\%$ covered by at least one grid window. The total number of candidate windows $M$ in this grid is strictly bounded:
\begin{equation}
M = \sum_{k = \lceil \log_2 L_{\min} \rceil}^{\lfloor \log_2 n \rfloor} \left\lfloor \frac{n - 2^k}{2^{k-1}} \right\rfloor + 1 \approx \frac{4n}{L_0}
\end{equation}
By restricting the search to $M \approx \mathcal{O}(n/L_{\min})$ windows, we reduce the hypothesis space by orders of magnitude, dramatically lowering the statistical tax required to claim significance.

\subsection{Fast Two-Stage Pipeline and Greedy Extraction}
Evaluating the rigorous, entropy-weighted Gamma distribution for all $M$ windows can be computationally heavy. To process arbitrarily large documents efficiently, we utilize a two-stage pipeline:
\begin{enumerate}
    \item \textbf{Fast Filtering:} We pre-calculate prefix sums of the unweighted raw scores $s_i$. The sum for any candidate interval in the grid can then be computed in $\mathcal{O}(1)$ time. We select the top candidates based on these raw sums.
    \item \textbf{Rigorous Scoring:} For the most promising candidates, we compute the exact entropy-weighted moment-matched Gamma $p$-value $p_{\text{raw}}$ (as defined in Section~\ref{subsec:entropy_detection}).
\end{enumerate}
To extract multiple zones, we proceed greedily. We find the window $I^*$ with the most significant $p_{\text{raw}}$. If its penalized significance (accounting for the search tax) is high enough, we flag it as watermarked, mask its tokens (setting their scores to zero), and repeat the search on the residual text. We aggregate disjoint intervals until the combined $p$-value fails to overcome the multiple-testing threshold, up to a maximum of $Y_{\max}$ zones.

\subsection{The Bonferroni Tax and False Positive Guarantee}
Evaluating $M$ intervals introduces the multiple comparisons problem. To maintain a strict family-wise error rate (FWER) $\epsilon$ under the null hypothesis $\mathcal{H}_0$, we apply a union bound. 

For a single-zone search, the Bonferroni correction factor is simply $M$. For a multi-zone search identifying $y$ disjoint regions, we must account for the number of ways to choose $y$ windows from the grid, $\binom{M}{y}$, as well as the optimization over $y \in [1, Y_{\max}]$. The corrected $p$-value in log-space is:
\begin{equation}
\ln p_{\text{corrected}} = \ln p_{\text{raw}} + \ln \binom{M}{y} + \ln Y_{\max}
\end{equation}
Under $\mathcal{H}_0$, the probability that the most significant window combination exceeds our threshold is strictly bounded:
\begin{equation}
\mathbb{P}\left( \bigcup_{i=1}^K \left\{ p_i \le \frac{\epsilon}{K} \right\} \;\middle|\; \mathcal{H}_0 \right) \le \sum_{i=1}^K \mathbb{P}\left( p_i \le \frac{\epsilon}{K} \;\middle|\; \mathcal{H}_0 \right) = K \cdot \frac{\epsilon}{K} = \epsilon
\end{equation}
where $K$ represents the total number of tested hypotheses in the search space. This guarantees that the probability of falsely accusing an entirely human-written text remains $\le \epsilon$, regardless of document length.

\subsection{Asymptotic Power Comparison: Global vs. Localized Detection}
\label{app:power_proof}

We define the crossover point where the localized multi-zone test yields a stronger rejection of $\mathcal{H}_0$ than the global test. Let $n$ be the document length, and $\rho \in (0, 1]$ be the fraction of tokens that are watermarked ($w = \rho n$). 

\paragraph{Setup and Approximation.}
Let the weighted token score $\tilde{s}_i$ have mean $\mu_0$ and variance $\sigma^2$ under $\mathcal{H}_0$. Under $\mathcal{H}_1$, the mean shifts to $\mu_w > \mu_0$. Let $\delta = (\mu_w - \mu_0)/\sigma$ be the per-token signal-to-noise ratio. Using a Gaussian tail approximation, the log $p$-value of a Z-score is $\ln p \approx -\frac{1}{2} Z^2$. We define $\Delta^2 = \delta^2/2$ as the expected log $p$-value accumulation rate per watermarked token.

\paragraph{Power of the Global Test.}
The global test evaluates all $n$ tokens. The expected Z-score is:
\begin{equation}
Z_{\text{global}} = \frac{\rho n \sigma \delta}{\sigma \sqrt{n}} = \rho \delta \sqrt{n} \implies \mathbb{E}[\ln p_{\text{global}}] \approx -\rho^2 n \Delta^2
\end{equation}
The signal strength scales quadratically with $\rho$; the $(1-\rho)n$ human tokens contribute no signal but inflate the variance, diluting the test.

\paragraph{Power of the Localized Test.}
Assuming a localized test correctly isolates the $\rho n$ watermarked tokens into $y$ zones, the variance is reduced to the watermarked subset $\rho n \sigma^2$. The expected raw log $p$-value is:
\begin{equation}
Z_{\text{local}} = \frac{\rho n \sigma \delta}{\sigma \sqrt{\rho n}} = \delta \sqrt{\rho n} \implies \mathbb{E}[\ln p_{\text{local, raw}}] \approx -\rho n \Delta^2
\end{equation}
Accounting for the combinatorial tax $\ln \binom{n}{2y} \approx 2y \ln n$, the penalized localized score is:
\begin{equation}
\mathbb{E}[\ln p_{\text{local, final}}] \approx -\rho n \Delta^2 + 2y \ln n
\end{equation}

\paragraph{The Crossover Point.}
The localized test dominates when $\mathbb{E}[\ln p_{\text{local, final}}] < \mathbb{E}[\ln p_{\text{global}}]$:
\begin{align}
-\rho n \Delta^2 + 2y \ln n &< -\rho^2 n \Delta^2 \\
n \Delta^2 (\rho - \rho^2) &> 2y \ln n \\
\rho(1-\rho) &> \frac{2y \ln n}{n \Delta^2}
\end{align}
This inequality demonstrates that localized detection is optimal when the signal is sufficiently concentrated (low $\rho$) such that the variance reduction from excluding human tokens outweighs the logarithmic search tax.

\section{Additional Experiments and Details}

\subsection{Benchmark Variance Analysis}
\label{app:variance_analysis}

Evaluating LLM performance on benchmarks with chain-of-thoughts involves variance due to the stochastic generation, and the final answer extraction heuristics (e.g., regex-based parsing for numbers or letters, or code execution for programming tasks).
To quantify this variance and assess whether watermarking introduces systematic degradation or improvement, we re-ran a subset of benchmarks with multiple random seeds for non-watermarked generation and multiple secret keys for watermarked generation.

\paragraph{Experimental Setup.}
We use Qwen~3.5-27B with reasoning enabled (reasoning temperature 0.6, top-$p = 0.95$, max 3,000 reasoning tokens).
We evaluate on five benchmarks: AIME (math), GSM8K (math), HumanEval (code), MBPP (code), and MMLU (multiple choice).
For non-watermarked generation, we use 5 different random seeds.
For watermarked generation, we use the same values as secret keys, with both $n$-gram deduplication enabled and disabled (see \autoref{rem:deduplication}): when enabled, watermark contexts that have already appeared in the generation fall back to vanilla sampling instead of watermarked sampling.

\paragraph{Results.}

\begin{table}[t!]
\centering
\small
\caption{Benchmark accuracy (\%) across 5 random seeds (non-watermarked) or 5 secret keys (watermarked). We report Mean $\pm$ Std to quantify generation variance. ``Dedup'' refers to $n$-gram deduplication at generation-time (see \autoref{rem:deduplication}). Differences between conditions fall within one standard deviation, indicating no systematic degradation from watermarking.}
\label{tab:variance_analysis}
\begin{tabular}{lccc}
\toprule
Benchmark & No Watermark & WM (no dedup) & WM (dedup) \\
\midrule
AIME & $41.0 \pm 1.2$ & $40.6 \pm 1.2$ & $40.8 \pm 0.9$ \\
GSM8K & $95.9 \pm 0.3$ & $95.6 \pm 0.3$ & $95.9 \pm 0.3$ \\
HumanEval & $97.1 \pm 0.9$ & $97.6 \pm 2.4$ & $97.8 \pm 0.7$ \\
MBPP & $50.2 \pm 0.6$ & $49.8 \pm 0.4$ & $49.7 \pm 0.6$ \\
MMLU & $87.6 \pm 0.5$ & $87.1 \pm 1.6$ & $87.8 \pm 0.6$ \\
\bottomrule
\end{tabular}
\end{table}

The standard deviation across seeds/keys ranges from 0.3\% to 2.4\%, depending on the benchmark and condition.
Code benchmarks exhibit variance due to the binary nature of test execution and sensitivity to minor formatting differences.
Crucially, the differences between watermarked and non-watermarked conditions fall within approximately one standard deviation, indicating no systematic performance degradation from watermarking.

\paragraph{Effect of $n$-Gram Deduplication.}
Enabling $n$-gram deduplication (falling back to vanilla sampling for repeated context windows) tends to produce lower variance, particularly visible on MMLU (0.6 vs 1.6 std) and HumanEval (0.7 vs 2.4 std).
This is consistent with the observation that repeated $n$-gram contexts in reasoning chains can lead to more deterministic (and potentially repetitive) generation patterns when not deduplicated.

\subsection{Multilingual QA}
\label{app:multilingual_setup}

This section provides the full experimental setup for the multilingual question-answering evaluation.
The same dataset and generation pipeline are used for both the watermark detection analysis below and the human preference evaluation in \autoref{app:human_eval_details}.

\paragraph{Experimental Configuration.}
We use GPT-OSS-20B with reasoning enabled (max 2,000 reasoning tokens) and watermarking applied to the reasoning trace.
Generation uses temperature $0.7$, top-$p = 0.95$, and a maximum length of 4,096 tokens.
The watermark employs Gumbel-Max with 3-gram context, dual-key early fusion ($\alpha = 0.1$), and a fixed secret key.

\paragraph{Datasets.}
We evaluate on 6,000 question-answer pairs across five languages:
English (2,000 samples from ELI5),
and Arabic, Chinese, Hindi, and Japanese (1,000 samples each from CaLMQA~\citep{arora2025calmqa}).

\paragraph{System Prompt.}
The following system prompt was used for all languages:
\begin{quote}
\small
\textit{``You are answering questions. Give a clear, concise explanation in plain language. Answer in the same language as the question. Keep your answer to 50--150 words. No bullet points, headers, or markdown formatting---just natural prose.''}
\end{quote}

\paragraph{Watermark Detection Results.}

\begin{table}[t!]
\centering
\small
\caption{Watermark detection performance on multilingual QA.}
\label{tab:multilingual_detection_pivot}
\begin{tabular}{l c c c c c c}
\toprule
Metric & English & Arabic & Chinese & Hindi & Japanese & Overall \\
\midrule
TPR@0.1\% & 53.6\% & 83.3\% & 79.5\% & 59.0\% & 51.0\% & 63.3\% \\
Median $\log_{10} p$ & $-3.15$ & $-5.23$ & $-4.62$ & $-3.51$ & $-3.04$ & $-3.72$ \\
\bottomrule
\end{tabular}
\end{table}

Arabic and Chinese show strongest detection, which is likely due to higher per-token entropy. 
Japanese shows lowest detection (51\%) due to the more constrained vocabulary and lower entropy in CJK scripts.

\paragraph{Statistical Tests for Differences.}
We apply McNemar's test~\citep{mcnemar1947note} to assess whether watermarking systematically affects script consistency or refusal rates.
For script consistency, we observe 52 discordant pairs where WM was wrong but Non-WM was correct, versus 39 where Non-WM was wrong but WM was correct; with continuity correction, this yields $\chi^2 = 1.58$ and $p = 0.21$.
For refusal rates, we find 21 pairs where WM refused but Non-WM answered, versus 15 where Non-WM refused but WM answered, giving $\chi^2 = 0.69$ and $p = 0.41$.
Both $p$-values are well above the significance threshold ($\alpha = 0.05$), indicating that watermarking does not seem to systematically increase script errors or refusals.

\subsection{Human Evaluation Details}\label{app:human_eval_details}

This section provides methodology and detailed results for the human evaluation study summarized in \autoref{subsec:human_eval}.
The experimental setup (model, datasets, generation parameters) is shared with the multilingual QA experiment described in \autoref{app:multilingual_setup}.

\paragraph{Preference Distribution.}
\autoref{tab:human_eval_full} shows the complete four-class preference breakdown before merging tie categories.
Annotators chose among: \emph{A is preferred}, \emph{B is preferred}, \emph{Both equally good}, and \emph{Both equally bad}.
We aggregate via majority vote (at least 2/3 annotators agree); samples with a three-way split (one vote per distinct category) are assigned to ``Tie.''
For the final analysis, ``Both Good,'' ``Both Bad,'' and splits are merged into a single Tie category.

\begin{table}[t!]
\centering
\small
\caption{
    \textbf{Full four-class preference breakdown} (majority vote, 3 annotators per sample).
    Split: items where no majority exists (three-way tie), counted as Tie in the final analysis.
}
\label{tab:human_eval_full}
\begin{tabular}{lrrrrrr}
\toprule
\textbf{Language} & \textbf{N} & \textbf{Prefer WM} & \textbf{Prefer Base} & \textbf{Both Good} & \textbf{Both Bad} & \textbf{Split} \\
\midrule
English  & 2{,}000 & 124 & 146 & 1{,}482 & 92 & 156 \\
Arabic   & 1{,}000 & 201 & 181 & 168 & 287 & 163 \\
Chinese  & 1{,}000 & 84  & 74  & 514 & 272 & 56 \\
Hindi    & 1{,}000 & 91  & 89  & 435 & 278 & 107 \\
Japanese & 1{,}000 & 143 & 130 & 275 & 228 & 224 \\
\midrule
\textbf{Overall} & 6{,}000 & 643 & 620 & 2{,}874 & 1{,}157 & 706 \\
\bottomrule
\end{tabular}
\end{table}

\paragraph{Net Win Rate.}
We define the \emph{net win rate} as
\begin{equation}
    \text{Net Win Rate} = \frac{n_{\text{WM}} - n_{\text{Base}}}{N},
\end{equation}
where $n_{\text{WM}}$ and $n_{\text{Base}}$ are the number of samples where the watermarked or baseline response was preferred (by majority vote), and $N$ is the total number of samples including ties.
The overall net win rate is $+0.38\%$, indicating a negligible difference.

\paragraph{Binomial Test.}
Among decisive (non-tie) samples, we test the null hypothesis $H_0: P(\text{WM preferred}) = 0.5$ using a two-sided exact binomial test.
No individual language reaches significance at $\alpha = 0.05$ (English: $p = 0.20$; Arabic: $p = 0.33$; Chinese: $p = 0.47$; Hindi: $p = 0.94$; Japanese: $p = 0.47$).
The overall pooled test yields $p = 0.54$, also non-significant.
The overall preference is nearly evenly split (50.9\% WM among decisive samples), indicating no quality degradation.

\paragraph{Equivalence Testing (TOST with Ties).}
To establish imperceptibility, rather than failing to detect a difference, we apply the Two One-Sided Tests (TOST) procedure~\citep{schuirmann1987comparison}.
We test:
\begin{equation}
    H_0{:}~|P(\text{WM preferred}) - P(\text{Base preferred})| \geq \Delta \quad \text{vs.} \quad H_1{:}~|P(\text{WM preferred}) - P(\text{Base preferred})| < \Delta
\end{equation}
where proportions are computed over \emph{all} $N$ samples (including ties in the denominator).
This formulation is more powerful than restricting to decisive samples, because ties represent direct evidence of imperceptibility (the annotator could not distinguish between outputs) and contribute to the sample size.

Let $\hat{d} = \hat{p}_{\text{WM}} - \hat{p}_{\text{Base}}$ with standard error $\text{SE} = \sqrt{(\hat{p}_{\text{WM}} + \hat{p}_{\text{Base}} - \hat{d}^2)/N}$.
The TOST procedure computes two one-sided $z$-tests:
$z_1 = (\hat{d} - \Delta)/\text{SE}$ and $z_2 = (\hat{d} + \Delta)/\text{SE}$,
and rejects $H_0$ when $\max(\Phi(z_1),\, 1 - \Phi(z_2)) < \alpha$.

\autoref{tab:tost_results} reports results for $\Delta = 5\%$.
Equivalence is established for four of five languages and overall ($p < 0.05$); Arabic marginally fails ($p = 0.062$) as its upper confidence bound ($5.2\%$) slightly exceeds the $\pm 5$ percentage-point margin.

\begin{table}[t!]
\centering
\small
\caption{
    \textbf{TOST equivalence test results} ($\Delta = 5\%$, $\alpha = 0.05$).
    Proportions computed over all $N$ samples. 90\% CI: Wald interval for the difference $P(\text{WM}) - P(\text{Base})$.
}
\label{tab:tost_results}
\begin{tabular}{lrcccl}
\toprule
\textbf{Language} & $N$ & $\hat{d}$ & 90\% CI & $p_{\text{TOST}}$ & Result \\
\midrule
English  & 2{,}000 & $-1.10\%$ & $[-2.5\%, +0.3\%]$ & $< 0.001$ & Equivalent \\
Arabic   & 1{,}000 & $+2.00\%$ & $[-1.2\%, +5.2\%]$ & $0.062$ & Marginal \\
Chinese  & 1{,}000 & $+1.00\%$ & $[-1.1\%, +3.1\%]$ & $< 0.001$ & Equivalent \\
Hindi    & 1{,}000 & $+0.20\%$ & $[-2.0\%, +2.4\%]$ & $< 0.001$ & Equivalent \\
Japanese & 1{,}000 & $+1.30\%$ & $[-1.4\%, +4.0\%]$ & $0.013$ & Equivalent \\
\midrule
\textbf{Overall} & 6{,}000 & $+0.38\%$ & $[-0.6\%, +1.4\%]$ & $< 0.001$ & Equivalent \\
\bottomrule
\end{tabular}
\end{table}

\paragraph{On Trinomial Tests.}
An alternative approach is the trinomial test for paired data with ties~\citep{bian2011trinomial, dathathri2024scalable}, which models the three-category distribution (WM, Base, Tie) directly.
We experimented with this approach but found that the chi-square statistic converges rapidly with the number of ties: once more than a handful of ties are present, the $p$-value stabilizes to the second decimal place and equals the standard binomial test on decisive samples.
Since 79\% of our samples are ties, the trinomial test provides no additional discriminative power, which motivates our use of the TOST procedure that explicitly leverages ties as evidence of imperceptibility.

\paragraph{Inter-Annotator Agreement.}
We measure agreement using two metrics: (i)~unanimous agreement rate (fraction of samples where all 3 annotators selected the same four-class option), and (ii)~mean pairwise agreement (average fraction of annotator pairs that agree on the four-class label).
\autoref{tab:human_eval_agreement} shows that agreement varies by language, with English and Chinese exhibiting the highest consistency.
The lower agreement rates for Arabic and Japanese may reflect the inherent subjectivity of quality judgments and cultural differences in evaluation norms.

\begin{table}[t!]
\centering
\small
\caption{
    \textbf{Inter-annotator agreement statistics by language} (four-class scale).
}
\label{tab:human_eval_agreement}
\begin{tabular}{lccc}
\toprule
\textbf{Language} & \textbf{Unanimous Rate} & \textbf{Majority ($\geq$2/3)} & \textbf{Pairwise Agreement} \\
\midrule
English  & 54.0\% & 92.2\% & 0.667 \\
Arabic   & 23.9\% & 83.7\% & 0.438 \\
Chinese  & 48.5\% & 94.4\% & 0.638 \\
Hindi    & 37.0\% & 89.3\% & 0.544 \\
Japanese & 17.0\% & 77.6\% & 0.372 \\
\midrule
\textbf{Overall} & 39.1\% & 88.2\% & 0.555 \\
\bottomrule
\end{tabular}
\end{table}

\subsection{Learnability Experimental Details}
\label{app:learnability_details}

\paragraph{Models \& Dataset.}
The teacher is DeepSeek-R1-Distill-Qwen-14B~\citep{guo2025deepseek}, an R1-style reasoning model that produces long chain-of-thought traces enclosed in \texttt{<think>} tags, and the student is Qwen2.5-3B~\citep{qwen25}.
We train on a subset of 5{,}000 problems drawn from OpenR1-Math-220k~\citep{openr1}, curated via a three-stage pipeline:
(i)~malformed or incomplete problems are removed;
(ii)~only problems that the student model fails to solve are retained, ensuring the training data teaches new capabilities;
(iii)~diversity sampling across 14 math categories with a 15\% cap per category prevents overrepresentation of any single topic.

\paragraph{Watermarked Trace Generation.}
We compare the three sampling-based methods introduced in \autoref{subsec:exp_setup}: \emph{Gumbel-Max}~\citep{aaronson2023watermarking}, \emph{TextSeal} (dual-key routing probability $\alpha = 0.1$), and \emph{SynthID}~\citep{dathathri2024scalable} (plus an unwatermarked control), all with watermark context window $k = 3$.
Secret keys are calibrated per method via a Kolmogorov--Smirnov test to ensure uniform PRF hashes on unwatermarked text as done in~\citet{fernandez2025good}.
The teacher generates 5{,}000 solutions using vLLM~\citep{kwon2023efficient} with flash-attention-2 on $4{\times}$H200 GPUs (tensor parallel), with $T = 1.0$, top-$p = 0.95$, and max 8{,}192 generated tokens.

\paragraph{Quality Filtering.}
Each teacher trace passes through four sequential filters (the first failure rejects the trace):
(i)~\emph{think closure}---the trace must contain a closing \texttt{</think>} tag;
(ii)~\emph{boxed presence}---the trace must include a \verb|\boxed{...}| final-answer pattern (skipped for multiple-choice datasets);
(iii)~\emph{repetition detection}---a sampled sliding-window check (window size 100 characters, ${\sim}200$ evenly spaced samples) rejects any trace in which a substring occurs ${\geq}3$ times (responses ${\leq}200$ characters auto-pass);
(iv)~\emph{answer verification}---the extracted answer is compared to the gold answer using the \texttt{math\_verify} library in a fail-open mode: if either side fails to parse, the trace is kept rather than rejected.

\paragraph{Student Fine-Tuning.}
The student is fine-tuned on the filtered traces using LoRA~\citep{hu2022lora} (rank 128, scaling factor 128, dropout 0.05) with learning rate $2 \times 10^{-5}$ and 3 epochs.
The loss is computed over the full teacher response (both the reasoning trace and the final answer) while the prompt tokens are masked out.

\paragraph{Watermark Detection.}
We evaluate watermark transfer using the \emph{open-model} radioactivity test of \citet{sander2024watermarking,sander2025detecting}.
The test operates in a \emph{teacher-forcing} setup: each training trace is fed into the student model, and the student's top-1 prediction $\hat{x}^{(t)} = \arg\max_{v \in \V} P_\theta(v \mid x_{<t})$ is recorded at every response position $t$.
Crucially, we score the student's \emph{predictions} rather than newly generated text: this isolates the watermark signal from confounding factors such as sampling noise and generation quality, while requiring only a single forward pass over the existing traces rather than expensive autoregressive generation.
If the student has internalized the watermark's token preferences during fine-tuning, its top-1 predictions will be systematically biased toward high-PRF tokens---even without access to the secret key.

We score each prediction using the watermark's PRF: $R_t = \text{PRF}(\hat{x}^{(t)},\, \mathbf{w}_t,\, \sk)$, where $\mathbf{w}_t = (x^{(t-k)}, \ldots, x^{(t-1)})$ is the trigram context window of \emph{teacher} tokens preceding position $t$ (as defined in \autoref{sec:background}).
Within each trace, each context window is scored only once; across traces, all (context, predicted token) pairs are pooled and deduplicated so that repeated tuples are counted only once, satisfying the independence assumption required by the statistical test~\citep{fernandez2023three}.
This yields ${\sim}1.4$--$2.2$M unique scored tokens per method.

For Gumbel-Max, a single pooled Gamma test produces the $p$-value: we compute $s_t = -\ln(1 - R_t)$ for each unique pair and sum over all $n$ unique scored tokens to obtain $S_n = \sum_{t=1}^{n} s_t$, which under $\Hnull$ follows $\Gamma(n, 1)$ (\autoref{prop:pvalue}).
For TextSeal, we use the entropy-weighted early-fusion score with $w_i^{\text{ent}} = \sqrt{\hat{H}_i}$ weighting (\autoref{subsec:entropy_detection}), where entropy is estimated from the student model's forward pass, and compute the $p$-value via the moment-matched Gamma approximation of \autoref{eq:p_value_gamma_combined}; the choice of weighting function is validated by the ablation in \autoref{fig:learnability_entropy}.
For SynthID, we apply the frequentist test described in \autoref{subsec:exp_setup}, computing a depth-weighted Z-score over the tournament layers.

\paragraph{Teacher Trace Quality.}
Pass rates are 48\% for the control, 48.2\% for SynthID, 47\% for TextSeal, and 39.8\% for Gumbel-Max, yielding 2{,}400, 2{,}408, 2{,}352, and 1{,}991 well-formed traces respectively.
Gumbel-Max traces are also notably shorter on average (${\sim}2{,}400$ response tokens vs.\ ${\sim}3{,}300$--$3{,}500$ for other methods), because its deterministic argmax selection causes more repetition loops at $T{=}1.0$; the filter removes these long repetitive traces, leaving only shorter clean ones.
As a result, the student is fine-tuned on different amounts of data across configurations; we do not normalize for this, as the variation in sample count is modest (${\sim}20\%$).

\begin{table}[t!]
\centering
\small
\caption{
    \textbf{Full learnability statistics (OpenR1, $N{=}5{,}000$ prompts).}
    Teacher $-\log_{10}(p)$ is the detection power of the watermark in the teacher's own traces (mean and median across individual traces).
    Student $-\log_{10}(p)$ is the pooled detection power after distillation (original setting, all retained traces).
    \textsuperscript{\dag}TextSeal uses entropy-weighted scoring.
}
\label{tab:learnability_full}
\begin{tabular}{lccccc}
\toprule
\textbf{Method} & \textbf{Retained} & \textbf{Pass} & \multicolumn{2}{c}{\textbf{Teacher $-\log_{10}(p)$}} & \textbf{Student} \\
 & \textbf{Traces} & \textbf{Rate} & Mean & Median & $-\log_{10}(p)$ \\
\midrule
Gumbel-Max   & 1{,}991 & 39.8\% & 14.89 & 9.09  & 24.80 \\
TextSeal     & 2{,}352 & 47.0\% & 33.15\textsuperscript{\dag} & 27.50\textsuperscript{\dag} & 35.82\textsuperscript{\dag} \\
SynthID      & 2{,}408 & 48.2\% & 14.39 & 12.12 & 13.54 \\
Control      & 2{,}400 & 48.0\% & 0.39  & 0.25  & --- \\
\bottomrule
\end{tabular}
\end{table}

\paragraph{Controlled Comparisons.}
The results above are obtained with each method's full set of well-formed traces, which differ in count ($1{,}991$--$2{,}408$) and average length (Gumbel-Max traces average ${\sim}2{,}400$ tokens vs.\ ${\sim}3{,}300$--$3{,}500$ for other methods).
To rule out training data volume as a confound, we repeat the experiment under two controlled conditions (\autoref{fig:learnability}):
(i)~\emph{equal traces}, where each method uses exactly $1{,}991$ traces (the Gumbel-Max minimum, randomly subsampled for the other methods), and
(ii)~\emph{equal tokens}, where each method is allocated ${\sim}15.1$M characters (subsampling traces for methods with more tokens, using all available traces for Gumbel-Max).
Under equal traces, TextSeal achieves the highest student accuracy ($81.0\%$), followed by SynthID and Control ($78.8\%$ each) and Gumbel-Max ($77.7\%$).
Under equal tokens, the spread narrows ($79.7\%$/$78.6\%$/$79.6\%$/$77.6\%$ for TextSeal/Gumbel-Max/SynthID/Control).
In both settings all watermarked students substantially improve over the pre-training baseline ($64.5\%$).
Detection results confirm that all three watermarks remain strongly detectable under both controls, validating that the learnability conclusions of \autoref{fig:learnability} are not artifacts of unequal training data volume.

\paragraph{Entropy Weighting Ablation.}
Each weighting variant in \autoref{fig:learnability_entropy} computes the weighted statistic $S_{\text{combined}} = \sum_{i=1}^{n} w_i^{\text{ent}} \cdot s_i$, where $s_i$ is TextSeal's early-fusion score (\autoref{eq:early_fusion}) and $w_i^{\text{ent}} = f(H_i)$ is a function of the local entropy $H_i$ at position $i$, estimated via a single forward pass of the student model.
The $p$-value is computed via the moment-matched Gamma approximation of \autoref{eq:p_value_gamma_combined}, which accounts for the heterogeneous weights.
Concave normalized-entropy transforms outperform linear/superlinear alternatives because they moderately upweight high-entropy positions---where the watermark has more room to influence token selection (\autoref{prop:expected_score})---without over-amplifying noisy extreme-entropy tokens.
Unnormalized power functions ($H_i^{1.0}$, $H_i^{1.5}$) are sensitive to the absolute entropy scale and perform no better than the unweighted baseline.

\section{Extended Related Work}\label{app:related}

\subsection{Post-Hoc Text Watermarking}

Early text watermarking altered surface-level text characteristics such as characters or spacing~\citep{brassil1995electronic}.
Later methods modify grammatical or syntactical structures via pre-established rules~\citep{topkara2005natural}, including synonym substitution~\citep{topkara2006hiding} and word reordering through passivization or topicalization~\citep{topkara2006words, topkara2006natural, meral2009natural}.
Text steganography follows similar principles~\citep{winstein1998lexical, chapman2001practical, bolshakov2004method, shirali2008new, chang2014practical, xiang2017novel}.
These edit-based systems exhibit low robustness and payload, \eg 1--2 bits per sentence~\citep{wilson2016avoiding}.
Deep learning methods have since been applied, including masked language models for steganography~\citep{ueoka2021frustratingly}, infilling models~\citep{yoo2023robust}, neural lexical substitution~\citep{qiang2023natural}, and encoder-decoders~\citep{abdelnabi2021adversarial, zhang2024remark, xu2024robust}.

\subsection{Generation-Time LLM Watermarking}

The first watermarks for machine-generated text date back to a method presumably used in Google Translate to filter translations from future training data~\citep{venugopal2011watermarking}.
For LLM-generated text, two concurrent approaches appeared shortly after the release of ChatGPT:
\citet{kirchenbauer2023watermark} bias a subset of the vocabulary (``green-red list''), while \citet{aaronson2023watermarking} alter the sampling via the Gumbel-max trick.
Both use pseudorandom seeds generated from a secret key and preceding tokens, enabling lightweight statistical detection without access to the model.

Subsequent work explores improved tests and multi-bit watermarking~\citep{fernandez2023three, yoo2024advancing, qu2024provably}, position-dependent seeds~\citep{christ2023undetectable, kuditipudi2023robust}, low-entropy optimizations~\citep{lee2023wrote, christ2023undetectable, huang2023optimal}, and semantic watermarks for improved robustness~\citep{liu2023semantic, liu2024adaptive, fu2024watermarking, hou2023semstamp, hou2024k}.
A key distinction is whether a method is \emph{distortion-free}: at each generation step, the next-token distribution is preserved, \ie $\Prob(\text{output}_t = v) = p_v^{(t)}$ for all $v$, where the probability is taken over the randomness of the watermark scheme (PRF seeds and, for dual-key methods, the key selection). Each individual token is drawn from the unmodified LLM distribution; only \emph{diversity} across repeated generations for the same prompt is reduced. See \autoref{app:method_overview} for detailed scheme descriptions.
Green-red list methods~\citep{kirchenbauer2023watermark} and low-entropy filtering methods (e.g., SWEET~\citep{lee2023wrote}, which skips watermarking on low-entropy tokens) are \emph{not} distortion-free: they alter the output distribution, degrading every generation.
MorphMark~\citep{wang2025morphmark} adaptively scales the green-red bias based on the natural green-list probability mass, reducing distortion in low-entropy contexts, but remains non-distortion-free since it still applies a logit bias.
Semantic watermarks~\citep{liu2023semantic, liu2024adaptive, hou2023semstamp} require auxiliary semantic encoders at generation time, making them harder to deploy.
Gumbel-max~\citep{aaronson2023watermarking}, Permute-and-Flip~\citep{zhao2024permute}, DiPMark~\citep{wu2023dipmark} (distortion-free green-red via pseudorandom permutations), SynthID-Text~\citep{dathathri2024scalable} (deployed in Google Gemini via tournament-based sampling), and WaterMax~\citep{giboulot2024watermax} (multiple generations per query, impractical for production) are distortion-free.
Toolkits have also been introduced to benchmark these methods~\citep{piet2023mark, pan2024markllm}.
Recent large-scale evaluations~\citep{fernandez2025good} show that Gumbel-max and SynthID achieve the best detectability-quality Pareto frontier among all methods, strictly dominating DiPMark, green-red variants, and semantic watermarks.

TextSeal builds on the Gumbel-max framework but introduces dual-key generation for diversity, entropy-weighted detection, and localized multi-region search---none of which are present in prior work.
We therefore compare TextSeal against these two practical baselines.
Because all three are distortion-free, the comparison is controlled: we fix the LLM, temperature, and top-$p$, and only vary the watermark-specific diversity parameter (key routing probability $\alpha$ for TextSeal, tournament depth for SynthID), isolating the watermark's effect from the decoding strategy.

\subsection{Post-Hoc LLM Watermarks for Data Protection}

Recent works apply LLM watermarks to training or evaluation data via paraphrasing.
Most exploit watermark radioactivity~\citep{sander2024watermarking}, \ie the detectable traces left when watermarked text is used for training.
Applications include detection of texts used in retrieval-augmented generation~\citep{jovanovic2025ward}, benchmark contamination detection~\citep{sander2025detecting}, and training data copyright~\citep{zhang2025leave}.
Waterfall~\citep{lau2024waterfall} evaluates post-hoc watermarking through LLM paraphrasing for provenance on code and natural text.
In \autoref{sec:learnability}, we demonstrate that TextSeal's watermark transfers through distillation, extending this line of work to reasoning trace provenance.

\end{document}